\documentclass[a4paper,USenglish,cleveref, autoref, thm-restate]{lipics-v2021}
\usepackage{tcolorbox}
\usepackage{tikz}
\usepackage{mathtools}
\usepackage{caption}
\usepackage{graphicx}

\graphicspath{{Figures/}}

\newcommand{\LPF}{\mathsf{LPF}}
\newcommand{\NDExt}[3]{\mathsf{DExt}_{#1}(#2,#3)}
\newcommand{\Occurrences}[2]{\mathsf{Occ}_{#1}(#2)}
\newcommand{\LCP}{\mathsf{LCP}}
\newcommand{\IPM}{\mathsf{IPM}}
\newcommand{\Clusters}{\mathsf{Clusters}}

\newcommand{\LCS}{\mathsf{LCS}}
\newcommand{\LCA}{\mathsf{LCA}}
\newcommand{\LPFpos}{\mathsf{LPFpos}}
\newcommand{\Otild}{\tilde{O}}
\newcommand{\per}{\mathsf{per}}
\newcommand{\Insert}{\mathsf{Insert}}
\newcommand{\Link}{\mathsf{Link}}
\newcommand{\Delete}{\mathsf{Delete}}
\newcommand{\MoveInterval}{\mathsf{MoveInterval}}
\newcommand{\findIAncestor}{\mathsf{LevelAncestor}}

\newcommand{\GetDepth}{\mathsf{GetDepth}}
\newcommand{\SelectPhrase}{\mathsf{SelectPhrase}}
\newcommand{\LA}{\mathsf{LevelAncestor}}
\newcommand{\ChargeTwo}{\mathsf{Charge}_2}
\newcommand{\ContainingPhrase}{\mathsf{ContainingPhrase}}
\newcommand{\LZLength}{\mathsf{LZLength}}
\newcommand{\LZss}{\text{LZ77}}
\newcommand{\depth}{\mathsf{depth}}

\newcommand{\ourvec}{\mathbf}
\newcommand{\ACoord}{\mathsf{ACoordinate}}
\newcommand{\BCoord}{\mathsf{BCoordinate}}
\newcommand{\AVect}{\mathsf{AVector}}
\newcommand{\BVect}{\mathsf{BVector}}
\newcommand{\Phrases}{\mathsf{Phrases}}

\newcommand{\Mat}{\mathsf{Matrix}}
\newcommand{\skipo}{\mathsf{Skip\text{-}}0}
\newcommand{\skiphalves}{\mathsf{Skip\text{-}}2\mathsf{\text{-}halves}}
\newcommand{\Dict}{\mathsf{Dictionary}}

\newcommand{\ext}{\mathsf{Extensions}}
\newcommand{\C}{\mathcal C}
\newcommand{\I}{\mathcal I}
\newcommand{\sm}{\setminus}

\newcommand{\OV}{\mathsf{OV}}
\newcommand{\eps}{\varepsilon}

\DeclareMathOperator*{\argmax}{arg\,max}

\newtheorem{fact}[theorem]{Fact}
\newcommand{\T}{\mathcal{T}}

\newcommand{\lpos}{\mathsf{L_z}}
\newcommand{\rpos}{\mathsf{R_z}}
\DeclarePairedDelimiter\ceil{\lceil}{\rceil}
\DeclarePairedDelimiter\floor{\lfloor}{\rfloor}

\newcommand{\para}[1]{\subparagraph*{#1}}

\newenvironment{problem}[1]{
  \refstepcounter{problemctr}
  \begin{tcolorbox}[
    colback=blue!5,
    colframe=blue!75!black,
    fonttitle=\bfseries,
    title=Problem~\theproblemctr: #1,
    boxrule=0.8mm,
    width=\textwidth
  ]
}{
  \end{tcolorbox}
}

\newcounter{problemctr}
\renewcommand{\theproblemctr}{\arabic{problemctr}}

\crefname{problemctr}{Problem}{Problems}

\crefname{claim}{Claim}{Claims}

\crefname{observation}{Observation}{Observations}

\title{The Complexity of Dynamic LZ77 is $\tilde{\Theta}(n^{2/3})$}
\author{Itai Boneh}{University of Wrocław, Poland \and \url{https://sites.google.com/view/itai-boneh}}{itai.boneh@cs.uni.wroc.pl}{https://orcid.org/0009-0007-8895-4069}{Partially supported by Israel Science Foundation grant 810/21. Partially supported by the Polish National Science Centre grant number 2023/51/B/ST6/01505.}

\author{Shay Golan}{Ariel University, Israel \and \url{https://sites.google.com/view/shaygolan}}{golansh1@biu.ac.il}{https://orcid.org/0000-0001-8357-2802}{Partially supported by Israel Science Foundation grant 810/21.}

\author{Matan Kraus}{Bar Ilan Univesity, Israel}{matan3@gmail.com}{https://orcid.org/0000-0002-2989-1113}{Supported by the ISF grant no. 1926/19, by the BSF grant 2018364, and by the ERC grant MPM under the EU's Horizon 2020 Research and Innovation Programme (grant no. 683064).}

\authorrunning{Boneh, Golan and Kraus}

\ccsdesc[500]{Theory of computation~Design and analysis of algorithms}

\Copyright{Jane Open Access and Joan R. Public}

\keywords{Text Algorithms, Lempel-Ziv 77, Dynamic Algorithms}

\category{}

\relatedversion{}
\nolinenumbers

\EventEditors{John Q. Open and Joan R. Access}
\EventNoEds{2}
\EventLongTitle{42nd Conference on Very Important Topics (CVIT 2016)}
\EventShortTitle{CVIT 2016}
\EventAcronym{CVIT}
\EventYear{2016}
\EventDate{December 24--27, 2016}
\EventLocation{Little Whinging, United Kingdom}
\EventLogo{}
\SeriesVolume{42}
\ArticleNo{23}
\hideLIPIcs

\begin{document}

\maketitle
\begin{abstract}
The Lempel-Ziv 77 (LZ77) factorization is a fundamental compression scheme widely used in text processing and data compression.  In this work, we investigate the time complexity of maintaining the LZ77 factorization of a dynamic string.  By establishing matching upper and lower bounds, we fully characterize the complexity of this problem.

We present an algorithm that efficiently maintains the LZ77 factorization of a string $S$ undergoing edit operations, including character substitutions, insertions, and deletions.  Our data structure can be constructed in $\tilde{O}(n)$ time for an initial string of length $n$ and supports updates in $\tilde{O}(n^{2/3})$ time, where $n$ is the current length of $S$.  Additionally, we prove that no algorithm can achieve an update time of $O(n^{2/3-\varepsilon})$ unless the Strong Exponential Time Hypothesis fails.  This lower bound holds even in the restricted setting where only substitutions are allowed and only the length of the LZ77 factorization is maintained.

\end{abstract}

\newpage

\section{Introduction}

Lempel-Ziv 77 ($\LZss$)~\cite{LZ77} is one of the most powerful and widely used compression techniques, forming the backbone of formats like gzip and PNG.
The $\LZss$ scheme partitions a string $S[1..n]$ into substrings, called \textit{phrases}, where each phrase is either the first occurrence of a character or a maximal substring that begins to the left of the current phrase in $S$.
Both types of phrases can be encoded by at most two integers, allowing the compression size of $S$ to be linear in the number of phrases.

In practice, $\LZss$ (and its variants) is one of the most widely used compression algorithms, forming the foundation of 57 out of 210 compressors listed in the Large Text Compression Benchmark \cite{ltcb}.
Its influence spans a diverse range of file formats, such as PNG, PDF, and ZIP, highlighting its versatility and efficiency.
Moreover, $\LZss$ plays a crucial role in modern web infrastructure, being embedded in virtually all contemporary web browsers and servers~\cite{AFFKOS19}.

Throughout the years the $\LZss$ factorization and its variants attracted significant attention both in practice~\cite{LiuNCW16,OzsoyS11,OzsoySC14,ZuH14,BKKP16,KempaP13,KarkkainenKP13,ShunZ13} and in theory~\cite{FischerGGK15,FischerIK15,GagieGKNP14,GagieGP15,GawrychowskiKM23,Hong0B23,KempaK17b,kreft2010navarro,KopplS16,KempaS22,NishimotoIIBT20,CIS08,KKP13,BilleCFG17,BilleEGV18,BelazzouguiP16,EllertFP23,bitlz,Gaw11,GotoB13,GotoB14,GNPlatin18,HanLN22,KarkkainenKP13b,KarkkainenKP14,Kosolobov15,Kosolobov15b,KosolobovVNP20,KleinW05,OhlebuschG11,PolicritiP15,Rytter03,Starikovskaya12,YamamotoIBIT14}.
In particular, in recent years, breakthroughs have been made in computing the $\LZss$ factorization of a text in sub-linear time~\cite{kempa2024lempel,ellert2023sublinear}, as well as in developing quantum algorithms for its computation~\cite{gibney2024near}.

In this work, we investigate the maintenance of the $\LZss$ factorization in dynamic settings.
We consider the settings in which the input string $S$ undergoes insertions, deletions, and substitutions.
Our objective is to maintain a data structure alongside $S$ that allows access to its current $\LZss$ factorization at any point during the sequence of updates.
The interface we provide for accessing the $\LZss$ factorization is exceptionally versatile and well-suited for dynamic applications:
given an index $i$, the data structure can either report the $\LZss$ phrase that contains $i$ in $S$ or return the
$i$th phrase in the factorization of $S$.
All access queries are supported in poly-logarithmic time.
For some previously studied applications (see~\cite{BCR24}), a more limited notion of dynamic $\LZss$ maintenance suffices: after each update to $S$, only the current size of its $\LZss$ factorization is reported.
Our data structure supports this functionality as well.

The problem of dynamic $\LZss$ maintenance has recently been considered by Bannai, Charalampopoulos, and Radoszewski~\cite{BCR24}. In their work, they focus on a limited semi-dynamic settings where only insertions at the end of the string and deletions from the beginning are permitted.
In these settings, \cite{BCR24} provide an algorithm with an amortized running time of $\Otild(\sqrt{n})$\footnote{Throughout the paper, we use the $\Otild$ notation to ignore multiplicative polylogarithmic factors of $n$, i.e., $\Otild(f(*)) = O(f(*) \cdot \log^c n)$ for some constant $c$, where $n$ is the current length of the string.}, which reports the size of $\LZss(S)$ after every update.

Although the semi-dynamic settings studied in \cite{BCR24} are well-motivated, a more natural and general approach to developing dynamic string algorithms is to consider the fully dynamic model, where the string undergoes insertions, deletions, and substitutions.
Recent work on dynamic strings in the fully dynamic model (and in some cases, even stronger models) includes substring equality queries \cite{gawrychowski2018optimal}, Longest Common Factor maintenance \cite{amir2020dynamic,charalampopoulos2020dynamic}, Approximate Pattern Matching \cite{CKW20,CKW22}, maintenance of all periodic substrings \cite{amir2019repetition}, maintenance of Suffix Arrays and Inverted Suffix Arrays \cite{kempa2022dynamic,AB21}, Edit Distance computation \cite{charalampopoulos2020alignment}, bounded Edit Distance computation \cite{GK24,CKW23}, Longest Increasing Subsequence \cite{KS21,GJ21}, and Dynamic Time Warping computation \cite{bringmann2024dynamic}.
Building on the ongoing progress in developing fully dynamic algorithms for classical string problems, and drawing inspiration from \cite{BCR24}, which explores $\LZss$ in dynamic settings, we study the complexity of maintaining the $\LZss$ factorization of a dynamic string.
We reach the (surprising) conclusion that the complexity of this problem is $\tilde\Theta(n^{2/3})$.

Interestingly, almost all string problems listed above with near-linear static running times admit a dynamic algorithm with poly-logarithmic update time.
The only exception is the Longest Increasing Subsequence problem, which has an $O(n\log n)$ static algorithm, but only an $\Otild(n^{2/3})$ dynamic algorithm is known~\cite{KS20} and no nontrivial lower bound is known\footnote{We note that there are polynomial lower bounds for \textbf{other} versions of the problem~\cite{GJ2021}.}.
This places the time complexity of $\LZss$ in a unique position: it is the first string problem with near-linear time complexity in the static settings that has a proven polynomial and sublinear time complexity in the dynamic settings.

\subsection{Our Contribution}
We present a fully dynamic data structure that supports access to the $\LZss$ factorization of a dynamic string $S$ which undergoes insertions, deletions, and substitutions.
We present the following interface for querying the $\LZss$ factorization of a string $S$.
\begin{enumerate}
    \item $\SelectPhrase_S(i)$: given an index $i$, report the $i$th $\LZss$ phrase $S[a..b]$ in the factorization of $S$ represented by $a$ and $b$.
    \item $\ContainingPhrase_S(i)$: given an index $i$, report the $\LZss$ phrase $S[a..b]$ in the factorization of $S$ such that $i\in [a..b]$ represented by $a$ and $b$.
    \item $\LZLength_S(i)$: report the length of the $\LZss$ factorization of $S[1..i]$.
\end{enumerate}
We use the term \textit{$\LZss$ queries} to collectively refer to $\SelectPhrase$, $\ContainingPhrase$, and $\LZLength$.
In this paper we provide a data structure for dynamic $\LZss$, which we formalize as the follows.

\begin{problem}{Fully Dynamic \LZss}\label{prb:ub}
\textbf{Preprocess$(S)$:} preprocess a string $S$.
\medskip
\\
\textbf{Query:} an $\LZss$ query on $S$.
\medskip
\\
\textbf{Updates:} apply a given edit operation on $S$.
\end{problem}

\begin{theorem}\label{thm:lz77}
    There is a data structure solving \cref{prb:ub} with $\Otild(n)$ preprocessing time, $\Otild(1)$ time per query and $\Otild(n^{2/3})$ time per update, where $n$ is the current length of $S$.
\end{theorem}

In \cref{sec:lb}, we provide evidence that our data structure is optimal up to sub-polynomial factors.
In particular, we prove that even if only substitutions are allowed and the data structure only needs to report $|\LZss(S)|$ after each update, achieving this with polynomial preprocessing and an update time of $O(n^{2/3-\eps})$ is impossible, assuming the Strong Exponential Time Hypothesis (SETH)~\cite{IP01,IPZ01}.

\begin{theorem}\label{thm:intro-lb}
For every $c,\eps>0$, there is no data structure that solves \cref{prb:ub} with $O(n^c)$ preprocessing time, $O(n^{2/3-\eps})$ query time and $O(n^{2/3-\eps})$ update time, where $n$ is the current length of $S$, unless SETH is false.
\end{theorem}

Exploiting  the tools we develop in the proofs of \cref{thm:lz77,thm:intro-lb}, we also obtain upper and lower bounds for \cref{prb:ub} as a function of  $|\LZss(S)|$.

\begin{theorem}\label{thm:weak}
        There is a data structure solving \cref{prb:ub} with $\Otild(n)$ preprocessing time, $\Otild(1)$ time per update and $\Otild(|\LZss(S)|)$ time per query, where $n$ is the current length of $S$.

        Moreover, for every $c,\eps>0$, there is no data structure that solves \cref{prb:ub} with $O(n^c)$ preprocessing time, $O(|\LZss(S)|^{1-\eps})$ query time and $O(|\LZss(S)|^{1-\eps})$ update time, unless SETH is false.
\end{theorem}

\subsection{Technical Overview}\label{sec:overview}
We provide a high-level discussion of the algorithm of \cref{thm:lz77} and the lower bound of \cref{thm:intro-lb}.
For this overview, we discuss only the task of reporting $|\LZss(S)|$ after each update, focusing on substitution updates.
Even for this restricted problem, all the machinery of our general algorithm is required (and this is also the setting of the lower bound).

An important concept in many $\LZss$ algorithms (e.g., \cite{FHSX03,CI08,CIS08,CIIKRW09,KKP13,BKKP16,CC19,ellert2023sublinear,BCR24,kempa2024lempel}) is the Longest Previous Factor ($\LPF$).
In a string $S$, the longest previous factor of an index $i$, denoted $\LPF_S(i)$, is the maximum $\ell$ such that $S[i..i+\ell)$ is not the leftmost occurrence of itself in $S$.
Formally,
\[
\LPF_S(i) = \max\{ \LCP_S(i,j) \mid j \in [1..i) \}
\]
where $\LCP_S(i,j)$ is the length of the longest common prefix of $S[i..n]$ and $S[j..n]$.
Notice that $r_i = i + \max\{\LPF_S(i),1\}$ is exactly the starting point of the $(x+1)$th phrase in $\LZss(S)$ if the $x$th phrase starts at index $i$ (and $S[i..r_i)$ is a phrase in the $\LZss$ factorization).
A related useful definition is $\LPFpos_S(i)$, which is the rightmost index $j \in [1..i)$ that is a previous occurrence of $S[i..r_i)$ (hence serving as a witness for $\LPF_S(i)$).
Following the approach of the semi-dynamic $\LZss$ algorithm~\cite{BCR24}, we employ the concept of the $\LPF$-tree.
The $\LPF$-tree of a string $S$, denoted $\T_S$, has the indices of $S$ as nodes, where the parent of $i$ in $\T_S$ is $r_i$.
It is easy to see that the depth of node $1$ in $\T_S$ is exactly $|\LZss(S)|$.
Thus, our task reduces to maintaining $\T_S$ while supporting queries for the depth of a node.
We achieve this in $\Otild(n^{2/3})$ time by carefully analyzing how the values of $r_i$ change when an update is applied to $S$.

We define four relevant types of indices in $S$: \emph{super-light}, \emph{light}, \emph{$L$-heavy}, and \emph{$R$-heavy}.
We show that any index $i$ that does not belong to any of these types necessarily has its $r_i$ unchanged as a result of the update, and therefore retains the same parent in $\T_S$ before and after the update.
We then provide separate algorithms for updating the parents of indices of each type, with a single algorithm handling both $L$-heavy and $R$-heavy indices.

\para{Characterizing indices whose $\LPF$ values change after an update.} In the following paragraphs, we introduce four types of indices that, together, cover all indices that change their $\LPF$ value as a result of the update. We emphasize that, for now, we are not concerned with the task of finding or even computing the new $\LPF$ values of these indices. Instead, we focus solely on proving that these sets indeed cover all possible indices $i$ such that $\LPF_S(i) \neq \LPF_{S'}(i)$.

First, let us assume that $S[z]$ is being substituted, resulting in a new string $S'$.
Let us fix an integer parameter $m$ (to be set later to $n^{1/3}$), and let $i$ be an index in $S$.
We focus on characterizing the set of indices whose $\LPF$ value is reduced as a result of the update, i.e., $\{i \mid \LPF_S(i) > \LPF_{S'}(i)\}$ (the characterization of the indices whose $\LPF$ value increased is symmetric).
The initial observation needed to characterize this set is that the update at $z$ must have disrupted the maximal equality that led to the $\LPF$ of $i$ in $S$.
Formally, let $j = \LPFpos_S(i)$.
The value $\LPF_S(i)$ arises from the equality $S[i..i+\LPF_S(i)) = S[j..j+\LPF_S(i))$. Since this equality cannot persist in $S'$, we must have $z \in [i..i+\LPF_S(i))$ or $z \in [j..j+\LPF_S(i))$.
The types we define correspond to these cases, and also roughly correspond to the distance between $i$ and $z$ (or $j$ and $z$, depending on the type).

Super light indices are simply defined as the indices $[z-m..z]$.
Let us consider a non-super light index $i$, with $z\in [i..i+\LPF_S(i))$.
Notice that since $i < z- m$, the string $S[i..r_i)$ contains $M_L=S[z-m..z)$.
In particular, this means that $S[j..j+\LPF_S(i))$ contains an occurrence of $M_L$.
Importantly, both occurrences of $M_L$ appear both before and after the update (i.e., in $S$ and $S'$).
We use the occurrences of $M_L$ as anchors to identify and update all such indices.
This gives rise to the following definition.
\begin{definition}[L-Heavy index]
    An index $i$ is $L$-Heavy if there is an occurrence of $M_L=S[z-m..z)$ at index $k$, and this occurrence is contained both in $S[i..i+ \LPF_S(i))$ and in $S'[i..i+\LPF_{S'}(i))$.
\end{definition}
\begin{figure}[h!]
        \centering
\includegraphics[width=0.7\linewidth]{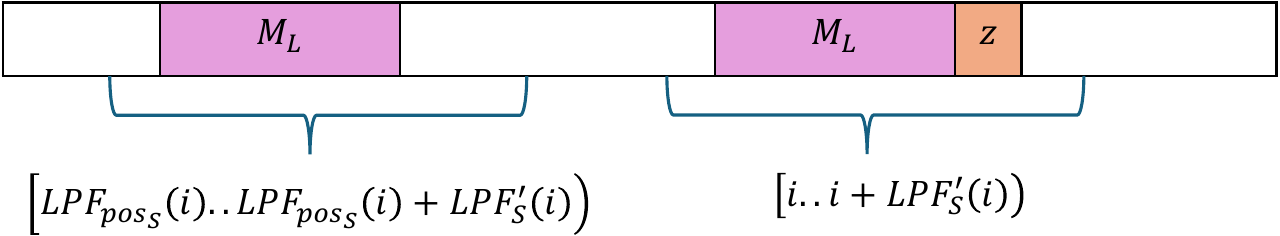}
\caption{Example of an $L$-heavy index. Notice that regardless of the update in $z$, the equality between $S[i..z-1]$ and $S[\LPFpos_S(i)..\LPFpos_S(i) + z-i-1]$ still persists in $S'$.
It follows that $M_L$ is contained in $S'[i..i+\LPF_{S'}(i)]$ as well.}
            \label{fig:case1bintro}
\end{figure}
The above discussion results in a characterization of all indices $i$ with $z\in [i..i+\LPF_S(i))$: all these indices are either super-light or $L$-heavy.

Providing a characterization for index $i$ with $z \in [\LPFpos_S(i)..\LPFpos_S(i)+\LPF_S(i))$ is a bit less direct.
We define light indices as follows.
We say that $i$ is a light index if it is the starting position of the leftmost occurrence, to the right of position $z+1$, of the substring $S[z-a..z+b]$, for some $a, b \in [m]$.
Light indices capture the scenario in which $z\in [\LPFpos_S(i)..\LPFpos_S(i)+\LPF_S(i))$, and $z$ is within a distance of $m$ from both $\LPFpos_S(i)$ and $\LPFpos_S(i) + \LPF_S(i)$.
If this happens, $S[i..i+\LPF_S(i))=S[z-a..z+b]$ for some $a,b\in [m]$.
It should be easy to see that $S[i..i+\LPF_S(i))$ is also the first occurrence of $S[z-a..z+b]$ after $z+1$.

    \begin{figure}[h!]
        \centering
\includegraphics[width=0.7\linewidth]{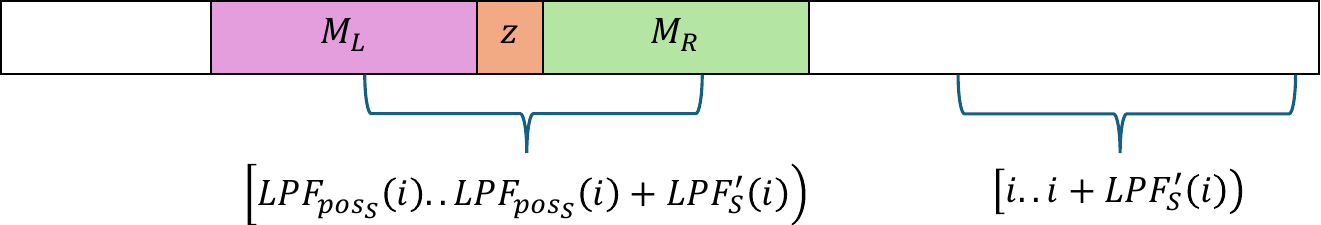}
        \caption{Example of a light index.}
        \label{fig:case2bintro}
    \end{figure}

We are left with the case of indices $z\in [\LPFpos_S(i)..\LPFpos_S(i)+\LPF_S(i))$ and $z$ is more than $m$ indices apart from one of the ends of this interval.
If $z \ge \LPFpos_S(i) + m$, we have that $S[\LPFpos_S(i).. z-1]$ contains $M_L = S[z-m..z-1]$.
Now, if $i$ does not fall within the previous case, i.e., $z\notin[i..i+\LPF_S(i)-1]$, it holds that $S'[\LPFpos_S(i) .. z-1] = S'[i..i +z-\LPFpos_S(i)-1]$.
This equality indicates that $\LPF_{S'}(i) \ge z - \LPFpos_S(i) -1$, and therefore the occurrence of $M_L$ at index $z-m$ is contained both in $S[i..i+\LPF_S(i))$ and in $S'[i..i+\LPF_{S'}(i))$.
It follows that $i$ is $L$-heavy.
\begin{center}

    \includegraphics[width=0.7\linewidth]{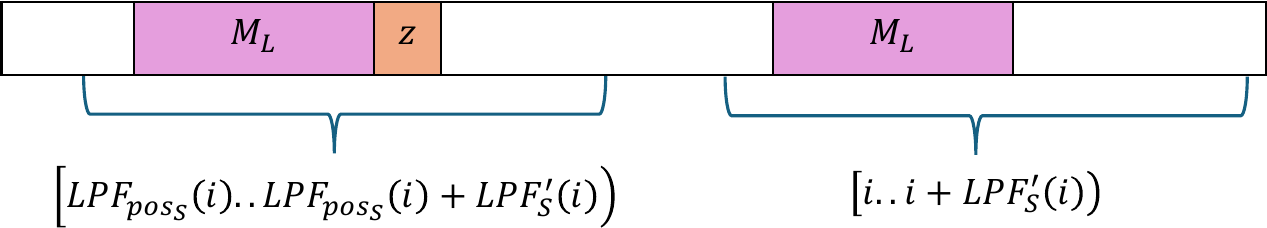}

    \captionof{figure}{An illustration of the case where $z\in [\LPFpos_S(i) .. \LPFpos_S(i) + \LPF_S(i))$, and $i$ is an $L$-heavy index }

\end{center}
Consider the complementary case in which $z< \LPFpos_S(i) + m$.
We can assume that $z < \LPFpos_S(i)+\LPF_S(i)-m$, as we have already covered the light case in which $z = \LPFpos_S(i)+\LPF_S(i)-b$ for some $b\in [m]$.
If the new $\LPF$ value is still large, i.e., $\LPFpos_S(i)+\LPF_{S'}(i) > z+m$, we have that both $S[i..i+\LPF_S(i))$ and $S[i..i+\LPF_{S'}(i))$ contain an occurrence of $M_R = S[z+1..m]$.
We call indices with this property $R$-heavy indices, and they are defined and handled in a similar manner to $L$-heavy indices.

\begin{figure}[h!]
    \centering
    \includegraphics[width=0.7\linewidth]{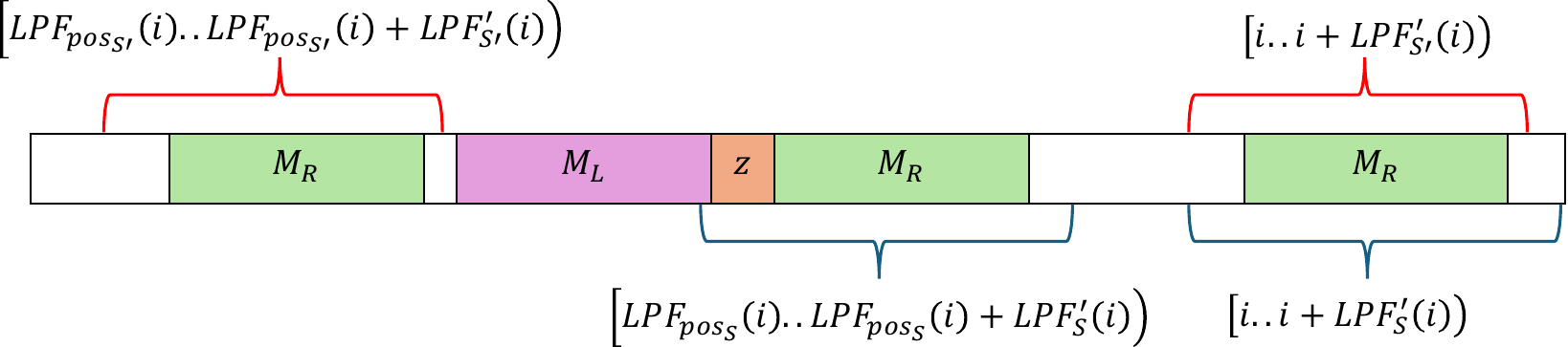}
      \caption{Example of $R$-heavy index. Notice that unlike $L$-heavy index, $\LPFpos$ is not the reason that $S'[i..i+\LPF_{S'}(i))$ contains an occurrence of $M_R$, as the equality between $i$ and $\LPFpos_{S}(i)$ in $S'$ 'breaks' at $z$, right before the start of $M_R$. }
        \label{fig:case2cintro}
\end{figure}

We remain with the case where $\LPF_{S'}(i) =a+b+1$ for some $b < m$ and $a=z-\LPFpos_S(i)$.
In this case, it must hold that $S'[i..i+\LPF_{S'}(i))$ is the leftmost occurrence of $S'[z-a..z+m]$ in $S'$ after $z+1$ (see Case \ref{case:complicated_shay} in the proof of \cref{lem:indextypes}).
We extend the definition of light indices to include these types of indices as well (i.e., first occurrences of $S'[z-a..z+b]$ after $z+1$ in $S'[z+1..n]$).
With that, we have finalized our characterization of the indices that change their $\LPF$ values.
 \begin{figure}[h!]
    \centering

    \includegraphics[width=0.7\linewidth]{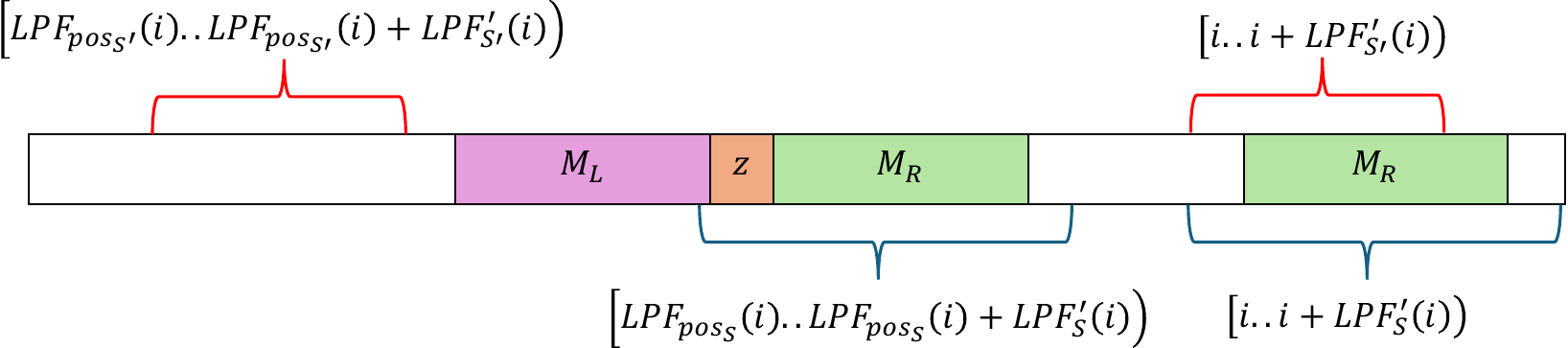}

    \caption{An example of light index of arising from $\LPF_{S'}(i) = a+b+1$.
    Here, since the new $\LPF$ value is small, it must hold that $i$ is the first occurrence after $z$ of $S[z-a..z+m]$.}
    \label{fig:case2c2intro}

\end{figure}

\para{Updating super-light and light indices.}
We proceed to show how to treat each type of indices, and properly update their parents in $\T_S$ efficiently.
We use top trees to represent $\T_S$ throughout the algorithm.
Similarly to \cite{BCR24}, we represent the $\LPF$-tree using dynamic trees in a manner that allows us to query for the depth of a node in $\Otild(1)$ time and apply link-cut updates in $\Otild(1)$ time.
This representation also allows for "interval" updates.
That is, given an interval of integers/nodes $[a..b]$ such that all nodes in $[a..b]$ have the same parent in $\T_S$, and a destination node $v$, set the parent of all nodes in $[a..b]$ to be $v$.
A key tool used in our algorithm is a data structure that given an index $i$, computes $\LPF_S(i)$ in a dynamic string $S$.
We show that there is an algorithm supporting this functionality with update and query time of $\Otild(1)$ (\cref{lem:subPMDS,lem:dynamicLPF})\footnote{Note that a dynamic $\LPF$ data structure immediately implies the upper bound of \cref{thm:weak}. Using such data structure, one can compute the $\LZss$ factorization of $S$ from scratch on query time in $\Otild(|\LZss(S)|)$ time.}.
Given such data structure, treating super-light indices is straight-forward.
The algorithm simply iterates over every index $i\in [z-m..z]$, and computes $\LPF_{S'}(i)$.
Then, the algorithm cuts $i$ from its parent and links it to $i+\LPF_{S'}(i)$.
This takes $\Otild(m)$ time.

To treat light indices, we develop a tailored tool that allows their identification by definition.
Specifically, we provide a data structure that, given substrings $P=S[i_P..j_P]$ and $T=S[i_T..j_T]$ of a dynamic string $S$, specified by their endpoints, returns the leftmost occurrence of $P$ in $T$.
The update and query time of this data structure are both $\Otild(1)$.
Using this tool, we can efficiently identify all light indices.
The algorithm iterates over every pair $(a,b) \in [m]\times [m]$ and finds $i_{a,b}$, the leftmost occurrence of $S[z-a..z+b]$ (and of $S'[z-a..z+b]$) in $S[z+1..n]$.
If $i_{a,b}$ exists, the algorithm proceeds as in the super-light case: it computes $\LPF_S(i_{a,b})$ and sets the parent of $i_{a,b}$ accordingly.
This takes $\Otild(m^2)$ time.

\para{Updating Heavy indices.}
We are left with updating $L$-heavy and $R$-heavy indices, which is the main technical challenge of the algorithm.
We discuss the manner in which we update all $L$-heavy indices, $ R$-heavy indices are treated in a completely symmetric manner.
Let us focus on a single occurrence of $M_L$ at index $i$ that exists both in $S$ and in $S'$, i.e., $S[i..j] = S'[i..j] = M_L$ for $j=i+m-1$.
As a subroutine, our algorithm finds all $L$-heavy indices $h$ such that $S[i..j]$ is a witness to them being $L$-heavy, i.e., $S[h..h+\min(\LPF_S(h),\LPF_{S'}(h)))$ contains $S[i..j]$.
We use the following helpful fact.
\begin{fact}[{cf. \cite[{Lemma 19}]{BCR24}}]\label{fact:techovmonotonic}
    For every $a<b \in [n]$, it holds that $a+\LPF_S(a)\le b + \LPF_S(b)$.
\end{fact}
It follows that there exists some $\ell_{\min}$ such that $h + \LPF_S(h) > j$ if and only if $h\ge \ell_{\min}$.
Therefore, using our data structure for $\LPF$ queries, we can binary search for the threshold value $\ell_{\min}$.
The indices $[\ell_{\min}..i]$ are exactly the set of indices $h$ such that $S[i..j]$ is contained in $S[h..h+\LPF_S(h))$.
We can apply the same binary search in $S'$ to obtain an interval $[\ell'_{\min}..i]$ that contains exactly all indices such that $S[h..h+\LPF_{S'}(h))$ contains $S[i..j]$.
By definition, $[\max(\ell_{\min},\ell'_{\min})..i]$ is exactly the set of $L$-heavy occurrences with $S[i..j]$ being a witness for their $L$-heaviness.
By doing the above for every occurrence of $M_L$, we can identify all $L$-heavy occurrences.

Two challenges emerge when advancing with this approach.
Firstly, $M_L$ may have a lot of occurrences (possibly $\Omega(n)$), and this approach clearly takes at least as much time as the number of $M_L$ occurrences in $S$.
Second, and more critical: unlike in the previous cases, there are possibly a lot ($\Omega(n)$) of $L$-heavy indices (even if there are only $O(1)$ occurrences of $M_L$).

As a warm-up, consider the case in which $M_L$ is aperiodic.
This immediately removes the first issue, as the number of occurrences of an aperiodic string of length $m$ in a string of length $n$ is bounded by $O(n/m)$.
Now, get an interval of indices for every one of the $O(n/m)$ occurrences of $M_L$, and the union of these intervals is exactly the set of all heavy indices.
Let us focus on one interval $[\ell..i]$ of $L$-heavy indices with the witness $M_L = S[i..j]$.
In order to properly set the parents of all of these indices, we need to gather more information regarding the $\LPF$ values of the indices in $[\ell..i]$ both in $S$ and in $S'$.
To do this, we use the notion of \textit{critical indices}.
An index $c\in [\ell..i]$ is critical (in $S'$) if it has $c + \LPF_{S'}(c) > (c-1) + \LPF_{S'}(c-1)$.
If we collect all critical indices $c_1,c_2, \ldots, c_t$ in $[\ell..i]$, we essentially have a partition of $[\ell..i]$ into $t+1$ intervals $I'_1,I'_2,\ldots I'_{t+1}$ which we denote as $\I_i^{S'}$.
These intervals have the property that, within each interval, all indices $h$ share the same value of $h + \LPF_{S'}(h)$; that is, they have the same parent in $\T_{S'}$.
The critical indices can be found in $\Otild(t)$ time as follows.
Let $x_k = c_k-1 + \LPF_{S'}(c_k-1)$ (initially, $x_0 = i + \LPF_{S'}(i)$).
The next critical index $c_{k+1}$ can be found by looking for the minimal $h\in [\ell .. i]$ such that $h + \LPF_{S'}(h) = x_k$.
This can be done using a binary search, employing the $\LPF$ data structure at every step.
Now, the next $x_{k+1}$ can be found using another $\LPF$ query.
This way, the algorithm finds $c_1,c_2,\ldots,c_t$ while spending $\Otild(1)$ time per critical index, for a total of $\Otild(t)$ running time.

The next natural step is to apply interval updates to $\T_S$.
Namely, for each interval $I'_k = [a..b]$ implied by the critical indices, set the parent of all nodes in $[a..b]$ to $a+\LPF_{S'}(a)$.
Notice that this update may be invalid: the interval update has a precondition, requiring that all vertices in the moved interval $[a..b]$ have the same parent in $\T_S$.
To comply with this condition, we also find critical indices in $S$ in the same manner, and obtain another partition $I_1,I_2,\ldots$ which we denote as $\I^S_i$.
Combining the partitions $\I^{S'}_i$ and $\I^S_i$ we obtain a single refined partition $\I_i$ such that for every interval $[a..b]\in\I_i$ in the refined partition, all indices $h$ have both $h+\LPF_S(h) = a +\LPF_S(a)$ and $h+\LPF_{S'}(h) =a + \LPF_{S'}(a)$.
Therefore, we can apply a move interval operation on every interval in $\I_i$ and it would be both correct (due to all indices in the interval having the same parent in $\T_{S'}$) and valid (due to all indices in the interval having the same parent in $\T_S$).
With this, we have updated all heavy occurrences that correspond to the occurrence of $M_L$ in $S[i..j]$ in $\Otild(|\I_i|)$.
With this, the update of the parents of all $L$-heavy indices is concluded.

\para{Bounding the running time for the heavy indices.}
We are left with the task of analyzing the running time of the algorithm.
It is relatively easy to show that $|\I_i| \in O(n/m)$.
On a high level, each critical index $c_x$  implies some equality $S[c_x..c_x + \LPF_S(c_x)) = S[\LPFpos_S(c_x) .. \LPFpos_S(c_x) + \LPF_S(c_x))$.
In particular, since $S[c_x ..c_x+\LPF_S(c_x))$ contains $S[i..j]=M_L$, there is some occurrence of $M_L$ aligned with $S[i..j]$ in $S[\LPFpos_S(c_x)..\LPFpos_S(c_x) + \LPF_S(c_x))$.
Denote this occurrence as $o_x$, we claim that there is no pair of two critical indices $c_x,c_y$ that have $o_x=o_y=o$.
Intuitively, if both $c_x$ and $c_y$ align $S[i..j]$ with $S[o..o+m-1]$, the length of the remainder of their right agreements with their corresponding $\LPFpos$ past $j$ should be the same, as it depends solely on $\LCP_S(i,o)$.
This is a contradiction to the fact that $c_x + \LPF_S(c_x) \neq c_y + \LPF_S(c_y)$ for every two critical indices.

Since there are $O(n/m)$ occurrences of $M_L$ (recall that we consider now the case where $M_L$ is aperiodic), and each critical index in $\I^S_i$ corresponds to a unique occurrence, we have that $|\I^S_i| \in O(n/m)$.
Due to the same arguments, we also have $|\I^{S'}_i| \in O(n/m)$ and therefore $|\I_i|\in O(n/m)$.
The running time of the algorithm for updating $L$-heavy indices is $\sum_{S[i..j]=S'[i..j]=M_L}\I_i$.
According to the bound above this sums up to $O(\sum_{S[i..j]=S'[i..j]=M_L} O(n/m)) = O((n/m)^2)$.
It follows from this analysis that
the total running time for updating the parents of all indices is $\Otild((n/m)^2+m^2+m) = O((n/m)^2+m^2)$.
Sadly, this expression is balanced when $m=\sqrt{n}$ and yields an $\Otild(n)$ running time.
This is not an improvement over simply computing $\LZss(S)$ from scratch after every update.
Our final running time follows from a tighter, and rather surprising, bound on the sum of all critical indices (see \cref{sec:chain}).
That is, we show that $\sum_{S[i..j]}|\I_i| \in O(|\Occurrences{S}{M_L}|\log^2n)= O((n/m) \log^2 n)$ (where $\Occurrences{S}{M_L}$ is the set of occurrences of $M_L$ in $S$).
The proof of this bound relies on a heavy-path decomposition of suffix trees and carefully crafted charging arguments.
This immediately improves the running time algorithm to $\Otild(n/m+m^2)$, which yields the desired $\Otild(n^{2/3})$ running time by selecting $m = n^{1/3}$.

\para{The periodic case.}
In the case where $M_L$ is periodic, the above analysis is not sufficient, as $|\Occurrences{S}{M_L}|$ can not be bounded by $O(n/m)$.
Via a novel analysis of the behavior of $\LPF$ and $\LPFpos$, we manage to show that it is always sufficient to run the above algorithm on a subset $A$ of $O(n/m)$ occurrences of $M_L$, and it is still guaranteed that all $L$-heavy indices are properly updated.
These occurrences of $A$ have the property that $\sum_{i \in A} |\I_i| \in O(n/m)$, which allows us to recreate the running time of $\Otild(n^{2/3})$ in the general case.

\subsubsection*{Lower Bound}
We provide an evidence for the optimality of our algorithm by showing a reduction from the Orthogonal Vectors problem to dynamic $\LZss$ maintenance.

\para{Orthogonal vectors.}
We consider the following variant of the Orthogonal Vectors problem.
Initially, a set $A$ of $d$-dimensional Boolean vectors (for some $d$ which is sub-polynomial in $n$) is given for preprocessing, and then a set $B$ is queried.
The goal is to decide if there are vectors $\ourvec{v}\in A$ and $\ourvec{u}\in B$ such that $\ourvec{v}\cdot \ourvec{u} = 0$, where $\ourvec{v} \cdot \ourvec{u}= \bigvee_{j=1}^d(\ourvec{v}[j] \wedge \ourvec{u}[j]) $ denotes the boolean inner product of $\ourvec{u}$ and $\ourvec{v}$.
It was proved by \cite{AVW21} that no algorithm can solve this problem with polynomial preprocessing time and $O(n^{2-\varepsilon})$ query time, where $n = |A|=|B|$ and $\varepsilon > 0$, unless SETH is false.

\para{Reduction from orthogonality to $|\LZss(S)|$.}
In this overview, we draw a connection between the size of $\LZss(S)$ and vector orthogonality.
The complete reduction appears in \cref{sec:lb}.
Consider a string $S$ that consists of three parts, $S=D \#  G_B(\ourvec{u})  \# G_A(\ourvec{v})$ where every $\#$ represents a fresh symbol.
We construct the components of $S$ using strings over alphabet that contains $\{0_i,1_i \mid i\in [d+1]\}$.
Additionally, the alphabet contains unique symbols, each occurring at most once in $S$ and represented as $\#$.
For a vector $\ourvec{v}\in A$, we define
\[G_A(\ourvec{v}) = \bigodot^d_{i=1}\begin{cases}
    0_i0_i& \ourvec{v}[i]=0\\
    1_i1_i& \ourvec{v}[i]=1
\end{cases}\,,\]
i.e., each coordinate $i$ is encoded by two consecutive occurrences of $\ourvec{v}[i]$ using the alphabet letter dedicated to this coordinate.
The first part of $S$, called the \emph{dictionary}, consists of \[D=\bigodot^d_{i=1}0_i0_i\,\,\#\,\,0_{i}0_{i+1}\,\,\#\,\,0_i1_{i+1}\,\,\#\,\,1_{i}0_{i+1}\,\,\#\,\,1_i1_{i+1}\,\,\#.\]
In words, $D$ contains a set of two $0$s for every possible coordinate, and for every $(x,y)\in \{ 0,1\}^2$, and coordinate number $i\in [d]$, $x_i$ followed  $y_{i+1}$.
Notice that these sequences naturally occur in $G_A(\ourvec{v})$ on transitions between coordinates.

\para{Intuition for the dictionary.}
Let us provide some intuition for the definition of $D$.
The length we are about to analyze is the number of $\LZss$-phrases in $\LZss(S)$ that cover $G_A(\ourvec{v})$, denoted as $L_{\ourvec{v}}$.
One functionality of $D$ is that it forces $L_{\ourvec{v}}$ to be at most $d+1$.
In particular, if $S=D \# G_A(\ourvec{v})$ the phrases covering $G_A(\ourvec{v})$ have the following structure.
Let $i' \in [d]$ be the first coordinate such that $\ourvec{v}[i'] = 1$.
The first $i'-1$ phrases would be $0_10_1, 0_20_2, \ldots , 0_{i'-1}0_{i'-1}$ (each is a substring occurring in $D$).
Once the $i'$ coordinate is met, there is no substring $1_{i'}1_{i'}$ anywhere in $D$ or elsewhere in $G_A(\ourvec{v})$, so the next phrase would have to be $1_{i'}$ from $D$ (as a partial word of $1_{i'}0_{i'+1}$, for instance).
This leaves one trailing $1_i$ symbol to be covered by the next phrase.
From this point on, there will always be one trailing symbol that is not covered by the previous phrase.
Specifically, the rest of the phrases would be of the form $\ourvec{v}[i]_i\ourvec{v}[i+1]_{i+1}$, covering the trailing letter of the $i$th coordinate from the previous phrase and failing to cover the next $(i+1)$th coordinate completely within the same phrase.
It can be easily verified that no phrase longer than the ones discussed here are available at any point throughout the $\LZss$ scan of $G_A(\ourvec{v})$.
This leads to $d+1$ phrases for a non-zero vector $\ourvec v$.
Specifically, there are $i'-1$ phrases before the $i'$th coordinate, followed by the phrase $1_{i'}$, then, $d-i'$ phrases that cover the trailing symbol of the previous phrase and the first symbol of the next coordinate, and finally a last phrase $\ourvec{v}[d]_d$ to cover the trailing symbol of the last coordinate.
We have shown that the $L_{\ourvec{v}} = d+1$ if $S=D \# G_A(\ourvec{v})$.
Notice that $D$ also allows the $\LZss$ factorization to 'skip' over $0$ coordinates in one step, and that the factorization of $G_{A}(\ourvec{v})$ transitions is forced to use phrases with trailing symbols only once it encounters a $1$ coordinate.
Clearly, introducing $G_B(\ourvec{u})$ between $D$ and $G_A(\ourvec{v})$ (regardless of its exact content) can only decrease the number of phrases covering $G_A(\ourvec{v})$, as $L_{\ourvec{v}}$ could use substrings of $G_B(\ourvec{u})$ in addition to substrings of $D$ (and itself).

\para{$B$'s gadget.}
We define $G_B(\ourvec{u})$ as follows,
\[G_B(\ourvec{u}) = \bigodot^d_{i=1}\begin{cases}
    0_i0_i\#& \ourvec{u}[i]=1\\
    1_i1_i\#& \ourvec{u}[i]=0
\end{cases}\]
In particular, $G(\ourvec{u})$ contains the substring $1_i1_i$ if and only if $\ourvec{u}[i]=0$.
Now, we claim that in $S=D \# G_B(\ourvec{u}) \# G_A(\ourvec{v})$, the value of $L_{\ourvec{v}}$ is $d$ if  $\ourvec{v} \cdot \ourvec{u}=0$ and $L_{\ourvec{v}} = d+1$ otherwise.
It should be easy to see that if $\ourvec{v} \cdot \ourvec{u}=0$, then $L_{\ourvec{v}} = d$.
The $0_i0_i$ sub-words in $D$ allows skipping over a coordinate $i$ in $G_A(\ourvec{v})$ with $\ourvec{v}[i]=0$ with a single phrase, and for every $i$ such that $\ourvec{v}[i]=1$, orthogonality implies that $\ourvec{u}[i] = 0$, and therefore the substring $1_i1_i$ is in $G_B(\ourvec{u})$, which again allows skipping over a coordinate gadget in one step for a total of $d$ steps to cover all of $G_A(\ourvec{v})$.
For the other direction, if there is some first index $i'$ such that $\ourvec{u}[i']=\ourvec{v}[i']=1$, the phrases preceding the $i'$th coordinate will be able to skip a coordinate each, but when the $i'$th coordinate would be reached, there is no occurrence of $1_{i'}1_{i'}$ to copy, and the remaining phrases would have to carry a trailing symbol as discussed before.

\para{Construction for the set $A$.}
Let us describe a natural, yet \textbf{faulty} construction, that builds upon the previous discussion.
Given sets $A$ and $B$ of vectors, build $S = D \# G_B(\ourvec{1}) \# \bigodot_{\ourvec{v}\in A}G_A(\ourvec{v})$.
Initialize dynamic data structures for reporting $|\LZss(S)|$ and $|\LZss(D \# G_B(\ourvec{1})\#)|$.
Notice that $\sum_{\ourvec{v}\in A}L_{\ourvec{v}}$ is exactly the difference between the length of the $\LZss$ factorization of these two strings.
(We highlight that this is an abuse of notation, as $L_{\ourvec{v}}$ is the number of phrases in $\LZss(S)$ covering $G_A(\ourvec{v})$, and here $L_{\ourvec{v}}$ is in a wider context where multiple $G_A(\cdot)$ gadgets are present in $S$.)
Now, for every vector $\ourvec{u}\in B$, modify $G_B(\ourvec{1})$ to be $G_B(\ourvec{u})$ using $O(d)$ substitutions.
It follows from our previous analysis that if $\ourvec{u}$ is not orthogonal with any vector of $A$, $\sum_{\ourvec{v}\in A}L_{\ourvec{v}} = |A|(d+1)$, and otherwise the sum would be strictly smaller.
Let $|A| = |B| = n$, and to avoid clutter, let us assume that $d \in \Otild(1)$.
The above discussion concludes in an algorithm that reports if there is a pair of orthogonal vectors in $A$ and $B$ using a dynamic $\LZss$ data structure on a string of length $\Otild(n)$ by applying $\Otild(n)$ updates.
It follows from \cite{AVW21}  that assuming SETH, the running time of each update can not be $O(n^{1-\varepsilon})$ for any $\varepsilon > 0$.

It should be clear that the above reasoning is problematic, as it contradicts \cref{thm:lz77}.
The main problem is that the intricate interaction between $G_B(\ourvec{u})$ and $G_A(\ourvec{v})$ is overshadowed by the interactions between the many copies of $G_A(\ourvec{v})$ included in this reduction.
In particular, if $A$ contains a vector with $\ourvec{v}[i] = 1$ for some coordinate $i$, all vectors $\ourvec{w}$ that have their gadget $G(\ourvec{w})$ appear to the right of $G(\ourvec{v})$ in $S$ can 'skip' the $i$th coordinate by copying from $\ourvec{v}$, regardless of $G_B(\ourvec{u})$.
To address this issue, we need to make sure that any gadget $G_A(\ourvec{v})$ can not use the gadgets preceding it in $S$ other than $G_B(\ourvec{u})$ and $D$ in any meaningful way.
To achieve this, on a high level, we have to introduce a larger version of $G_A(\ourvec{v})$ on each subsequent concatenation, which results in $|S| \in \Otild(n^{1.5})$.
Now, \cite{AVW21} suggests that it is impossible that all of the $\Otild(n)$ updates have a running time of $O(n^{1-1.5\varepsilon}) = O(|S|^{2/3-\varepsilon})$, which proves \cref{thm:intro-lb}.

\subsection{Open Problems}
We note that the data structures of \cref{thm:lz77,thm:weak} use $\tilde{O}(n)$ words of space.
A major open problem for future research is whether one can design a data structure with the same runtime guarantees as in \cref{thm:lz77}, or at least with sublinear update and query times, while using only $\tilde{O}(|\LZss(S)|)$ space.
The main obstacle to improving the space usage lies in our reliance on the $\LPF$-tree (see below), which inherently requires $\Theta(n)$ space.
Any positive resolution of this problem would therefore require the development of an alternative set of tools.

While we focused on developing a data structure with polylogarithmic query time, it may be interesting to explore different trade-offs between query and update times for a dynamic Lempel-Ziv data structure.
In particular, one may consider the opposite end of the query–update trade-off spectrum, aiming to design an algorithm with polylogarithmic update time.
While \cref{thm:weak} states that $\Otild(1)$ update time and $\Otild(|\LZss(S)|)$ is possible, it remains an open problem to develop a data structure with $\Otild(1)$ update time and guaranteed $O(n^{1-\eps})$ query time (even for highly non-compressible strings).

Another interesting direction is to study other compression schemes in dynamic settings. How efficiently can other variants of Lempel-Ziv (LZ78~\cite{LZ78}, LZ-END~\cite{kreft2010navarro}, LZ-SE \cite{SNYI25}, or LZ without overlaps~\cite{KKNPS17}) be maintained dynamically?
Similar questions can be asked about entirely different compressibility measures, such as the run-length encoding of the Burrows-Wheeler Transform~\cite{BW94} and substring complexity~\cite{RBKNP21}.

\renewcommand{\LPF}{\mathsf{LPF}'}
\section{Preliminaries}
\para{Integers notations.}
We use range notation to denote consecutive ranges of integers.
For integers $i,j$ denote $[i..j]=\{i,i+1,\ldots, j\}$ and $[i..j)=[i..j-1]$ (if $j<i$, $[i..j]=\emptyset$).
Denote $[i]=[1..i]$.
We sometimes deal with arithmetic progressions of integers.
We represent a set $A= \{a,a+d,a+2d,\ldots ,b=a+d\cdot (|A|-1) \}$ by the triplet $(a,b,d)$.
We call $d$ the \emph{difference} of the arithmetic progression $(a,b,d)$.

\para{Strings.}
A string $S$ of length $|S|=n$ over an alphabet $\Sigma$ is a sequence of symbols, $S=S[1]S[2]S[3]\ldots S[n]\in\Sigma^n$.
We denote a substring $S[i]S[i+1]\ldots S[j]$ by $S[i..j]$.
We say that $S[i..j]$ is a substring that starts at index $i$ and ends at index $j$.
We call a substring that starts at index $1$ a prefix of $S$, and a substring that ends at index $n$ a suffix of $S$.
We denote the reversed string of $S$ as $S^R=S[n]S[n-1]\ldots S[1]$.
Let $\bigodot$ denote
concatenation (in increasing order of the running index).

For a pattern $P$, we denote $\Occurrences{S}{P}=\{i\mid S[i..i+|P|-1]=P\}$.
We call an index $i\in \Occurrences{S}{P}$ an occurrence of $P$ in $S$, and we say that $P$ occurs in $S$ if $\Occurrences{S}{P} \neq \emptyset$.

For two strings $S$ and $T$, the length of longest common prefix of $S$ and $T$ is denoted by $\LCP(S,T)=\max \{\ell \mid S[1..\ell] = T[1..\ell]\}$.
The length of the longest common suffix of $S$ and $T$ is denoted by $\LCS(S,T)=\LCP(S^R,T^R)$.

For a string $S$ of length $n$ and two indices $i,j \in [n]$, we denote $\LCP_S(i,j) = \LCP(S[i..n],S[j..n])$ and $\LCS_S(i,j) = \LCS(S[1..i],S[1..j])$.

We say that an integer $p\ge 1$ is a period of $S$ if $S[i]=S[i+p]$ for every $i\in [1..n-p]$.
We say that $p$ is the period of $S$, denoted as $\per(S)$ if it is the minimal period of $S$.
We say that $S$ is periodic if $\per(S) \le \frac{n}{2}$.
Otherwise, we say that $S$ is aperiodic.
A run in a string $S$ is a maximal periodic substring.
Formally, $R= S[i..j]$ is a run if $R$ is periodic with period $p$ and $S[i-1] \ne S[i-1+p]$, and $S[j+1] \ne S[j+1-p]$ (if one of the indices is undefined, the inequality is considered satisfied).
We use the following fact regarding runs.

\begin{fact}[{\cite[{Lemma 8}]{BGS24}}]\label{fact:index_coverd_two_pruns}
Let $S$ be a string, and $p\in[|S|]$. Every index $i$ in $S$ intersects with at most two runs with period $p$.
\end{fact}

The following is a well known fact.
For completeness, we prove it in \cref{sec:missing_proofs}.
\begin{fact}\label{fact:break_period}
    Let $S[1..n]$ be a periodic string with period $p$.
    Then for $x\ne S[n-p+1]$ the string $Sx$ is aperiodic.
    Symmetrically, for $x \ne S[p]$, the string $xS$ is aperiodic.
\end{fact}

The following is a well-known fact regarding strings and periodicity.
It follows directly from the periodicity lemma (\cite{FiWi65}).
\begin{fact}\label{fact:aperfarapart}
    Let $T[1..n]$ and $P[1..m]$ be two strings, let $p=\per(P)$, and let $i<j$ be two consecutive occurrences of $P$ in $T$.
    It must hold that either $j-i=p$, or $j-i \ge \frac{m}{2}$.
    In particular, if $P$ is aperiodic then $j-i \ge \frac{m}{2}$.
\end{fact}

In particular, \cref{fact:aperfarapart} suggests the following structure of occurrences of a string.
\begin{fact}\label{fact:per_structure}
    For a string $P[1..m]$ with period $p$ and a string $T[1..n]$, the occurrences of $P$ in $T$ can be partitioned into $O(n/m)$ non-overlapping arithmetic progressions.
    Each arithmetic progression has a difference $p$, and the occurrences implied by the progression are exactly the occurrences of $P$ contained in the same run of $S$ with period $p$.
    In particular, if $(a,b,p)$ is an arithmetic progression then $a-p,b+p\notin\Occurrences{T}{P}$.
\end{fact}

For a string $P$, we denote the set of arithmetic progressions representing the occurrences of $P$ in a string $T$ that arises from \cref{fact:per_structure} as $\Clusters_T(P)$, and every arithmetic progression that contains all the occurrences of $P$ in some run of $T$ with difference $p$ is called a cluster.

\para{Trees.}
We assume basic familiarity with trees and tree-based data structures.
For a node $v$ in a rooted tree $T$, we denote by $\depth_T(v)$ the number of edges in the unique simple path from $v$ to the root of $T$.
Given a node $v$ and an integer $k$, the $k$-th ancestor of $v$, denoted by $\LA_T(v, k)$, is the ancestor of $v$ whose depth is exactly $\depth_T(v) - k$.
For two nodes $u,v\in T$, we denote by $\LCA_T(u,v)$ the lowest common ancestor of $u$ and $v$ in $T$.

\para{Tries.}
An edge-labeled tree (over an alphabet $\Sigma$) is a tree with a root node $r$ such that every edge is assigned a label $\sigma\in \Sigma$.
Each node $v$ in an edge-labeled tree is associated with a string $L(v)$, also called the label of $v$.
The root $r$ is associated with the empty string $L(r) = \varepsilon$, and every non-root node $v$ has $L(v) = L(u) \cdot \sigma$ such that $u$ is the parent of $v$ and $\sigma$ is the label of the edge $(u,v)$.
Notice that for every two nodes $u,v$ in an edged-labeled tree with $\LCA(u,v)=z$, it holds that $L(z) = L(u)[1..\LCP(L(u),L(v)]]$.

A \emph{trie} is an edge-labeled tree that is derived from a set of strings as follows.
We define the trie of a set $\mathcal{S}$ of strings recursively.
The trie of an empty set $\mathcal{S}= \emptyset$ is a single root vertex.
Let $T'$ be the trie of a set $\mathcal{S}$ of $t\ge 0$ strings.
For a string $S$, let $A= \argmax_{A \in \mathcal S }\LCP(S,A)$.
Let $v$ be the node in $T'$ with $L(v) = S[1..\LCP(S,A)]$.
The trie of $\mathcal{S} \cup \{ S\}$ is obtained by adding a path from $v$ of length $|S| - \LCP(S,A)$ with the labels $S[\LCP(S,A)+1],S[\LCP(S,A)+2],\dots,S[|S|]$ on the edges.
Notice that the label of the node at the end of this path is $S$.
Also notice that if $\LCP(S,A) = |S|$, nothing is added to $T'$.
In words, the trie $T$ of $\mathcal S$ is the minimal edge labeled tree such that for every string $S$ in $\mathcal S$ there is a vertex $v\in T$ with $L(v) = S$.

For two strings $S,S'\in\mathcal{S}$
let $v$ and $v'$ be the nodes in $T$ such that $L(v)=S$ and $L(v')=S'$.
It holds that $\LCP(S,S')=\depth_T(\LCA_T(v,v'))$.

\para{$\LZss$.} The Lempel-Ziv \cite{LZ77} factorization is obtained by the following algorithm.
Given a string $S$, output a sequence of  \emph{phrases}, denoted $\LZss(S)$.
The algorithm starts by initializing a left to right scan with $i=1$ and an empty sequence $P$ of phrases.

At every iteration, if $i = n+1$, the algorithm halts and returns the sequence $P$.
Otherwise, the algorithm adds a new phrase as follows.
If $S[i]$ is a new character, i.e., $\Occurrences{S[1..i-1]}{S[i]}= \emptyset$, add $S[i]$ as a phrase to $P$ and continue with $i \leftarrow i+1$.
Otherwise, let $j$ be an index in $[i-1]$ that maximizes $\LCP_S(i,j)$.
The algorithm appends $(j,\LCP_S(i,j))$ to $P$ as a new phrase and continues with $i \leftarrow i+\LCP_S(i,j)$.

We denote the final sequence of phrases obtained from applying the above procedure on $S$ as $\LZss(S)$.

\subsection{Dynamic Strings}
\para{Edit operations.}
In this paper, we develop and analyze algorithms over a dynamic string $S$ over alphabet $\Sigma$ that undergoes \textit{edit operations}.
An edit operation on $S[1..n]$ can be one of three options listed below, each resulting in a new string $S'$.
\begin{enumerate}
    \item \textit{Substitution:} Represented as a pair $(i,\sigma)\in [n]\times \Sigma$.
    Results in $S'=S[1..i-1] \cdot \sigma \cdot S[i+1..n]$.
    \item \textit{Deletion:} Represented as an index $i\in [n]$.
    Results in $S'=S[1..i-1]  \cdot S[i+1..n]$.
    \item \textit{Insertion:} Represented as a pair $(i,\sigma)\in [n+1]\times \Sigma$.
    Results in $S'=S[1..i-1] \cdot \sigma \cdot S[i..n]$.
\end{enumerate}
A dynamic algorithm receives as inputs a sequence of operations, each represented by the type of the edit operation (substitution, deletion, or insertion) and an index or a pair representing the operation.
When describing our algorithms, we abstain from providing a direct implementation for the substitution update, as it can be implemented using a pair of deletion and insertion.

\para{Dynamic indices.}
Even though edit operations are very local, a single operation may 'shift' a large number of indices.
For instance, when applying a deletion at index $i$, every index $j\ge i$ in $S$ becomes the index $j-1$ in $S'$.
This phenomenon is problematic for any data structure that stores indices of the text, as every index $j>i$ stored implicitly in the data structure should be modified.
These 'shifts' also introduce clutter when making statements about strings with dynamic indices.

Therefore, we represent the indices of $S$ not as explicit numbers but as nodes in a dynamic tree.
Specifically, we store a balanced search tree $T_I$ with $n$ nodes, such that every node $v\in T_I$ stores the size of the sub-tree rooted at $v$ as auxiliary information.
The $i$th node in the in-order traversal of $T_I$ \textit{is} the $i$th index of $S$.
Note that due to the auxiliary information, one can find the $i$th node in $O(\log n)$ time given $i$, and also, given a node, one can find the index $i$ corresponding to it in $O(\log n)$ time.
When an index is deleted (resp.\ inserted) from $S$, we delete (resp.\ insert) the corresponding index (node) from $T_I$, which automatically shifts all other indices.

With this framework, when discussing an index $i$ of a dynamic string, we actually refer to the node in $T_I$ representing $i$ rather than the actual numeric value $i$.
This introduces an additional multiplicative factor of $O(\log n)$ whenever our algorithm accesses an index of the dynamic text.
To avoid clutter, we ignore this dynamic index implementation when describing algorithms and instead assume that indices are stored explicitly.

\subsection{Stringology Tools}

First, we make use of a data structure with the following functionality.

\begin{lemma}[{\cite[{Section 8}]{kempa2022dynamic}}]\label{lem:dympillar}
There is a data structure that maintains a dynamic string $S$ that undergoes edit operations and supports the following queries:
\begin{enumerate}
    \item $\LCP$ query: Given two indices $i,j$, report $\LCP_S(i,j)$.
    \item $\LCS$ query: Given two indices $i,j$, report $\LCS_S(i,j)$.
    \item $\IPM_S(P,T)$ query: Given a pattern $P=S[i_P..j_P]$ and a text $T=S[i_T..j_T]$, both substrings of $S$ represented using their endpoints such that $|T| \le 2|P|$, return $\Occurrences{T}{P}$ represented as at most two arithmetic progressions with difference $ \per(P)$.
\end{enumerate}
All updates and queries take $\Otild(1)$ time.
\end{lemma}

It is well known that the following can be obtained from \cref{lem:dympillar}.
For the sake of cmpletness, we provide the proof in \cref{sec:missing_proofs}
\begin{lemma}\label{lem:ipmalllengths}
There is a data structure that maintains a dynamic string $S$ that undergoes edit operations and supports the following query:

Given a pattern $P=S[i_P..j_P]$ and a text $T=S[i_T..j_T]$, both substrings of $S$ represented using their endpoints, return $\Occurrences{T}{P}$ represented by $\Clusters_T(P)$.
The query time is $\Otild({|T|}/{|P|})$ and the update time of the data structure is $\Otild(1)$.
\end{lemma}

\para{Substrings pattern matching.}
In \cref{sec:DynSPM}, we develop the following tool, which is crucial for our algorithm and may be of independent interest.
\begin{restatable}{lemma}{subPMDS}
\label{lem:subPMDS}
    There is a data structure that maintains a dynamic string that undergoes edit operations and can answer the following query:
    Given four indices $i_T,j_T,i_P,j_P$ representing two substrings $T=S[i_T..j_T]$ and $P = S[i_P .. j_P]$, return TRUE if there is an occurrence of $P$ in $T$, and FALSE otherwise.
    The query and update of the data structure take $\Otild(1)$ time.
\end{restatable}

We further exploit \cref{lem:subPMDS} to obtain the following data structure.
\begin{lemma}\label{lem:first_occ}
    There is a data structure that maintains a dynamic string $S$ that undergoes edit operations and given four indices $i_T$, $j_T$, $i_P$, and $j_P$, can answer the following  queries:
    \begin{enumerate}
        \item Find $\min(\Occurrences{S[i_T..j_T]}{S[i_P..j_P]})$.
        \item Find $\max(\Occurrences{S[i_T..j_T]}{S[i_P..j_P]})$.
    \end{enumerate}
    Each update and query takes $\Otild(1)$ time.
\end{lemma}

\section{LPF and LPF-Trees}\label{sec:LPFandLPFTree}

In this section we introduce several notations and data structures that are useful for maintaining the $\LZss$ factorization.
We note that some of the notations and data structures already appeared in \cite{BCR24}.

Let $S[1..n]$ be a string and let $i\in [n]$.
The \emph{Longest Previous Factor} of $i$ in $S$ (denoted as $\mathsf{LPF}_S(i)$) is the length of the longest prefix of $S[i..n]$ that occurs strictly before $i$.
Formally,
\[
\mathsf{LPF}_S(i) = \max_{j<i}(\LCP_S(i,j)).
\]

We also denote $\LPF_S(i)=\max(\mathsf{LPF}_S(i),1)$.
Additionally, let $\LPFpos_S(i)$ be the rightmost position that is a witness for $\mathsf{LPF}_S(i)$.
Formally,
\[
\LPFpos_S(i) = \max \{ j < i \mid \LCP_S(i,j)=\mathsf{LPF}_S(i)\}.
\]
It is straightforward that if the $i$th phrase in $\LZss(S)$ is $S[a..b)$  for $b\ne n+1$ then the $(i+1)$th phrase is $S[b..b+\LPF_S(b))$.

In \cite{BCR24}, the authors prove that the values $i+\LPF_S(i)$  are monotone, as stated formally below.

\begin{lemma}[{\cite[{Proof of Lemma 19}]{BCR24}}]\label{lem:19}
    For every text $S$ and $i\in [|S|-1]$, it holds that $i+\LPF_S(i)\le i+1+\LPF_S(i+1)$.
\end{lemma}

In the following lemma we show how to compute $\LPF_S(i)$ in $\Otild(1)$ time in dynamic settings.
\begin{lemma}\label{lem:dynamicLPF}
    There is a data structure that maintains a dynamic string $S$ that undergoes edit operations and that given an index $i$, returns $\LPF_S(i)$ in $\Otild (1)$ time.
    The update time of the data structure is $\Otild(1)$.
\end{lemma}

\begin{proof}
    We use the data structure of \cref{lem:subPMDS}.
    We demonstrate the method for calculating $\mathsf{LPF}_S(i)$, noting that the query $\LPF_S(i)$ simply provides $\max(\mathsf{LPF}_S(i),1)$.
    Observe that $\mathsf{LPF}_S(i) = \max\{\ell \mid \Occurrences{S[1..i+\ell-1)}{S[i..i+\ell)}\neq \emptyset \}$.
    In particular, for every $x \in [1..n]$ we have that $\Occurrences{S[1..i+x-1)}{S[i..i+x)}\neq \emptyset$ if and only if $x \le \mathsf{LPF}_S(i)$.
    The algorithm determines the correct value of $\mathsf{LPF}_S(i)$ by performing a binary search on the range $[0..n-1]$, checking if the current value $x$ is larger or smaller than $\mathsf{LPF}_S(i)$ by querying the data structure of \cref{lem:subPMDS}.
    Each step of the binary search is executed in $\Otild(1)$ time, so the overall time complexity of computing $\mathsf{LPF}_S(i)$ is $\Otild(1)$.
\end{proof}

Exploiting \cref{lem:dynamicLPF}, we establish the proof of the upper bound of \cref{thm:weak}.

\begin{proof}[Proof of \cref{thm:weak} (upper bound)]
Given a string $S$, the algorithm builds the dynamic data structure described in \cref{lem:dynamicLPF} that supports $\LPF$ queries.
Therefore, the algorithm takes $\Otild(1)$ time per update.
Upon $\LZss$ query, the algorithm computes the complete $\LZss(S)$ factorization using $|\LZss(S)|$ subsequent $\LPF$ queries.
Having obtained the $\LZss(S)$ factorization, answering all types of $\LZss$ queries is straightforward.
\end{proof}

\para{The $\LPF$-tree data structure.}
We adapt the framework of \cite{BCR24} to the fully dynamic settings.
Namely, \cite{BCR24} define $\LPF$-tree and show how to maintain it in the semi-dynamic settings.
We will show how to maintain the same structure in the fully dynamic settings.
The $\LPF$-tree of a string $S$ of length $n$, denoted as $\T_S$ is a tree $\T_S=(V,E)$ with nodes $V=[n+1]$ representing the indices of $S$ and an additional node $n+1$, and $E=\{ (i,j)\mid i\in [n], j= i+\LPF_S(i) \}.$\footnote{In~\cite{BCR24}, a label is also defined for every edge in the tree.
Our algorithm does not use these labels, hence we omit them.}
This defines a tree whose root is the node $n+1$.
For a node $i\in \T_S$, we denote $\depth_S(i)= \depth_{\T_S}(i)$.
The following connection between $\T_S$ and $\LZss(S)$ was observed by \cite{BCR24}.
\begin{observation}[cf. {\cite[{Observation 18}]{BCR24}}]\label{obs:lpflzconnection}
Let $\pi$ be the path from $1$ to the root of $\T_S$. Then, the number of edges in $\pi$ is $|\pi|= \depth_S(1)=|\LZss(S)|$.
Additionally, for $k \le |\pi|$, if the $k$th edge on $\pi$ is $(i,i+\LPF_S(i))$, then the $k$th phrase of $\LZss(S)$ is equal to $S[i..i+\LPF_S(i))$.
\end{observation}

\subsection{Updating an LPF-tree}\label{sec:lpftreeds}

In this section, we introduce the data structure representing the $\LPF$-tree.
The data structure maintains a rooted tree $U=(V,E)$ that represents $\mathcal{T}_S$.
In our representation, every node maintains a pointer to its parent in the tree.
Specifically, when an update modifies $S$ to be $S'$ the algorithm modifies $U$ from representing $\T_S$ to represent $\T_{S'}$.

\para{Supported updates.}
Our algorithm manipulates and queries $U$ via the following interface.
\begin{enumerate}
    \item $\Insert(v)$ - adds an isolated vertex $v$ (may temporarily make $U$ a forest).
    \item $\Delete(v)$ - removes a leaf $v$ and the edge outgoing from $v$.
    \item $\Link(v,u)$ - adds root $v$ as a child of node $u$.
    \item $\MoveInterval(v,[i..j])$ - gets an interval of vertices $[i..j]$ such that there is a vertex $u\in V$ with all the vertices in $[i..j]$ being children of $u$.
    Moves all the vertices in $[i..j]$ to be children of $v$ instead.
    We define $\MoveInterval(v,i)=\MoveInterval(v,[i..i])$.
    \item $\GetDepth(v)$ - returns $\depth_U(v)$.
    \item $\findIAncestor(v,i)$ - returns the $i$th ancestor of $v$.
\end{enumerate}
The authors of \cite{BCR24} show how to support this functionality using link-cut trees, achieving amortized $\Otild(1)$ time per operation.
By employing top trees \cite{AHLT05} instead of link-cut trees, this functionality can be obtained with worst-case $\Otild(1)$ update time.
We further explain this implementation in \cref{sec:top_trees}.
We call the data structure with the above functionality a \textit{dynamic tree}\footnote{The name "dynamic tree" is often used to describe data structures with link-cut functionality. We name the specialized data structure used in this paper "dynamic tree" as it implemented via such a classical dynamic tree.}.

\newcommand{\ND}{\mathsf{ND}}
\section{Dominant Extension Points}\label{sec:chain}
In this section, we present the concept of Dominant Extension points.
The combinatorial properties of dominant extension points proven in this section play a key role in the runtime analysis of our algorithm.

Let $S$ be a string and let $I\subseteq[n]$ be a set of indices in $S$.
For $i\in I$, let $I_i=I\cap[1..i-1]$ be all the indices of $I$ that are smaller than $i$.
For every $j\in I_i$ let $p_{i,j,S}=(\LCS_S(i,j),\LCP_S(i,j))$  and let $\ext_{I,i,S}=\{p_{i,j,S}\mid j\in I_i\}$.
Let $\NDExt{I}{i}{S}=\ND(\ext_{I,i,S})$ be the set of non-dominated points in $\ext_{I,i}$.
Formally, for a set of $2$-dimensional points $P$, the set of non-dominated points of $P$ is $\ND(P)=\{(x,y)\in P\mid \nexists_{(x',y')\in P\setminus \{(x,y)\}} x'\ge x \textit{ and } y' \ge y \}$.
When the string is clear from context we omit it from the notation.

Notice that $|\NDExt{I}{i}{S}|\le|\ext_{I,i,S}|\le |I_i|$ and therefore the sum $\sum_{i\in I} |\NDExt{I}{i}{S}|$ is bounded by $O(|I|^2)$.
In the next lemma we prove a tighter upper bound on this sum.

\begin{lemma}\label{lem:chain}$\sum_{i\in I} |\NDExt{I}{i}{S}|=O(|I|\log^2n)$
\end{lemma}

\begin{proof}
Let $T_R$ be the trie over the suffixes $\mathcal{S}_R= \{S[i..n] \mid i \in I \}$.
Let $T_L$ be the trie over the prefixes $\mathcal{S}_L = \{ S[1..i]^R \mid i \in I \}$.
In $T_R$ (resp.\ in $T_L$) the node corresponding to $S[i..n]$ (resp.\ to $S[1..i]^R$) is called $i_R$ (resp.\ $i_L$).
The size of each string in $\mathcal{S}_R$ or in $\mathcal{S}_L$ is at most $n$, so the total size of $T_L$ and $T_R$ is $O(n^2)$ and in particular $\log |T_L| , \log|T_R|\in O(\log n)$.

For a direction $d\in \{ L,R\}$, let  $\LCA^d_{i}= \{\LCA_{T_d}(i_d,j_d) \mid p_{i,j}\in \NDExt{I}{i}{S}\}$.
Recall that $\depth_{T_R}(\LCA_{T_R}(i_R,j_R))=\LCP_S(i,j)$ and $\depth_{T_L}(\LCA_{T_L}(i_L,j_L))=\LCS_S(i,j)$.
Also, notice that by definition of non-dominated points, there are no two points $p_{i,j},p_{i,j'}\in\NDExt{I}{i}{S}$ with $\LCP_S(i,j) = \LCP_{S}(i,j')$ or with $\LCS_S(i,j) = \LCS_S(i,j')$.
This implies that the depths of elements $\LCA^L_i$ are distinct and the depths of elements in $\LCA^R_i$ are distinct.
This means that there are no two points $p_{i,j},p_{i,j'}\in \NDExt{I}{i}{S}$ with $\LCA_{T_L}(i_L,j_L) = \LCA_{T_L}(i_L,j'_L)$ or $\LCA_{T_R}(i_R,j_R) = \LCA_{T_R}(i_R,j'_R)$.

We make use of the following well-known fact.
\begin{fact}[Heavy path decomposition]\label{fact:heavypathdecompos}
    A rooted tree with $n$ nodes can be decomposed into simple paths, called \emph{heavy paths}, such that every node-to-root path intersects $O(\log n)$ heavy paths.
\end{fact}

We assume that $T_L$ and $T_R$ are partitioned into heavy paths.
For every $i\in I$ we denote by $\pi^L_i$ and $\pi^R_i$ the paths from $i_L$ to the root of $T_L$ and from $i_R$ to the root of $T_R$, respectively.
Denote by $H^L_i$ and $H^R_i$ the heavy paths that intersect with $\pi^L_i$ and with $\pi^R_i$, respectively.
Then $|H^L_i|+|H^R_i|\in O(\log n)$.

We make the following charging argument.
Let $p_{i,j}\in \NDExt{I}{i}{S}$.
Let $h^R_{i,j}$ and $h^L_{i,j}$ be the heavy paths containing $\LCA_{T_R}(i_R,j_R)$ and $\LCA_{T_L}(i_L,j_L)$, respectively.
We say that $p_{i,j}$ is a type 1 point of $\NDExt{I}{i}{S}$ if $\LCA_{T_R}(i_R,j_R)$ has the maximal depth among the vertices of $\LCA^R_i \cap h^R_{i,j}$ or if $\LCA_{T_L}(i_L,j_L)$ has the maximal depth in $\LCA^L_{i} \cap h^L_{i,j}$.
Otherwise, $p_{i,j}$ is a type 2 point of $\NDExt{I}{i}{S}$.
If $p_{i,j}$ is a type 1 element, we charge $p_{i,j}$ on $i$.
Otherwise, we charge $p_{i,j}$ on $j$.

Since both $\pi^L_i$ and $\pi^R_i$ visit $O(\log n)$ heavy paths, and since the depths of elements in $\LCA^L_i$ are distinct and the depths of elements in $\LCA^R_i$ are distinct, every $i\in I$ is charged on $O(\log n)$ points of type 1.

We next prove that every $j\in I$ is charged on $O(\log^2 n)$ points of type 2.
Let $\ChargeTwo(j)$ be the set of indices $i$ such that $p_{i,j}$ is a type 2 point in $\NDExt{I}{i}{S}$.
We claim that there are no two indices $i,i'\in \ChargeTwo(j)$ such that $(h^L_{i,j},h^R_{i,j})= (h^L_{i',j},h^R_{i',j})$.
Notice that since for every $j\in I$ it holds for every $k\in I$ that $(h^L_{k,j},h^R_{k,j}) \in H^{L}_j \times H^R_j$, the above claim directly implies $|\ChargeTwo(j)|=O(\log^2 n)$.
This would conclude the proof, as it shows that every index in $I$ is charged at most $O(\log^2 n + \log n)$ times, leading to $\sum_{i\in I} |\NDExt{I}{i}{S}| \in O(|I| \log^2 n)$.

Assume by contradiction that there are two elements $i<i'\in \ChargeTwo(j)$ such that $(h^L_{i,j},h^R_{i,j})=(h^L_{i',j},h^R_{i',j})$.
For $d\in \{L,R\}$, let $v_d$ be the lowest node in $h^d_{i,j}$ and let $z_d$ be the child of $\LCA_{T_d}(v_d,j_d)$ in $h^d_{i,k}$ (see \cref{fig:chain}).
Since $p_{i,j}$ is a point of type 2 in $\NDExt{I}{i}{S}$, there is a node in $\pi^d_i\cap h_{i,j}^d$ strictly bellow $\LCA_{T_d}(v_d,j_d)$.
Therefore, $z_d$ must exist, and it holds that $z_d\in \pi^d_i$.
By symmetric reasoning $z_d\in \pi^d_i\cap \pi^d_{i'}$ which implies
\begin{equation}\label{eq:depzd}
    \depth_{T_d}(\LCA_{T_d}(i_d,i'_d))\ge \depth_{T_d}(z_d).
\end{equation}
In addition,
\begin{equation}\label{eq:jisthetop}
\LCA_{T_d}(v_d,j_d)=\LCA_{T_d}(i_d,j_d)=\LCA_{T_d}(i'_d,j_d).
\end{equation}

\begin{figure}[h]
    \centering
    \includegraphics[width=0.6\textwidth]{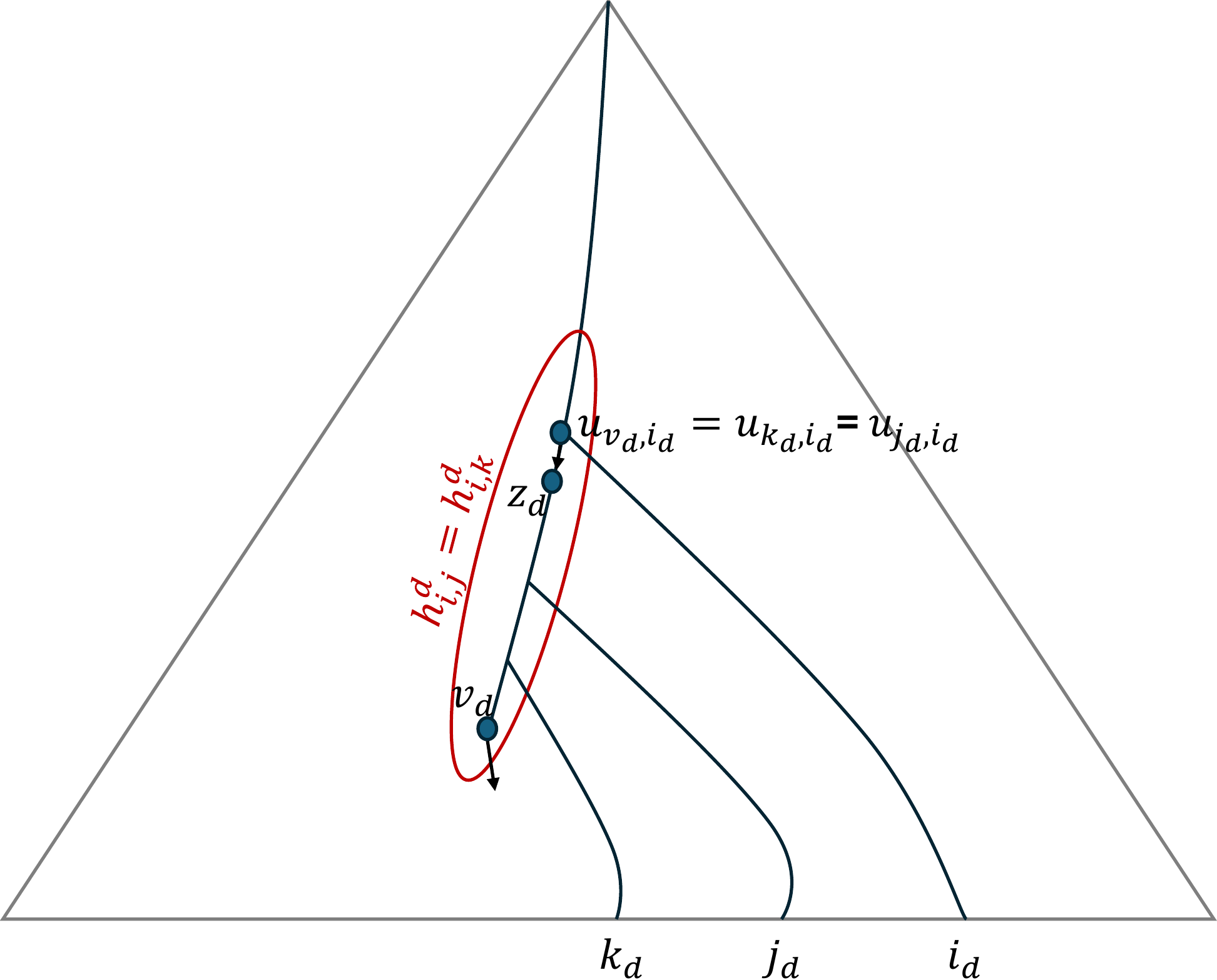}
    \caption{A schematic figure of $\T_d$.}
    \label{fig:chain}
\end{figure}

\noindent
Therefore,
\begin{align*}
\LCP_S(i',i)&= \depth_{T_R}(\LCA_{T_R}(i'_R,i_R))\stackrel{(\ref{eq:depzd})}{\ge} \depth_{T_R}(z_R)\\&>\depth_{T_R}(\LCA_{T_R}(v_R,j_R))
\stackrel{(\ref{eq:jisthetop})}{=}\depth_{T_R}(\LCA_{T_R}(i'_R,j_R))=
\LCP_S(i',j)
\end{align*}
By a symmetric argument $\LCS_S(i',i)>\LCS_S(i',j)$.
Therefore $p_{i',i}\in\ext_{I,i'}$ dominates on $p_{i',j}$ (i.e. is strictly larger in both coordinates), contradicting the assumption $p_{i',j}\in\NDExt{I}{i'}{S}$.
\end{proof}

\section{Types of Indices}\label{sec:alg}

Our data structure maintains an instance $U$ of a dynamic tree described in \cref{sec:LPFandLPFTree} that represents the $\LPF$ tree $\T_S$ throughout the dynamic sequence of updates.
This section is dedicated to describing how to properly update $U$ to represent $\T_S$ when an edit operation is applied to $S$.
That is, we assume that $U$ currently represents $\T_S$ for the dynamic string $S$, and that an edit operation is applied to $S$ at index $z$, resulting in $S'$.
We describe an algorithm that modifies $U$ to represent $\T_{S'}$.
We fix the notation of $S$,$S'$, and $z$.
Also, we refer to the update as $(S,S')$.

Let $m$ be an integer to be fixed later.
Let $M_L$ be the substring of length $m$ to the left of $z$, i.e., $M_L = S[z-m..z-1]$, and let $M_R$ be the substring of length $m$ to the right of $z$, i.e., $M_R = S[z+1..z+m]$.\footnote{If $z \le m$ or $z \ge n-m+1$, $M_L$ and $M_R$ are truncated at the edges of $S$, we ignore that in future discussion to avoid clutter.} See \cref{fig:string}.

\begin{figure}[h!]
    \centering
\begin{center}
\begin{tikzpicture}
    \draw (1,0) rectangle (10,0.5);

    \draw[fill=purple!40] (3,0) rectangle (5,0.5) node[pos=.5] {\Large $M_L$};

    \draw[fill=orange!50] (5,0) rectangle (5.5,0.5) node[pos=.5] {};

    \draw[fill=green!30] (5.5,0) rectangle (7.5,0.5) node[pos=.5] {\Large $M_R$};

    \node[below] at (5.25,0) {$z$};
\end{tikzpicture}
\end{center}
    \caption{An illustration of $M_L$, $z$, $M_R$.}
    \label{fig:string}
\end{figure}

An index $i$ is called \emph{active} if $i$'s parent in the $\LPF$-tree changed as a result of the update, i.e., $\LPF_S(i) \neq \LPF_{S'}(i)$.
Clearly, only active indices need to be found and updated by our algorithm.
Let $i$ be an active index.
We consider several classes of indices.
Firstly, \textit{heavy} indices are indices that have very large $\LPF$ values both before and after the update.
\textit{Super light} and \textit{light} are two other classes: super light indices are indices with $i$ close to $z$, and light indices are indices $i$ with either $\LPFpos_S(i)$ or $\LPFpos_{S'}(i)$ close to $z$.
We provide an algorithm for identifying all of the indices of each type, and update its parent in $U$ to its possibly new parent in $\T_{S'}$.
We also show that every active index $i$ must belong to one of these classes.

We formally define heavy indices as follows.
\begin{definition}[L-Heavy index]\label{def:heavy}
    We say that index $i\in [|S'|]$ is an \emph{$L$-heavy index} (resp.\ \emph{$R$-heavy index}) for the update $(S,S')$ if $\LPF_{S}(i) \neq \LPF_{S'}(i)$ and there is an index $k \in \Occurrences{S}{M_L} \cap \Occurrences{S'}{M_L}$ (resp.\ $M_R$) such that $k\in [i.. i+ \min (\LPF_{S}(i),\LPF_{S'}(i))-m]$.
\end{definition}

We also formally define that an index $i$ is super light if $i\in [z-m..z]$ and that an index $i$ is light if for some $T\in\{S,S'\}$  there are integers $a,b\in[0..m]$ such that $i=\min(\Occurrences{T[z+1..n]}{T[z-a..z+b]})$.
In the following lemma, we show that all active indices belong to (at least) one of the defined classes.
\begin{lemma}\label{lem:indextypes}
    Every active index $i\in [|S'|]$ is either $L$-heavy, $R$-heavy, super light, or light.
\end{lemma}

\begin{proof}
Let $i \in [|S'|]$ be an active index.
In the following we assume that $\LPF_S(i)>\LPF_{S'}(i)$.
The case where $\LPF_{S'}(i) > \LPF_{S}(i)$ can be proved in a similar manner, switching the roles of $S$ and $S'$.
Notice that due to our assumption that $\LPF_S(i) > \LPF_{S'}(i)\ge 1$ we have in particular $\LPF_S(i) > 1$ and therefore $\LPF_S(i) = \mathsf{LPF}_S(i)$.

\begin{claim}\label{clm:cases}
    $z\in[i..i+\LPF_{ S}(i))\cup[\LPFpos_{S}(i)..\LPFpos_{{S}}(i)+\LPF_{{S}}(i))$
\end{claim}

\begin{claimproof}
    Assume to the contrary that $z\not\in[i..i+\LPF_{S}(i))$, and $z\not\in  [\LPFpos_{{S}}(i)..\LPFpos_{{S}}(i)+\LPF_{{S}}(i))$.
    We therefore have that $S'[i..i+\LPF_{S}(i))= S[i..i+\LPF_{S}(i))= S[\LPFpos_S(i)..\LPFpos_S(i)+\LPF_{S}(i))= S'[\LPFpos_S(i)..\LPFpos_S(i)+\LPF_{S}(i))$ (the second equality follows from the definition of $\LPFpos$).
    In particular, $\LPF_{S'}(i) \ge \LPF_S(i)$, a contradiction.
\end{claimproof}

We consider the two cases arising from \cref{clm:cases} regarding the position of $z$ (see \cref{fig:cases}).

\begin{figure}[h]
    \centering
    \begin{subfigure}{\textwidth}
        \centering
        \begin{minipage}{0.65\textwidth}          \includegraphics[width=\linewidth]{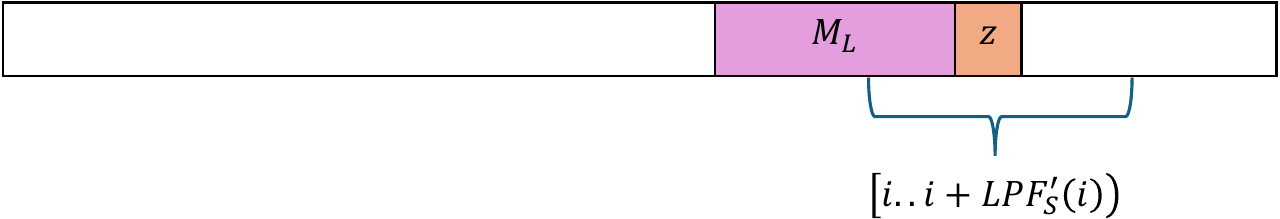}
        \end{minipage}
        \hspace{15pt}
        \begin{minipage}{0.25\textwidth}
            \caption{Example of Case 1, super light index.}
            \label{fig:case1a}
        \end{minipage}
    \end{subfigure}

    \vspace{10pt}

    \begin{subfigure}{\textwidth}
        \centering
        \begin{minipage}{0.65\textwidth}          \includegraphics[width=\linewidth]{case1b.pdf}
        \end{minipage}
        \hspace{15pt}
        \begin{minipage}{0.25\textwidth}
            \caption{Example of Case 1, $L$-heavy index.}
            \label{fig:case1b}
        \end{minipage}
    \end{subfigure}

    \vspace{10pt}

    \begin{subfigure}{\textwidth}
        \centering
        \begin{minipage}{0.65\textwidth}
            \includegraphics[width=\linewidth]{case2a.pdf}
        \end{minipage}
        \hspace{15pt}
        \begin{minipage}{0.25\textwidth}
            \caption{Example of Case 2a, $L$-heavy index.}
            \label{fig:case2a}
        \end{minipage}
    \end{subfigure}

    \vspace{10pt}

    \begin{subfigure}{\textwidth}
        \centering
        \begin{minipage}{0.65\textwidth}
            \includegraphics[width=\linewidth]{case2b.pdf}
        \end{minipage}
        \hspace{15pt}
        \begin{minipage}{0.25\textwidth}
            \caption{Example of Case 2b - light index.}
            \label{fig:case2b}
        \end{minipage}
    \end{subfigure}

    \vspace{10pt}

    \begin{subfigure}{\textwidth}
        \centering
        \begin{minipage}{0.65\textwidth}
            \includegraphics[width=\linewidth]{case2c1.pdf}
        \end{minipage}
        \hspace{15pt}
        \begin{minipage}{0.25\textwidth}
            \caption{example of Case 2c, $R$-heavy case.}
            \label{fig:case2c}
        \end{minipage}
    \end{subfigure}

    \begin{subfigure}{\textwidth}
        \centering
        \begin{minipage}{0.65\textwidth}
            \includegraphics[width=\linewidth]{case2c2.pdf}
        \end{minipage}
        \hspace{15pt}
        \begin{minipage}{0.25\textwidth}
            \caption{Example of Case 2c, light index.}
            \label{fig:case2c2}
        \end{minipage}
    \end{subfigure}

    \caption{An illustration of the cases in the proof of \cref{lem:indextypes}.
    In each of the illustrations above, the blue parenthesis below the string marks the equality $S[i..i+\LPF_S(i)) = S[\LPFpos_S(i) ..\LPFpos_S(i) + \LPF_S(i))$.
    The red parenthesis above the string demonstrate the equivalent equality in $S'$.\label{fig:cases}}
\end{figure}

\begin{enumerate}
    \item
\underline{Case 1}: $z\in[i..i+\LPF_{ S}(i))$.
    If $z-i \le m$, we have that $i$ is a super light index (see \cref{fig:case1a}).
    Otherwise, we prove that $i$ is an $L$-heavy index (see \cref{fig:case1b}).
    Let $d=(z-1)-i < \LPF_{ S}(i)$, be the difference between $i$ and the end of $M_L$.
    Since the update does not change the content of $S$ in indices smaller than $z$, we have $S[i..i+d] = S'[i..i+d]$.
    From the definition of $\LPFpos_{ S}(i)$ and from $d < \LPF_{S}(i)$ we have $S [i..i+d] =  S[\LPFpos_{ S}(i).. \LPFpos_{S}(i) + d]$.
    Since $\LPFpos_{S}(i) < i$, we have that $\LPFpos_{ S}(i) + d < z-1$.
    Again, since the update does not affect indices smaller than $z$, we have
    $S'[\LPFpos_{S}(i)..\LPFpos_{S}(i) + d]=S[\LPFpos_{S}(i)..\LPFpos_{S}(i) + d]$.
    From transitivity, we have $S'[i..i+d] = S'[\LPFpos_{S}(i).. \LPFpos_S(i) + d]$.
    This implies that $\min(\LPF_S(i),\LPF_{S'}(i)) = \LPF_{S'}(i)  \ge d+1=z-i> m$.
    It follows that $i$ is an $L$-heavy index since for $k=z-m$ we have that $k\in[i..i+\min(\LPF_S(i),\LPF_{S'}(i))-m]$ and $k$ is an occurrence of $M_L$ both in $S$ and in $S'$.

\item    \underline{Case 2}: $z\in [\LPFpos_{{S}}(i)..\LPFpos_{{S}}(i)+\LPF_{{S}}(i)) \setminus [i..i+\LPF_S(i))$.
Let $\lpos={S}[\LPFpos_{{S}}(i)..z)$ and $\rpos={S}[z..\LPFpos_{{S}}(i)+\LPF_{{S}}(i))$.
There are three sub-cases:
\begin{enumerate}
    \item  $|\lpos|\ge m$ (see \cref{fig:case2a}).
    Since $\lpos$ occurs before $i$ both in $S$ and in $S'$, and since $S[i..i+\LPF_S(i)) = S'[i..i+\LPF_S(i))$ due to $z\notin [i..i+\LPF_S(i))$, we have that $\min(\LPF_S(i),\LPF_{S'}(i))\ge|\lpos|$.
    It follows that $\LPFpos_{S}(i) + \min(\LPF_{S}(i), \LPF_{S'}(i)) \ge z$.
    From $S[i..i+\LPF_S(i))=S[\LPFpos_S(i)..\LPFpos_S(i)+\LPF_S(i))$ we have, in particular, $S[i+|\lpos|-m.. i+|\lpos|) = S[\LPFpos_{S}(i) + |\lpos|-m..\LPFpos_{S}(i) + |\lpos|) = S[z-m..z)=M_L$.
    We have shown that there is an occurrence of $M_L$ in $S$ at index $i+|\lpos|-m$.
    That is an occurrence in $S'$ as well, as $S[z-m..z]$ is not affected by the update.
    We have also shown that $\min(\LPF_S(i), \LPF_{S'}(i)) \ge |\lpos|$, which together indicates that $i$ is an $L$-heavy index.

    \item  $|\lpos|< m$ and $|\rpos|\le m$ (see \cref{fig:case2b}).
    In this case, we have $z-m \le \LPFpos_{S}(i) < i$.
    If $i\in [z-m,z]$ then $i$ is super light.
    Otherwise, $i\ge z+1$.
    Notice that in this case, $S[i..i+\LPF_{S}(i))= S[\LPFpos_{S}(i) .. \LPFpos_{S}(i) + \LPF_{S}(i))=S[z-a..z+b]$ for  $a=|\lpos|<m,b=|\rpos|-1<m$.
    Recall that among all occurrences of $S[i..i+\LPF_{S}(i))$ to the left of $i$, the occurrence at $\LPFpos_{S}(i)$ is the rightmost one.
    Therefore, $i$ is a light index as the first occurrence of ${S}[z-a..z+b]$ after $z$.

    \item \label{case:complicated_shay}$|\lpos|< m$ and $|\rpos|> m$ (see \cref{fig:case2c,fig:case2c2}).
    Note that in this case $\LPFpos_{S}(i) = z-a$ for some $a\in [m]$.
    As in the previous case, $i > z-a \ge z-m$ and therefore if $i\le z$ then $i$ is a super light index.
    We assume $i \ge z+1$.
    Let $j=\min(\Occurrences{{S}[z+1..n]}{{S}[z-a..z+m]})$.
    Notice that $j\le i$, also notice that $j=\min(\Occurrences{{S'}[z+1..n]}{S[z-a..z+m]})$ since the content of $S$ after index $z$ was not changed.

    If $j=i$, then $i$ is a light index.
    Otherwise, we claim that $i$ is $R$-heavy.
    Notice that there is an occurrence of $M_R= S[z+1..z+m]$ at index $k=i+a+1$ both in $S$ and in $S'$.
    Additionally, due to the indices after $z+1$ not being affected by the update we have $\min(\LCP_S(i,j),\LCP_{S'}(i,j)) \ge a+m+1$.
    It follows from the maximality of $\LPF$ that $\min(\LPF_S(i),\LPF_{S'}(i)) \ge a+m+1$ and therefore  $k=i+a+1\in[i..i+(a+m+1)-m]\subseteq [i.. i+\min(\LPF_S(i),\LPF_{S'}(i))-m]$, as required.\qedhere
\end{enumerate}
\end{enumerate}
    \end{proof}

In the rest of the section we describe how to modify $U$, which  initially represents $\T_S$, to represent $\T_{S'}$.
We present three separate algorithms, one for updating each type of indices.
By updating a type of indices, we mean that after the algorithm for this type is applied, every index $i$ of this type has its parent in $U$ properly set to $i+\LPF_{S'}(i)$.
Notice that non-active indices do not need to be treated.
The algorithms for the light and for the super light indices are straightforwardly implemented following the definition of the types.
The $L$-heavy and $R$-heavy indices require a more intricate care.

\para{Super light indices.}
To update all super light Indices, the algorithm applies the following procedure.
For every $i\in [z-m-1..z]$, compute $\LPF_{S'}(i)$ on $S'$ using \cref{lem:dynamicLPF} and apply $\MoveInterval(i+\LPF_{S'}(i),i)$ to $U$.
It is easy to see that all super light indices have their parents properly set after this procedure is applied.
The following directly follows.
\begin{observation}\label{obs:superlighttime}
    There is an algorithm that sets the parent of every super light index $i$ of the update to $i+\LPF_{S'}(i)$ in $\Otild(m)$ time.
\end{observation}

\para{Light indices.}
The algorithm updates light indices as follows.
For every $a,b\in [0..m]$, the algorithm finds $k = \min (\Occurrences{S[z+1..n]}{S[z-a..z+b]})$ using \cref{lem:first_occ}.
The algorithm applies $\MoveInterval(k+\LPF_{S'}(k),k)$ to $U$.
For  $k' = \min (\Occurrences{S'[z+1..n]}{S'[z-a..z+b]})$, the algorithm applies $\MoveInterval(k'+\LPF_{S'}(k'),k')$ to $U$.

Clearly, every light index has its parent in $U$ properly set after running the above procedure.
The following directly follows.
\begin{observation}\label{obs:lighttime}
    There is an algorithm that sets the parent of every light index $i$ of the update $(S,S')$ to $i+\LPF_{S'}(i)$ in $\Otild(m^2)$ time.
\end{observation}
The last ingredient of the algorithm is stated in the following lemma, whose proof is the most technically involved part of the paper.
We dedicate \cref{sec:heavy} for the proof of this lemma.
\begin{lemma}\label{lem:heavytime}
    There is an algorithm that sets the parent of every $L$-heavy and $R$-heavy index $i$ of the update $(S,S')$ to $i+\LPF_{S'}(i)$ in $\Otild(n /m)$ time.
\end{lemma}

Combining \cref{obs:superlighttime,obs:lighttime,lem:heavytime}, we are ready to prove \cref{thm:lz77}.

\begin{proof}[Proof of \cref{thm:lz77}]
    Given a string $S$, the algorithm builds the dynamic data structures described in \cref{lem:dympillar,lem:ipmalllengths,lem:first_occ,lem:dynamicLPF} that support pattern matching, $\LCP$, $\LCS$, and $\LPF$ queries.
    Since the update time for each of these data structures is $\Otild(1)$, they can be constructed in $\Otild(n)$ time by initializing them for the empty string and inserting $S$ symbol by symbol.
    In addition, the algorithm builds $U=\T_S$ in $\Otild(n)$ time using $O(n)$ $\LPF$ queries.

When an update is applied to the string $S$ at index $z$, resulting in a modified string $S'$, the algorithm applies the update on the data structures mentioned above, and also applies \cref{obs:superlighttime,obs:lighttime,lem:heavytime} to properly assign the parents of all super light, light, and heavy indices in $S$ with respect to the update $(S,S')$.
The remaining indices are non-active due to \cref{lem:indextypes}, and therefore require no update.
Thus, the application of \cref{obs:superlighttime,obs:lighttime,lem:heavytime}, on $U = \T_S$ results in $U=\T_{S'}$.
We refer to this part of the algorithm as the \emph{transition}.

If the update inserts an index $z$, the algorithm adds $z$ to $U$ with $\Insert(z)$ before the transition.
If the update deletes the index $z$, the algorithm deletes the index $z$ with $\Delete(z)$ from $U$ after the transition.

By \cref{obs:superlighttime,obs:lighttime,lem:heavytime}, the transition is applied in $\Otild(m^2 + \frac{n}{m})$ time.
By setting $m= n^{1/3}$ we obtain $\Otild(n^{2/3})$ update time.

To query $\SelectPhrase(i)$, the algorithm calls $a= \findIAncestor(1,i-1)$ and $b= \findIAncestor(1,i)-1$ in $\Otild(1)$ time.
Then, the algorithm reports that the $i$th phrase of $\LZss(S)$ is $S[a,b]$. This query computation is correct due to \cref{obs:lpflzconnection}.

To query $\ContainingPhrase(i)$, the algorithm binary search $\LZss(S)$ using $\SelectPhrase$ to find $k$ such that the $k$th phrase is a range $[a..b]$ satisfies $i\in [a..b]$.
Each step in the binary search takes $\Otild(1)$ time, and hence $\ContainingPhrase(i)$ is implemented in $\Otild(1)$ time.

To query $\LZLength(i)$, the algorithm applies a binary search for $z$ such that  $S[a,b] = \SelectPhrase(z)$ satisfies $i\in [a,b]$.
This is implemented using $O(\log n)$ queries to $\SelectPhrase$ for a total of $\Otild(1)$ time, as required.
\end{proof}

\section{The Heavy Algorithm}\label{sec:heavy}

This section is devoted to proving \cref{lem:heavytime}.
We start in \cref{sec:aperiodic-heavy} by considering a simple case, where $M_L$ and $M_R$ are aperiodic strings.
Then, in \cref{sec:periodic} we introduce the algorithm for the more complicated case where $M_L$ or $M_R$ are periodic.

\subsection{Warm-up - the Aperiodic Case}\label{sec:aperiodic-heavy}

As a warm-up, we present a proof of \cref{lem:heavytime} under the simplifying assumption that both $M_L$ and $M_R$ are aperiodic.
The general proof is deferred to \cref{sec:periodic}.
For an interval $[a..b]$ partitioned into sub-intervals, $[a = a_1..b_1],[a_2..b_2],\dots,[a_k..b_k=b]$ with $a_{i+1}=b_i+1$ for every $i\in[k-1]$, we denote the set of partition intervals as $[a_1,a_2, \ldots ,a_k, b]$.
We start by presenting some tools.

\para{Sequence of critical indices.} Given a string $S$ and two indices $i<j$ we define the sequence $L_S(i,j)=(\ell_{i,0},\ell_{i,1},\ldots)$.
Intuitively, the elements of $L_S(i,j)$ are critical indices in which the value $x+ \LPF_S(x)$ changes, limited to the range in which this value exceeds $j$ (in decreasing order).
Formally, let $\ell_{\min}$ be the minimal index in $[1..i]$ such that $\ell_{\min}+\LPF_S(\ell_{\min})> j$.
If $\ell_{\min}$ does not exist, $L_S(i,j)=(\ell_{i,0})=(i)$; otherwise, we define an element $\ell_{i,k}\in L_S(i,j)$ recursively.
\[
\ell_{i,k} =
\begin{cases}
i & \text{if }k=0,\\
\min\{i'\in[\ell_{\min}..i]\mid i'+\LPF_S(i')=i+\LPF_{S}(i)\}
& \text{if } k = 1, \\
\min\{i'\in[\ell_{\min}..i]\mid i'+\LPF_S(i')=\ell_{i,k-1}-1+\LPF_{S}(\ell_{i,k-1}-1)\} & \text{if } k > 1.
\end{cases}
\]
The following observation follows from the monotonicity of $x+\LPF_S(x)$ values (\cref{lem:19}).
\begin{observation}\label{obs:critical_indices}
    $x\in L_S(i,j)\setminus\{\ell_{i,0}\}$ if and only if $x+\LPF_S(x)>x-1+\LPF_S(x-1)$ and $x+\LPF_S(x)>j$.
\end{observation}
$|L_S(i,j)|$ is the number of defined elements in $\{\ell_{i,k}\}$.
We sometimes omit the subscript $i$ from $\ell_{i,k}$ and simply write $\ell_k$ when $i$ is clear from context.

\begin{lemma}\label{lem:fastintervalscomputation}
    There is an algorithm that given $i$, $j$, and access to the data structure of \cref{lem:dynamicLPF}, computes $L_S(i,j)$ in $\Otild(|L_S(i,j)|)$ time.
\end{lemma}
\begin{proof}
Recall that due to \cref{lem:19}, the values $j+\LPF_S(j)$ are monotonically increasing.
This property allows us to find, given a value $x$ and a range $[a..b]$, the minimal index $c\in [a,b]$ with $c+\LPF_S(c)-1=x$ (or with $c + \LPF_S(c)-1 \ge x$) by employing a binary search.
Since $\LPF_S(c)$ can be found in $\Otild(1)$ time for an input vertex $c$ using \cref{lem:dynamicLPF}, this binary search is executed in $\Otild(1)$ time.

The lemma follows by applying this procedure $|L_S(i,j)|+1$ times as follows.
First, find $\LPF_S(i)$.
Then, find $\ell_{\min}$ which is the minimal index in $[1..i]$ with $\ell_{\min} + \LPF_S(\ell_{\min})-1\ge j$.
If $\ell_{\min}$ does not exist, return $L_S(i,j)=(\ell_{i,0})=(i)$.
Otherwise, find $\ell_{i,1}$ as the minimal element in $[\ell_{\min}..i]$ with $\ell_{i,1} + \LPF_S(\ell_{i,1}) = i+\LPF_S(i)$.
Notice that $\ell_{i,1}$ always exists since $i\in[\ell_{\min}..i]$.
As long as $\ell_{k}$ exists, $\ell_{k+1}$ can be found straightforwardly according to the definition of $L_S(i,j)$ by applying the binary search procedure.
When the algorithm finds that $\ell_{i,k}$ is undefined, the algorithm reports $L_S(i,j)$.
\end{proof}

We say that an interval $[s..t]$ is \emph{$S$-clean} if for all $i,j\in [s..t]$, it holds that $i+\LPF_S(i) = j+\LPF_S(j)$.
A sequence of intervals $R=[s_1,s_2,\dots,s_k]$ is $S$-clean if every interval $r\in R$ is $S$-clean.
If $R$ is both $S$-clean and $S'$-clean, we say that $R$ is clean.

Recall that the last element of $L_S(i,j)$ is $\ell_{\min}$ (unless $L_S(i,j)=(i)$), and that $L_S(i,j)$ has a decreasing order.
Let $R_S(i,j)$ be the sequence of intervals $[\ell_{\min}=\ell_{|L_S(i,j)|-1},\ell_{|L_S(i,j)|-2},\dots,\ell_{1},\ell_{0}=i]$.
Obviously, $R_S(i,j)$ partitions $[\ell_{\min}..i]$.
By \cref{lem:19}, there is an index $u$ such that for all $v\in[\ell_k..\ell_{k-1}-1]$,
it holds that $u=v+\LPF_S(v)$.
It holds that for every $i < j$, the sequence $R_S(i,j)$ is $S$-clean.

\para{The (aperiodic) algorithm.} We present an algorithm for updating the parents of all indices that are $L$-heavy for the update.
An algorithm for updating all $R$-heavy indices is obtained in an identical manner by replacing $M_L$ with $M_R$.

First, find all occurrences of $M_L$ in $S$ and in $S'$ using \cref{lem:dympillar}.
For every pair $i,j$ such that $S[i..j] = S'[i..j] = M_L$, use \cref{lem:fastintervalscomputation} to find both $L_{S}(i,j)$ and $L_{S'}(i,j)$.
Let $p=\max\{\min (L_S(i,j)),\min (L_{S'}(i,j))\}$.
Sort $(L_{S}(i,j) \cup L_{S'}(i,j))\cap[p..i]$ in increasing order to obtain a sequence of positions $P=(a_1,a_2,\dots,a_t)\subseteq[p..i]$.
Notice that both $p\in P$ and $i\in P$.
The sequence $P$ induces a partition of $[p..i]$ to $R= [a_1,a_2,\dots,a_t]$.
Notice that $R$ is clean, since every $[a..b]\in R$ is contained in some interval of $R_S(i,j)$ and in some interval of $R_{S'}(i,j)$.
For every $[a..b]\in R$, the algorithm applies $\MoveInterval(a + \LPF_{S'}(a), [a,b])$ to $U$.
Notice that this operation is valid since $R$ is clean (in particular, $S$-clean).

\para{Correctness.}
We show that for every index $v$ which is an $L$-heavy index, the algorithm updates the parent of $v$  in $U$ to $v+\LPF_{S'}(v)$.
Let $v$ be an $L$-heavy index.
Since $v$ is $L$-heavy, there are indices $i,j$ such that $S[i..j]=S'[i..j]=M_L$ and $S[i..j]$ is contained in $S[v.. v+ \min(\LPF_{S'}(v),\LPF_{S}(v))-1]$.
Let $p=\max(\min(L_{S}(i,j)),\min(L_{S'}(i,j)))$.
We prove that $v\in [p..i]$.
Let $q=\min(L_{S}(i,j))$.
Recall that $q$ is the minimal index with $q+\LPF_S(q)-1 \ge j$.
Since $v+\LPF_S(v)-1\ge j$ it follows that $v\ge q=\min(L_{S}(i,j))$.
Similarly, we have that $v\ge \min(L_{S'}(i,j))$.
Thus, $v\ge p$ as we claimed.
We also have $v\le i$.
It follows that when the algorithm processes $S[i..j]$ as an occurrence of $M_L$, it generates the sequences $L_{S}(i,j)$ and $L_{S'}(i,j)$ such that $v \in [p..i]$.
Therefore, the algorithm creates a partition $R$ of this interval and applies $\MoveInterval(a + \LPF_{S'}(a) , [a,b])$ for every interval $[a,b] \in R$, in particular to the interval $[a_v,b_v]$ containing $v$.
Since all intervals in $R$ are $S'$-clean, we have that $v + \LPF_{S'}(v) = a_v +\LPF_{S'}(a_v)$, and the parent of $v$ is set as required.

\para{Time complexity.}
We use \cref{lem:ipmalllengths} to find all occurrences of $M_L$ in $\Otild(n/m)$ time.
Then, we find $L_{S}(i,j)$ and $L_{S'}(i,j)$ for every $i,j$ such that $S[i..j] = S'[i..j] = M_L$ using \cref{lem:fastintervalscomputation} in $\Otild(|L_S(i,j))|$ and $\Otild(|L_{S'}(i,j)|)$ time.
Finally, we merge and sort $L_S(i,j)$ and $L_{S'}(i,j)$ to obtain $R$, and apply a constant amount of $\LPF$ queries (using \cref{lem:dynamicLPF}) and a constant number of operations on $U$ for each element in $R$.
The total time complexity is therefore $\Otild (n/m + \sum_{S[i..j]=M_L}\big(|L_{S}(i,j)|+|L_{S'}(i,j)|\big))$.

We observe the following connection between $\NDExt{I}{i}{S}$ and $L_S(i,j)$.

\begin{lemma}\label{lem:bound_Li}
    Let $S$ be a string and let $i,j$ be two indices. Then,  \[|L_S(i,j)|\le
    |\NDExt{\Occurrences{S}{S[i..j]}}{i}{S}|+1.\]
\end{lemma}
Since $M_L$ is aperiodic, we have that $|\Occurrences{S}{M_L}|, |\Occurrences{S'}{M_L}| \in O(n/m)$ due to \cref{fact:aperfarapart}.
We therefore have $\sum_{S[i..j]=M_L}|L_{S}(i,j)| \le O(n/m) +  \sum_{i\in \Occurrences{S}{M_L}} |\NDExt{\Occurrences{S}{M_L}}{i}{S}| \le \Otild(n/m+|\Occurrences{S}{M_L}|) \le \Otild(n/m)$ where the first inequality follows from \cref{lem:bound_Li} and the second inequality follows from \cref{lem:chain}.
Similar arguments lead to $\sum_{i\in \Occurrences{S}{M_L}} |L_{S'}(i,j)|\in \Otild(n/m)$, which in conclusion shows that the running time of the algorithm is $\Otild(n/m)$.
Hence, to complete the proof of \cref{lem:heavytime} for the special case where $M_L$ and $M_R$ are aperiodic, it remains only to prove \cref{lem:bound_Li}.
\begin{proof}[Proof of \cref{lem:bound_Li}]
Let $\ell\in L_S(i,j)\setminus\{i\}$.
By \cref{obs:critical_indices} we have $\ell+\LPF_S(\ell)>\ell-1+\LPF_S(\ell-1)$.
In addition, $\ell+\LPF_S(\ell)>j$.

We denote $d=\LPFpos_S(\ell)$.
Due to the definition of $\LPF$, we have $S[\ell..\ell+\LPF_S(\ell))= S[d..d +\LPF_S(\ell))$ which in particular yields the equality $S[i..j]= S[d + (i-\ell)..d+(j-\ell)]$.
We showed that $o = d + (i-\ell) \in \Occurrences{S}{M_L}$.
The following claim will be useful for us.

  \begin{claim}\label{clm:asmuchtotheleft}
        $\LCS_S(i,o) = i-\ell+1$
    \end{claim}
    \begin{claimproof}
        Recall that $S[d..o] = S[\ell..i]$.
        This already implies that $\LCS_S(i,o)\ge i-\ell+1$.
        Assume to the contrary that $\LCS_S(i,o)> i-\ell+1$, we therefore have that $S[\ell-1]=S[d-1]$ which implies $\LCP_S(\ell-1,d-1) = \LCP_S(\ell,d) + 1$.
        This in turn implies $\LPF_S(\ell-1) \ge \LPF_S(\ell)+1$ which due to \cref{lem:19} yields $\LPF_S(\ell-1) = \LPF_S(\ell)+1$.
        We therefore have that $\ell-1 + \LPF_S(\ell-1) = \ell + \LPF_S(\ell)$.
        This is a contradiction to the minimality of $\ell$.
    \end{claimproof}

We proceed to prove that the point $p_{i,o}=(\LCS_S(i,o),\LCP_S(i,o))$ is a non-dominated point in $\ext_{\Occurrences{S}{M_L},i}$.
Assume to the contrary that there is an occurrence $o'\in \Occurrences{S}{M_L}$ such that $p_{i,o'}\in \ext_{\Occurrences{S}{M_L},i} \sm \{p_{i,o}\}$ such that $\LCP_S(i,o') \ge \LCP_S(i,o)$ and $\LCS_S(i,o') \ge \LCS_{S}(i,o)$.
Notice that the following equality holds: $S[o'-\LCS_S(i,o')+1 .. o' + \LCP_S(i,o')-1] = S[i-\LCS_S(i,o')+1 .. i + \LCP_S(i,o')-1]$.
We consider two cases.
\begin{enumerate}
    \item If $\LCP_S(i,o')>\LCP_S(i,o)$ and $\LCS_S(i,o')\ge \LCS_S(i,o)$ then we have $S[\ell .. \ell+\LPF_S(\ell)+1) = S[o'-(i-\ell)..o' + \LPF_S(\ell)-(i-\ell)+1)$.
    This is a contradiction to the maximality of $\LPF_S(\ell)$
    \item If $\LCP_S(i,o')\ge\LCP_S(i,o)$ and $\LCS_S(i,o')> \LCS_S(i,o)$ then we have $S[\ell-1 .. \ell+\LPF_S(\ell)) = S[o'-(i-\ell)-1..o' + \LPF_S(\ell)-(i-\ell))$.
    This indicates that $\ell-1 + \LPF_S(\ell-1) \ge \ell +\LPF_S(\ell)$, a contradiction to $\ell+\LPF_S(\ell)>\ell-1+\LPF_S(\ell-1)$.
\end{enumerate}

Thus, we show that $o\in \Occurrences{S}{M_L}$ and that $p_{i,o}$ is a non-dominated point in $\ext_{\Occurrences{S}{M_L},i}\cap [1..i)$.
Thus, by definition  we have $p_{i,o}\in \NDExt{\Occurrences{S}{M_L}}{i}{S}$.
By \cref{clm:asmuchtotheleft}, we have $\LCS_S(i,o)=i-\ell+1$.
Since $p_{i,o}=(\LCS_S(i,o),\LCP_S(i,o))$, we have that the first coordinate of $p_{i,o}$ is $i-\ell+1$.
Therefore, $p_{i,o}$ is a unique non-dominated point in $\ext_{\Occurrences{S}{M_L},i}\cap[1..i)$ corresponding to $\ell$.
Thus, we have shown $|L_S(i,j)|\le |\NDExt{\Occurrences{S}{M_L}}{i}{S}|+1$ (where the plus one is due to not considering $i$ as a possible value of $\ell$).
\end{proof}

\subsection{The General Case Algorithm} \label{sec:periodic}
We show how to update all $L$-heavy indices without the assumption that $M_L$ is aperiodic.
The algorithm for $R$-heavy indices is symmetric.
\para{The algorithm.}
Let $p$ be the period of $M_L$.
The algorithm uses \cref{lem:ipmalllengths} to obtain the clusters $\mathcal C_1=\Clusters_S(M_L)$ and  $\mathcal C_2=\Clusters_{S'}(M_L)$ (representing all occurrences of $M_L$ in $S$ and $S'$, respectively).
Let $\mathcal{C}=\mathcal{C}_1\cup \mathcal{C}_2$.
Notice that all clusters in $\mathcal{C}$ are maximal,
meaning that for every cluster $C=(a,b,p)$ of $S$ (resp.\ of $S'$), we have $a-p,b+p\notin \Occurrences{S}{M_L}$ (resp.\ $\notin \Occurrences{S'}{M_L}$).

The algorithm iterates over all clusters $C\in\C$.
Let $C=(a_C,b_C,p)$ be a cluster from $\Clusters_{\hat S}(M_L)$  for some $\hat S\in\{S,S'\}$ and let $R_C=\hat S[s_C,e_C]$ be the run in $\hat S$ containing all occurrences of $M_L$ implied by $C$.
The algorithm finds $e_C$ by applying $\LCP_{\hat S}(a_c, a_c +p)$ using \cref{lem:dympillar}.
The algorithm computes for $C$ the following sequences for every $T \in\{S,S'\}$:

\begin{enumerate}
    \item $L_{T,C}^1=L_{T}(a_C,a_C+m-1)=L_{T}(i_C^1,j_C^1)$.
    \item $L_{T,C}^2=L_{T}(a_C+p,a_C+p+m-1)=L_{T}(i_C^2,j_C^2)$.
    \item $L_{T,C}^e=L_{T}(b_C,e_C)=L_{T}(i_C^e,j_C^e)$.
\end{enumerate}

For every $x\in \{1,2,e\}$ such that $i_C^x\in \Occurrences{S}{M_L}\cap \Occurrences{S'}{M_L}$ the algorithm executes the following:
Let $q=\max\{\min (L_{S,C}^x),\min (L_{S',C}^x)\}$.
Sort $(L_{S,C}^x \cup L_{S',C}^x)\cap[q..i_C^x]$ in increasing order to obtain a sequence of positions $P=(a_1,a_2,\dots,a_t)\subseteq[q..i_C^x]$.
Notice that both $q\in P$ and $i_C^x\in P$.
The sequence $P$ induces a partition of $[q..i_C^x]$ to $R= [a_1,a_2,\dots,a_t]$.
Notice that $R$ is clean, since every $[a..b]\in R$ is contained in some interval of $R_S(i_C^x,j_C^x)$ and in some interval of $R_{S'}(i_C^x,j_C^x)$.
For every $[a..b]\in R$, the algorithm applies $\MoveInterval(a + \LPF_{S'}(a), [a,b])$ to $U$.
Notice that this operation is valid since $R$ is clean (in particular, $S$-clean).

\para{Correctness.}
We need to show that every $L$-heavy active index $v$, is updated.
We first prove that if $v$ is $L$-heavy, then there is an occurrence $k$ of $M_L$ that is a witness for $v$ being $L$-heavy such that $k\in \{a_C, a_C+p , b_C\}$ in some cluster $C$.

\begin{lemma}\label{lem:heavy_i'}
    Let $v$ be an  $L$-heavy index.
    Then, there is a cluster $C \in \C$ and $k\in \{a_C,a_C+p,b_C\}$ such that $k\in \Occurrences{S}{M_L}\cap \Occurrences{S'}{M_L}$ and $k\in[v..v+ \min(\LPF_{S'}(v),\LPF_{S}(v))-m]$.
\end{lemma}

\begin{proof}
    Since $v$ is $L$-heavy, there is an occurrence of $M_L$ in both in $S$ and in $S'$ at some index $\hat k$ with $\hat k\in [v..v+ \min(\LPF_{S}(v),\LPF_{S'}(v))-m]$.
    $\hat k$ must belong to some cluster $C=(a_C,b_C,p)$ in $S$ and a cluster $C'=(a_{C'},b_{C'},p)$ in $S'$.
    If $\hat k\in \{a_C,a_{C'},a_C+p,a_{C'}+p\}$, we are done.
    Otherwise, both $C$ and $C'$ are of size at least 3.
    Let $\hat a = \max(a_C,a_{C'})$.
    Notice that since $\hat k \in C \cap C'$ it must be the case that $\hat a \in C \cap C'$.
    Since $\hat k$ is at least the third occurrence in both clusters, we have $\hat{a}+p< \hat k \le  v+ \min(\LPF_{S'}(v),\LPF_{S}(v)) - m$.
    We consider two cases.
    \para{Case 1: $\hat a+p \ge v$.}
    In this case, we have that  $\hat a +p\in [v..v+ \min(\LPF_{S'}(v),\LPF_{S}(v))-m]$.
    Since $\hat a, \hat k \in C \cap C'$,  and $\hat a+p \in [\hat a ..\hat k]$, we have that $\hat a +p \in C \cap C'$.
    In particular, $\hat a + p \in \Occurrences{S}{M_L}\cap \Occurrences{S'}{M_L}$, which concludes this case.

    \para{Case 2: $v> \hat a+p$ .}
    Recall that $R_{C} = S[s_C..e_C]$ (resp.\ $R_{C'}=S[s_{C'}..e_{C'}]$) is the run with period $p$ containing all occurrences of $M_L$ in $S$ (resp.\ $S'$) represented by $C$ (resp.\ $C'$).
    Let $\hat s = \max(s_C,s_{C'})$ and $\hat e=\min (e_C,e_{C'})$.
    Notice that for all $\hat S\in \{ S,S'\}$, the string $\hat S[\hat s..\hat e]$ is periodic with period $p$.
    Due to $v > \hat a+p \ge \hat s+p$, it holds that $\hat S[v..\hat e]=\hat S[v-p..\hat e-p]$.
    It follows that $v + \LCP_{\hat S}(v,v-p)> \hat e$ and in particular $\hat e<  v+ \min(\LPF_{S'}(v),\LPF_{S}(v))$.
    Assume $\hat e = e_C$ (the case in which $\hat e = e_{C'}$ is similar).
    The run $R_C$ contains the occurrence of $M_L$ at $b_C$, so $\hat e \ge b_C+m-1$.
    Since $b_C$ is the last occurrence in a cluster containing $\hat k$, we have $b_C \ge \hat k \ge v$.
    We claim that $b_C$ is an occurrence of $M_L$ in $S'$ as well.
    Assume to the contrary $b_C \notin C'$.
    Since $\hat k <b_C$ is in $C'$ and $\hat k = b_C \bmod p$, we have that $b_{C'} \le b_C -p$.
    Since $S'[s_{C'} .. \hat e]$ is periodic with period $p$, and since $\hat e \ge b_C +m - 1$, we have that $S[b_{C'} .. b_{C'}+m-1] = S[b_{C'} + p.. b_{C'} + m -1+p]$.
    We have shown that $b_{C'} + p$ is an occurrence of $M_L$ in $S'$, which is a contradiction to the maximality of $C'$.
\end{proof}

\begin{lemma}\label{lem:update_i_periodic}
    Let $v$ be an $L$-heavy index.
    The algorithm sets the parent of $v$ in $U$ to $v + \LPF_{S'}(v)$.
\end{lemma}
\begin{proof}
    Since $v$ is $L$-heavy, by \cref{lem:heavy_i'} there is an occurrence $k \in \Occurrences{S}{M_L} \cap \Occurrences{S'}{M_L}$ such that
    $k \in [v..v+ \min(\LPF_{S}(v),\LPF_{S'}(v)) - m]$ and for some cluster $\hat C\in \C$ it holds that $k \in \{ a_{\hat C},a_{\hat C} + p , b_{\hat C}\}$.
    Let $C$ and $C'$ be the clusters in $\Clusters_S(M_L)$ and $\Clusters_{S
    }(M_L)$  containing $k$, respectively.
    We consider two cases.

    \begin{enumerate}
    \item{$k\in \{a_C,a_C + p, a_{C'},a_{C'}+p\}$.} We prove the case in which $k=a_C$.
    The remaining cases are identical.
    Let $i=i^1_C$ and $j=j^1_C$.
    Let $q=\max(\min(L_{S}(i,j)),\min(L_{S'}(i,j)))$.
    We prove that $v\in [q..i]$.
    Let $a=\min(L_{S}(i,j))$.
    Recall that $a$ is the leftmost index with $a+\LPF_S(a) > j$.
    Since $v+\LPF_S(v) > k + m-1 =  i + m -1= j $, it follows from the minimality of $a$ that $v\ge a=\min(L_{S}(i,j))$.
    Similarly, we have that $v\ge \min(L_{S'}(i,j))$.
    Thus, $v\ge q$ as we claimed.
    We also have $v\le k = i$.
    It follows that when the algorithm processes $S[i..j]=S'[i..j]$ as an occurrence of $M_L$, it generates the sequences $L^1_{S,C}$ and $L^1_{S',C}$ such that $v \in [q..i]$.
    Therefore, the algorithm creates a partition $R$ of $[q..i]$ and applies $\MoveInterval(a + \LPF_{S'}(a) , [a,b])$ for every interval $[a,b] \in R$, in particular to the interval $[a_v,b_v]$ containing $v$.
    Since all intervals in $R$ are $S'$-clean, we have that $v + \LPF_{S'}(v) = a_v +\LPF_{S'}(a_v)$, and the parent of $v$ is set as required.

    \item{$k \in \{b_C,b_{C'}\}$.}
    We prove for $k = b_C$, the proof for $k = b_{C'}$ is symmetric.
    Let $\hat a = \max \{ a_C,a_{C'}\}$ and notice that $\hat a \in \Occurrences{S}{M_L}\cap \Occurrences{S'}{M_L}$.
    If some $\hat k \in \{ \hat a, \hat a + p \}$ satisfies $\hat k \in [v.. v+\min(\LPF_S(v),\LPF_{S'}(v)) -m]$ the proof follows from the previous case.
    We assume that $\hat a + p \notin [v..v+\min (\LPF_S(v),\LPF_{S'}(v)) -m]$, which can only happen if $v > \hat a + p$ (since $\hat a + p < \min \{b_C,b_{C'} \} \le k \le v+ \min (\LPF_S(v),\LPF_{S'}(v))-m$).
    Recall that $R_C = S[s_C..e_C]$ and $R_{C'}=S'[s_{C'}..e_{C'}]$ are the runs in $S$ and in $S'$ that contain exactly all occurrences of $M_L$ implied by $C$ and by $C'$, respectively.
    Since $v > \hat a + p \ge \max (s_C,s_{C'}) + p$ and due to the periodicity of $R_C$ and $R_{C'}$ we have $S[v..e_C] = S[v - p .. e_C - p]$ and $S'[v..e_{C'}] = S'[v-p..v+e_{C'}-p]$.
    It follows that $v+\min(\LPF_S(v),\LPF_{S'}(v)) > \min(e_C,e_{C'})$.
    Let $\tilde C \in \{C,C'\}$ be the cluster which satisfies $e_{\tilde C} = \min ( e_C,e_{C'})$.
    When the algorithm processes the cluster $\tilde C$, it generates the sequences $L^e_{S,\tilde C} = L_S(b_{\tilde C}, e_{\tilde C})$ and $L^e_{S',\tilde C}=L_{S'}(b_{\tilde C},e_{\tilde C})$.
    Then, the algorithm creates a partition $R$ of the interval $[q..b_{\tilde C}]$ with $q = \max (\min(L^e_{S,\tilde C}) , \min(L^e_{S',\tilde C}))$.
    We claim that $v \in [q..b_{\tilde C}]$.
    $v \le b_{\tilde C}$ arises from $v \le k \le b_{\tilde C}$ as $b_{\tilde C}$ is the last occurrence in the cluster $\tilde C$ containing $k$.
    We proceed to show $v \ge q $.
    Recall that $q_{\min} = \min(L^e_{S,\tilde C})=\min(L_S(b_{\tilde C},e_{\tilde C}))$ is the minimal index such that $q_{\min} + \LPF_S(q_1) > e_{\tilde C}$.
    We have previously shown $v + \LPF_S(v)>e_{\tilde C}$, which indicates that $v \ge q_{\min}$ due to the minimality of $q_{\min}$.
    We have thus proved $v \in [q..b_{\tilde C}]$.
    The algorithm processes every interval in $R$, and in particular some interval $[a,b] \in R$ that contains $v$, and applies $\MoveInterval(a+\LPF_{S'}(a),[a,b])$.
    Since $[a,b]$ is $S'$-clean, $a + \LPF_{S'}(a) = v + \LPF_{S'}(v)$, so the parent of $v$ is set correctly, which concludes the case.\qedhere

    \end{enumerate}
\end{proof}

\para{Time complexity.}
The algorithm finds $\C_1 = \Clusters_S(M_L)$ and $\C_2 = \Clusters_{S'}(M_L)$ using \cref{lem:ipmalllengths} in $\Otild(n/m)$ time.
Then, the algorithm computes for each cluster $C \in \C = \C_1 \cup \C_2$ the sequences $L^x_{T,C}$ for every $T \in \{ S, S' \}$ and $x \in \{1,2,e\}$ using \cref{lem:fastintervalscomputation} in $\Otild(|L^x_{T,C}|)$ time.
For each $x \in \{1,2,e\}$, the algorithm then sorts $L^x_{S,C}$ and $L^x_{S',C}$ and applies a constant number of $\LPF$ queries, and a constant number of tree operations on $U$ per element in the combined list.

It follows that the time complexity is $\Otild(n/m + \sum_{x\in \{1,2,e\}} \sum _{T \in \{ S,S'\}} \sum _{C\in \C}|L^x_{T,C}|)$.
We show that the second expression is bounded by $\Otild(n/m)$.
We start by bounding $\sum_{C\in \mathcal C}|L^e_{S,C}|$.
Note that for every cluster $C$ such that $z\notin [s_C-1..e_C+1]$, both $S[s_C..e_C]$ and $S'[s_C..e_C]$ are runs.
At most two runs with period $p$ can intersect an index (see \cref{fact:index_coverd_two_pruns}), so for all but $O(1)$ clusters in each string (two touching each of $z-1$,$z$, and $z+1$), the endpoints of $R_C$ enclose a run both in $S$ and in $S'$.
Let $\mathcal{C}'$ be the set $O(1)$ clusters such that $z\in [s_C-1..e_C+1]$ for $C\in \mathcal C '$.
In the rest of this section, we bound the expression $\sum_{x\in \{1,2,e\}} \sum _{T \in \{ S,S'\}} \sum _{C\in \C}|L^x_{T,C}|$ in three parts.
In particular, we decompose this expression into three sub-sums as follows.

\begin{align}
    \sum_{x\in \{1,2,e\}} \sum _{T \in \{ S,S'\}} \sum _{C\in \C}|L^x_{T,C}|=&\nonumber\\\label{eq:subsums}
     \underbrace{ \sum_{x\in \{ 1,2\}} \sum_{T\in \{ S,S'\}}\sum_{C\in \C \sm \C'}|L^x_{T,C}|}_{{\text{\cref{lem:firsts}}}} &+
    \underbrace{\sum_{T\in \{ S,S'\}}\sum_{C\in \C \sm \C'}|L^e_{T,C}|}_{{\text{\cref{lem:lasts}}}} +
     \underbrace{\sum_{x\in \{ 1,2,e\}}\sum_{T \in \{ S,S'\}}\sum_{C \in \C'} |L^x_{T,C}|}_{{\text{\cref{lem:annoyinguys}}}}
\end{align}

In the above, the Lemma specified beneath each sub-sum provides a bound of $\Otild(n/m)$ for the sub-sum.
The task of proving that the running time of the algorithm is $\Otild(n/m)$ is hence reduced to proving these three lemmas.
Together, this concludes the proof of \cref{lem:heavytime}.

We start providing a bound on the second term of \cref{eq:subsums}.
\begin{lemma}\label{lem:lasts}
    $\sum_{C \in \C\sm\C'}|L_{S,C}^e|+|L_{S',C}^e|\in \Otild(n/m)$.
\end{lemma}
\begin{proof}
    We bound the first part $\sum_{C \in \C\sm\C'}|L_{S,C}^e|$, the bound of $\sum_{C \in \C\sm\C'}|L_{S',C}^e|$ is symmetric.
    For every cluster $C\in \C_1$, let $P_{C}=S[b_C..e_C+1]$ be the string beginning in the last occurrence of $M_L$ in $C$ and ending in the first character after $R_C$ (to incorporate the case where $e_C=n$, we assume that $S[n+1] = \$$ for some fresh symbol $\$$).
    By \cref{fact:break_period}, $P_C$ is aperiodic.
    Moreover, let $\mathcal{P}=\{P_{C}\mid C\in \C_1\}$.
    Let $B=\{b_C\mid C\in \C_1\}$.
    For each $P\in \mathcal{P}$, we define $B_{P}= \Occurrences{S}{P}$.
    Notice that  $B\subseteq \cup_{P\in\mathcal{P}}B_P$ trivially.
    In addition, for every $b\in B_{P}$ we have $S[b..b+|P|-1)$ is periodic and therefore contained in some $R_C$ and $S[b..b+m)$ must be the last occurrence of $M_L$ in $R_C$ since $S[b..b+|P|)$ is aperiodic.
    Therefore $B=\cup_{P\in\mathcal{P}}B_P$.

   For a cluster $C$, let $q=|L_{S,C}^e|$.
    By the construction of $L_{S,C}^e$, the element $\ell_{b_C,q-1}$ is the leftmost index in $L_{S,C}^e$, and $\ell_{b_C,q-2}$ is the leftmost index in $L_{S,C}^e\setminus \{\ell_{b_C,q-1}\}$.
    By the construction of $L_{S,C}^e$, it holds that  $\ell_{b_C,q-2}+\LPF_S(\ell_{b_C,q-2})>e_C$+1.
    Denote $L_3=L_S(b_C,e_C+1)$.
    By the construction of $L_{S,C}^e$,
    \[
    |L_{S,C}^e\setminus \{\ell_{b_C,q-1}\}|
    =|L_S(b_C,e_C)|-1
    \le
    |L_S(b_C,e_C+1)|
    =|L_3|
    \]
    By \cref{lem:bound_Li}, $|L_3|\le |\NDExt{B_{P_C}}{b_C}{S}|+1$.
    Therefore, for a given cluster $C$, it holds that $|L_{S,C}^e|\le |\NDExt{B_{P_{C}}}{b_C}{S}|+2$.
    To conclude,

\begin{align*}
\sum_{C\in \C\sm\C'}|L_{S,C}^e|&\le\sum_{C\in \mathcal{C}_1}|L_{S,C}^e|&\text{since $\C\sm\C'\subseteq\C_1$}
\\&\le
    \sum_{b_C\in B}(|\NDExt{B_{P_C}}{b_C}{S}|+2)&\text{}
    \\&=2|\mathcal{C}_1| + \sum_{b_C\in B}|\NDExt{B_{P_C}}{b_C}{S}| &\text{}
    \\&= 2|\mathcal{C}_1| +
    \sum_{P\in\mathcal{P}} \sum_{b_C\in B_P}|\NDExt{B_{P}}{b_C}{S}|&\text{$\{B_P\}_{P\in \mathcal{P}}$ is a partition of $B$}
    \\&= 2|\mathcal{C}_1| + \sum_{P\in\mathcal{P}}\Otild(|B_P|)&\text{by \cref{lem:chain}}
    \\&= 2|\C_1|+\Otild (|B|)=2|\C_1|+\Otild (|\C_1|)&\text{$\{B_P\}_{P\in \mathcal{P}}$ is a partition of $B$}
    \\&=\Otild (|\C_1|)=\Otild(n/m)&\text{By \cref{fact:per_structure}}.&\qedhere
\end{align*}
\end{proof}

We proceed to bound the first term of \cref{eq:subsums}.

\begin{lemma}\label{lem:firsts}
    $\sum_{C\in \mathcal{C}\setminus\mathcal{C}'}|L_{S,C}^1|+|L_{S,C}^2|+|L_{S',C}^1|+|L_{S',C}^2|\in \Otild(n/m)$.
\end{lemma}

\begin{proof}
    We prove the lemma for $\sum_{C\in \mathcal{C}\setminus\mathcal{C}'}|L_{S,C}^2|$.
    The proofs for the other cases are symmetric.
    Let $C\in \C \setminus \C'$ be a cluster.
    Since $C\notin \C'$, the endpoints of the run $R_C = S[s_C..e_C]$ are the same in $S$ and in $S'$.
    We emphasize that $a_C$ and $a_C+p$ are the first two occurrences of $M_L$ within $R_C$ (even if $C$ is a cluster of occurrences in $S'$).
    We distinguish between three types of $\ell_k\in L_{S,C}^2$.
    \begin{enumerate}
        \item Let $L_{C,\alpha}=\{\ell_k\in L_{S,C}^2\mid \ell_k<s_C\}$.
        \item Let $L_{C,\beta}=\{\ell_k\in L_{S,C}^2\mid s_C< \ell_k<\ell_0\}$.
         \item $\ell_k\in\{s_C,\ell_0\}$.
    \end{enumerate}
    We will show that $\sum_{C\in{\C\sm\C'}}|L_{C,\alpha}| \in \Otild(n/m)$ (in \cref{clm:before_sC}) and $\sum_{C\in{\C\sm\C'}}|L_{C,\beta}| \in \Otild(n/m)$ (in \cref{clm:after_sC}).
    Clearly, $\sum_{C\in{\C\sm\C'}}|\{s_C,\ell_0\}|=\sum_{C\in{\C\sm\C'}}2=|\C\sm\C'|\in \Otild(n/m)$.
    Therefore, we simply obtain $\sum_{C\in \C\sm\C'}|L_{S,C}^2|\le \sum_{C\in{\C\sm\C'}}\left(|L_{C,\alpha}|+|L_{C,\beta}|+2\right)=\Otild(n/m)$, as required.

    We prove the following claim (in a similar way to the proof of \cref{lem:lasts}).
    \begin{claim}\label{clm:before_sC}
        $ \sum_{C\in{\C\sm\C'}}|L_{C,\alpha}| \in \Otild(n/m).$
    \end{claim}
    \begin{claimproof}
    First, notice that $a_C+p=\ell_0\notin L_{C,\alpha}$.
    For every cluster $C\in \C_1$, let $P_{C}=S[s_C-1..a_C+m-1]$ be the string beginning in the last character before $R_C$ (to incorporate the case where $s_C=1$, we assume that $S[0] = \$$ for some fresh symbol $\$$) and ending in the last character of the first occurrence of $M_L$ in $C$.
    By \cref{fact:break_period}, $P_C$ is aperiodic.
    Moreover, let $\mathcal{P}=\{P_{C}\mid C\in \C_1\}$.
    Let $F=\{s_C-1\mid C\in \C_1\}$.
    For each $P\in \mathcal{P}$, we define $F_{P}= \Occurrences{S}{P}$.
    Notice that  $F\subseteq \cup_{P\in\mathcal{P}}F_P$ trivially.
    In addition, for every $f\in F_{P}$ we have $S[f+1..f+|P|)$ contains an occurrence of $M_L$ and has period $p$.
    Therefore, $S[f+1..f+|P|)$ is contained in some $R_{\hat C}$ for $\hat C\in \C_1$ and $S[f+|P|-m..f+|P|)$ must be the first occurrence of $M_L$ in $R_{\hat C}$ since $S[f..f+|P|)$ is aperiodic.
    Therefore $F=\cup_{P\in\mathcal{P}}F_P$.

    Let $\hat L_{C,\alpha}=L_S(s_C-1,a_C+m-1)$.
    Let $\ell_k\in L_{C,\alpha}$.
    It holds that $\ell_k+\LPF_S(\ell_k)> a_C+m-1$.
    By \cref{obs:critical_indices} we have $\ell_k+\LPF_S(\ell_k)>\ell_k-1+\LPF_S(\ell_k-1)$.
    Therefore, by \cref{obs:critical_indices} we have $\ell_k\in \hat L_{C,\alpha}$.
    Hence, $|L_{C,\alpha}|\le|\hat L_{C,\alpha}|$.

    By \cref{lem:bound_Li}, $|\hat L_{C,\alpha}|\le |\NDExt{F_{P_C}}{s_C-1}{S}|+1$.
    Therefore, for a given cluster $C\in \C_1$, it holds that $|L_{C,\alpha}|\le |\NDExt{F_{P_{C}}}{s_C-1}{S}|+2$.
    To conclude,
\begin{align*}
\sum_{C\in \C\sm\C'}|L_{C,\alpha}|&\le \sum_{C\in \mathcal{C}_1}|L_{C,\alpha}| &\text{since $\C\sm\C'\subseteq\C_1$}
\\&\le
    \sum_{s_C-1\in F}(|\NDExt{F_{P_C}}{s_C-1}{S}|+2)&\text{}
    \\&=2|\mathcal{C}_1| + \sum_{s_C-1\in F}|\NDExt{F_{P_C}}{s_C-1}{S}| &\text{}
    \\&= 2|\mathcal{C}_1| +
    \sum_{P\in\mathcal{P}} \sum_{s_C-1\in F_P}|\NDExt{F_{P}}{s_C-1}{S}|&\text{$\{F_P\}_{P\in \mathcal{P}}$ is a partition of $F$}
    \\&= 2|\mathcal{C}_1| + \sum_{P\in\mathcal{P}}\Otild(|F_P|)&\text{by \cref{lem:chain}}
    \\&= 2|\C_1|+\Otild (|F|)=2|\C_1|+\Otild (|\C_1|)&\text{$\{F_P\}_{P\in \mathcal{P}}$ is a partition of $F$}
    \\&=\Otild (|\C_1|)=\Otild(n/m)&\text{by \cref{fact:per_structure}}.
\end{align*}
\end{claimproof}

   We proceed to bound $L_{C,\beta}$.
    \begin{claim}\label{clm:after_sC}
        $ \sum_{C\in{\C\sm\C'}}|L_{C,\beta}| \in \Otild(n/m).$
    \end{claim}
    \begin{claimproof}
We denote $A=\{a_C,a_C+p\mid C\in\C\}$.
Let $C\in\C\sm\C'$, let $i=a_C+p$ and let $j=i+m-1$.
We first show that $|L_{C,\beta}|\le |\NDExt{A}{i}{S}|$.

Let $\ell\in L_{C,\beta}\setminus\{a_C+p\}$.
By \cref{obs:critical_indices} we have $\ell+\LPF_S(\ell)>\ell-1+\LPF_S(\ell-1)$.
In addition, $\ell+\LPF_S(\ell)>j$ and $\ell\in[s_C+1..a_C+p]$.

We denote $d=\LPFpos_S(\ell)$.
Due to the definition of $\LPF$, we have $S[\ell..\ell+\LPF_S(\ell))= S[d..d +\LPF_S(\ell))$ which in particular yields the equality $S[i..j]= S[d + (i-\ell)..d+(j-\ell)]$.
We showed that $o = d + (i-\ell) \in \Occurrences{S}{M_L}$.

We have to show that $o\in A$.
Assume to the contrary that $o\notin A$.
Let $R_o$ be the run with period $p$ containing $o$.
We have $o-d = i-\ell \le a_C+p - s_C - 1\le 2p-1$.
It follows from $o\notin A$ that $d -1 \ge o - 2p$ is within $R_o$ and therefore $S[d-1]=S[d-1+p]$.
From the definition of $\LPF$, we have $S[\ell..\ell+\LPF_S(\ell)) = S[d..d+\LPF_S(\ell))$.
Since $\ell > s_C$, the periodicity of $R_C$ yields $S[\ell-1] = S[\ell-1+p]$.
Using all equalities specified above, we have $S[d-1] = S[d-1 + p] = S[\ell-1+p]=S[\ell-1]$.
We have shown that $S[d-1..d+\LPF_S(\ell)) = S[\ell-1..\ell + \LPF_S(\ell))$.
This implies that $\ell-1 + \LPF_S(\ell-1) \ge \ell + \LPF_S(\ell)$, a contradiction.

We proceed to prove that the point $p_{i,o}=(\LCS_S(i,o),\LCP_S(i,o))$ is a non-dominated point in $\ext_{A,i}$.
Assume to the contrary that for some occurrence $o'\in A$ there is an occurrence $p_{i,o'}\in \ext_{A,i} \sm \{p_{i,o}\}$ such that $\LCP_S(i,o') \ge \LCP_S(i,o)$ and $\LCS_S(i,o') \ge \LCS_{S}(i,o)$.
Notice that the following equality holds $S[o'-\LCS_S(i,o')+1 .. o' + \LCP_S(i,o')-1] = S[i-\LCS_S(i,o')+1 .. i + \LCP_S(i,o')-1]$.
We consider two cases.
\begin{enumerate}
    \item If $\LCP_S(i,o')>\LCP_S(i,o)$ and $\LCS_S(i,o')\ge \LCS_S(i,o)$ then we have $S[\ell .. \ell+\LPF_S(\ell)+1) = S[o'-(i-\ell)..o' + \LPF_S(\ell)-(i-\ell)+1)$.
    This is a contradiction to the maximality of $\LPF_S(\ell)$
    \item If $\LCP_S(i,o')\ge\LCP_S(i,o)$ and $\LCS_S(i,o')> \LCS_S(i,o)$ then we have $S[\ell-1 .. \ell+\LPF_S(\ell)) = S[o'-(i-\ell)-1..o' + \LPF_S(\ell)-(i-\ell))$.
    This indicates that $\ell-1 + \LPF_S(\ell-1) \ge \ell +\LPF_S(\ell)$, a contradiction to $\ell+\LPF_S(\ell)>\ell-1+\LPF_S(\ell-1)$.
\end{enumerate}

Thus, we show that $o\in A$ and that $p_{i,o}$ is a non-dominated point in $\ext_{A,i}\cap [1..i)$.
Thus, by definition we have $p_{i,o}\in \NDExt{A}{i}{S}$.
By \cref{clm:asmuchtotheleft}, we have $\LCS_S(i,o)=i-\ell+1$.
Since $p_{i,o}=(\LCS_S(i,o),\LCP_S(i,o))$, we have that the first coordinate of $p_{i,o}$ is $i-\ell+1$.
Thus, $p_{i,o}$ is a unique point in $\ext_{A,i}\cap[1..i)$ corresponding to $\ell$.
Thus, we have shown $|L_{C,\beta}|\le |\NDExt{A}{i}{S}|$.

Thus, by \cref{lem:chain} we obtain \[\sum_{C\in \mathcal{C}\setminus\C'}|L_{C,\beta}|\le \sum_{C\in\C\sm\C'}|\NDExt{A}{a_C+p}{S}\le \sum_{a\in A}|\NDExt{A}{i}{S}|\in \Otild(n/m+|A|)= \Otild(n/m)\] as required.
\end{claimproof}
\end{proof}

It only remains to bound the third term of \cref{eq:subsums}.
Recall that there are $O(1)$ terms in this sub-sum.
We show that each is bounded by $\Otild(n/m)$.
\begin{lemma}\label{lem:annoyinguys}
    For every cluster $C\in \C$, for every $x\in \{1,2,e\}$ such that $i^x_{C}\in \Occurrences{S}{M_L}$, it holds that $|L^x_{S,C}| \in O(n/m)$
\end{lemma}
\begin{proof}
    Let $i \in \{a_C,a_C+p,b_C\}$ and let $j= i +m-1$.
    We will show that $|L_S(i,j)|=O(n/m)$.
    For $i \in \{a_C, a_C+p\}$, this is exactly $L^1_{S,C}$ and $L^2_{S,C}$.
    For $i = b_C$, we have that $L^e_{S,C} = L_S(b_C,e_C)$ with $e_C \ge b_C + m-1$.
    It follows immediately from the definition of $L_S(\cdot,\cdot)$ that $|L^e_{S,C}|\le |L_S(b_C,b_C+m-1)|$.

    Let $R_i=S[s_i..e_i]$ be the run with period $p$ containing the occurrence $S[i..j]$ of $M_L$ in $S$.
    We distinguish between three subsets of $L_S(i,j)$.
    \begin{enumerate}
        \item $L_1 = L_S(i,j) \cap [s_i+1..i-1]$
        \item $L_2 = L_S(i,j) \cap [1..s_i-1]$
        \item $L_S(i,j) \cap \{ i,s_i\}$
    \end{enumerate}
    Clearly, the third set is of constant size.
    We bound $L_1$ and $L_2$ separately.
    \begin{claim}
        $|L_1| \in O(n/m)$
    \end{claim}
    \begin{claimproof}
        Let $\ell \in L_1$ and let $d = \LPFpos_S(\ell)$.
        We start by showing that for some cluster $C_1\in \C_1$, the index $d$ is the first index of the run $R_{C_1}$ i.e., $d=s_{C_1}$.
        Firstly, let us show that $d \in [s_{C_1}..e_{C_1}]$ for some cluster $C_1 \in \C_1$.
        From the definition of $\LPF$, we have $S[\ell .. \ell+\LPF_S(\ell) ) = S[d..d+\LPF_S(\ell))$.
        Since $\ell > s_i$, and $S[\ell..\ell + \LPF_S(\ell))$ contains an occurrence of $M_L$, the string $S[d..d+\LPF_S(\ell))$ also contains an occurrence of $M_L$ and has a periodic prefix with period $p$ containing this occurrence of $M_L$.
        Therefore, it must be the case that $d\in [s_{C_1}..e_{C_1}]$.

        Assume to the contrary that $d > s_{C_1}$ and recall that $\ell>s_i$.
        Then, from the periodicity of $R_{C_1}$ we get $S[d-1]=S[d-1+p]$.
        From the periodicity of $R_i$, we also have $S[\ell-1]=S[\ell-1+p]$.
        Combining the above equalities, we obtain $S[\ell-1] = S[\ell-1+p] = S[d-1+p] = S[d-1]$.
        We have shown that $S[\ell-1 .. \ell+\LPF_S(\ell)) = S[d-1..d+\LPF_S(\ell))$.
        This implies $\ell-1 + \LPF_S(\ell-1) \ge \ell+ \LPF_S(\ell)$.
        This is a contradiction to \cref{obs:critical_indices}.

        Let $\ell \in L_1$ be an index with $\LPFpos_S(\ell) = s_{C_1}$ and $\ell + \LPF_S(\ell) > e_i + 1$.
        We claim that $ e_i -\ell=e_{C_1} -s_{C_1}$.
        This must be true, since the period $p$ breaks at $e_i+1$ after $\ell$ and at $e_{C_1}$ after $s_{C_1}$.
        Therefore, there is at most one $\ell \in L_1$ with $\LPFpos_S(\ell) = s_{C_1}$ and $\ell + \LPF_S(\ell) > e_i + 1$ (we have shown that $s_{C_1}$ decides the value of $\ell$).

        We further claim that there is at most one element $\ell \in L_1$ with $\LPFpos_S(\ell) = s_{C_1}$ such that $\ell + \LPF_S(\ell) \le e_i+1$.
        First, notice that for any $\ell$ with $\ell + \LPF_S(\ell) \le e_i +1$, it holds that $\ell \in [s_i+1..s_i+p]$.
        That is true since $S[s_i+p..e_i] = S[s_i..e_i-p]$ and therefore $s_i +p + \LPF_S(s_i + p) \ge e_i + 1$.
        Therefore, for every $\ell>s_i+p$ with $\ell+\LPF_S(\ell)\le e_i+1$ we have $\ell+\LPF_S(\ell)=\ell-1+\LPF_S(\ell-1)$, which by \cref{obs:critical_indices}, implies $\ell\notin L_{1}$.
        Assume to the contrary that there are two indices $\ell_1,\ell_2\in L_1 \cap[s_i+1..s_i+p]$ and $s_{C_1}=\LPFpos_S(\ell_1)=\LPFpos_S(\ell_2)$.
        Hence, $\LCP_S(\ell_1,s_{C_1}) \ge m$ and $\LCP_S(\ell_2,s_{C_1}) \ge m$.
        It follows that the string $m' = S[s_{C_1}..s_{C_1} +m-1]$, which is periodic with period $p$, occurs both in $\ell_1$ and in $\ell_2$.
        Since $|\ell_2 - \ell_1| < p$, this is a contradiction to \cref{fact:per_structure}.

        We have shown that for every $C\in C_1$, there are at most 2 elements $\ell \in L_1$ with $\LPFpos_S(\ell) = s_{C_1}$, and that every element in $L_1$ must have $\LPFpos_S(\ell) = s_{C'}$ for some $C' \in \C_1$.
        Due to $|\C_1| \in O(n/m)$, this concludes the proof.
    \end{claimproof}

    We proceed to bound $|L_2|$.
    \begin{claim}
        $|L_2| \in O(n/m)$
    \end{claim}
    \begin{claimproof}
        Consider the string $P= S[s_i-1..j]$.
        Due to \cref{fact:break_period}, $P$ is aperiodic.
        Clearly, $|L_2| \le |L_S(s_i-1..j)|$ from the definition of $L_S(\cdot,\cdot)$.
        By \cref{lem:bound_Li}, $|L_S(s_i-1,j)| \le |\NDExt{\Occurrences{S}{P}}{s_i-1}{j}|\le |\Occurrences{S}{P}|$.
        Since $P$ is aperiodic, we have $|\Occurrences{S}{P}| \le O(n/m)$, which concludes the proof.
    \end{claimproof}
    \end{proof}

\section{Lower Bound}\label{sec:lb}

In this section, we show that there is no data structure solving \cref{prb:ub} with polynomial preprocessing time and $O(n^{2/3-\eps})$ update and query times for any $\eps>0$, unless SETH~\cite{IP01,IPZ01} is false, thus proving \cref{thm:intro-lb}.
We first introduce an easier version of \cref{prb:ub}.

\begin{problem}{Dynamic $|\LZss(S)|$ under substitutions}\label{prb:lb}
\textbf{Preprocess$(S)$:} preprocess a string $S$ of length $N$ and return $|\LZss(S)|$.
\medskip
\\
\textbf{Update$(i,\sigma)$:} apply $S[i]\gets \sigma$ and return $|\LZss(S)|$.
\end{problem}
Clearly, the following lower bound implies \cref{thm:intro-lb}.
\begin{lemma}\label{thm:lb}
For every $c,\eps>0$, there is no data structure that solves \cref{prb:lb} with $O(N^c)$ preprocessing time and $O(N^{2/3-\eps})$ update time, unless SETH is false.
\end{lemma}
We note that the queries of \cref{prb:ub} can be used to compute $|\LZss(S)|$ within $\tilde{O}(1)$ queries.
The number of phrases can be obtained directly using $\ContainingPhrase_S(n)$.
Alternatively, one can perform a binary search on $\SelectPhrase_S(i)$ to identify the last phrase.

The rest of this section is dedicated for the proof of \cref{thm:lb}.
We build upon the lower bound of Abboud and Vassilevska Williams~\cite{AVW21} for the Orthogonal Vectors ($\OV$) problem in the following setting.
A set $A$ of vectors is given for preprocessing, and then a set $B$ is queried, the goal is to decide if there are vectors $\ourvec{v}\in A$ and $\ourvec{u}\in B$ such that $\ourvec{v}\cdot \ourvec{u} = 0$, where $\ourvec{v} \cdot \ourvec{u}= \bigvee_{j=1}^d(\ourvec{v}[j] \wedge \ourvec{u}[j]) $ denotes the boolean inner product of $\ourvec{u}$ and $\ourvec{v}$.
The lower bound of \cite{AVW21} is stated as follows.
\begin{lemma}[{\cite[{cf. Theorem 13}]{AVW21}}]\label{lem:OVA}
There is no algorithm that given two sets $A$ and $B$ of $n$ boolean vectors in
$d(n) = \log^{\omega(1)} n$ dimensions, preprocesses the set $A$ in $O(n^c)$ time, and subsequently solves $\OV$ on $A$ and $B$ in $O(n^{2-\eps})$ time, for some $c,\eps > 0$, unless SETH is false.
\end{lemma}
We will show how to solve the problem of \cref{lem:OVA} using a data structure of \cref{prb:lb} with $N=O(d\cdot n^{1.5})$ and $O(d\cdot n)$ updates.

\para{Overview.}
Let $A$ be a set of $n$ Boolean $d = d(n)$-dimensional vectors.
The string $S$ consists of three components, which we refer to as the dictionary, the vector, and the matrix.
The dictionary is defined as a function of $n$, independent of the vectors in $A$ or $B$.
The matrix represents all vectors in $A$, while the vector component at any given time represents a single $\ourvec{u} \in B$.
Our goal is to encode $A$ and some $\ourvec{u} \in B$ such that if all vectors in $A$ are non-orthogonal to $\ourvec{u}$, then $|\LZss(S)|$ attains a predetermined value.
Conversely, if there exists at least one $\ourvec{v} \in A$ such that $\ourvec{u} \cdot \ourvec{v} = 0$, then $|\LZss(S)|$ is strictly smaller.
We introduce a representation that encodes each $\ourvec{v} \in A$ separately.
The number of phrases in $\LZss(S)$ contributed by the encoding of $\ourvec{v}$ is $d+1$ if $\ourvec{u} \cdot \ourvec{v} = 0$ and $d+2$ otherwise.
(In particular, this allows the algorithm not only to detect the existence of an orthogonal pair but also to count the number of such pairs.)
Finally, transitioning the representation from one vector in $B$ to another requires at most $O(d)$ updates.
Thus, iterating over all vectors in $B$ involves at most $O(d \cdot n)$ updates.

\para{The reduction.}

The alphabet we use contains the symbols $0_j$, $1_j$, and $2_j$ for every $j \in [d]$. In addition, it includes a large number of unique characters, which we denote by the symbol $\#$. Each occurrence of $\#$ refers to a distinct symbol that appears exactly once in the strings we define.
Without loss of generality, we assume that  $\sqrt n$ is an integer number.
We first explain how to represent the set $A$.
For any $i\in[1..n]$ let $(i_1,i_2)\in[0..\sqrt n]\times [0..\sqrt n-1]$ be the two unique integers such that $i=i_1\cdot \sqrt n+i_2$.
We start by encoding of a single entry in a vector of $A$.
Let $\ourvec{v}_1,\ourvec{v}_2,\dots. \ourvec{v}_n$ be the vectors of $A$ in some arbitrary order and let $\ourvec{v}_i[j]$ denote the $j$th entry of $\ourvec{v}_i$.
$\ourvec{v}_i[j]$ is represented by \[\ACoord(i,j)= \begin{cases}
        0_j{}^{3\sqrt n+i_1}2_j{}^{i_2} & \ourvec{v}_i[j]=0\\
        1_j{}^{3\sqrt n+i_1}2_j{}^{i_2}  & \ourvec{v}_i[j]=1
\end{cases}.\]

The representation of the $i$th vector of $A$ is obtained by concatenating the representation of all its coordinates by \[    \AVect(i)=\bigodot_{j=1}^d \ACoord(i,j).\]

Finally, to represent the set $A$ we concatenate the representations of all the vectors.
We use $\#$ to denote a unique character that appears only once in $S$ (every time we use $\#$ it means a completely new symbol).
Thus, \[\Mat(A)=\bigodot_{i=1}^n \left(\AVect(i)\#\right).\]

The dictionary is made up of two types of gadgets.
The types, called \textit{skip-$0$} and \textit{skip-two-halves} gadgets, defined for every $j\in[d], i_2\in[0..\sqrt n],(x,y)\in\{0,1\}^2$ as
\[\skipo(j)=0_j{}^{4\sqrt n}2_j{}^{\sqrt n}\]
\[\skiphalves(j,i_2,x,y)=x_j{}^{2\sqrt n}2_j{}^{i_2}y_{j+1}{}^{2\sqrt n}.\]
We define the dictionary as
\[ \Dict=\left(\bigodot_{j=1}^d\skipo(j)\#\right)\cdot \left(
\bigodot_{j=1}^d\bigodot_{i_2=0}^{\sqrt n}\bigodot_{(x,y)\in\{0,1\}^2}\skiphalves(j,i_2,x,y)\#
\right).\]
A coordinate of a vector $\ourvec{u}\in B$ is encoded by \[\BCoord(\ourvec{u},j)= \begin{cases}
        1_j{}^{4\sqrt n+1}2_j{}^{\sqrt n} & \ourvec{u}[j]=0\\
        1_j{}^{2\sqrt n}{\#}1_j{}^{2\sqrt n}2_j{}^{\sqrt n}  & \ourvec{u}[j]=1
\end{cases}.\]

We emphasize that the encoding of a 0-coordinate in $\ourvec{u}$ by $\BCoord(\ourvec{u},\cdot)$ does not include the symbol $0_j$, which slightly differs from the intuition presented in \cref{sec:overview}. This design choice is motivated by the need to allow the encoding of a 0-coordinate to be obtained from the encoding of a 1-coordinate by only a few modifications.

Thus, a vector $\ourvec{u}\in B$ is encoded by \[\BVect(\ourvec{u})=\bigodot_{j=1}^d \BCoord(\ourvec{u},j)\#.\]

As described above, the dynamic string $S(\ourvec{u})$ when testing a vector $\ourvec{u}\in B$ is a concatenation of the three main gadgets: \[S(\ourvec{u})=\Dict\cdot \BVect(\ourvec{u})\cdot \Mat(A).\]

\subsection{Analysis}
In our analysis, we make use of the concept of $\LZss$-like factorization.
\begin{definition}[cf. \cite{LZ76}]
    A factorization $F$ of a string $S$ into disjoint substrings is $\LZss$-like if every substring in $F$ is either a single character, or a substring that is not a left-most occurrence of itself in $S$.
\end{definition}

We extensively use the following lemma by Lempel and Ziv~\cite{LZ76}, that says that the \LZss{} factorization is optimal.
\begin{fact}[{\cite[{Theorem 1}]{LZ76}}]\label{thm:LZ_optimal}
For any strings $S$, the \LZss{} factorization of $S$ has the smallest number of phrases among all \LZss-like factorizations of $S$.
\end{fact}

We apply \cref{thm:LZ_optimal} in the following manner.
Let $S = X \# G \# Y$, where each occurrence of $\#$ denotes a unique symbol.
Due to the uniqueness of each $\#$ in $S$, an \LZss{} phrase starts both before and after $G$ in $\LZss(S)$.
Consider the phrases in $\LZss(S)$ covering $G$, assume there are $z$ such phrases.
By \cref{thm:LZ_optimal}, any partition of $G$ into $\LZss$-like phrases uses at least $z$ phrases (as otherwise, one can replace the phrases of $\LZss(S)$ covering $G$ with  $z'<z$ phrases, obtaining a strictly smaller $\LZss$-like factorization of $S$).
We conclude this discussion with the following.

\begin{corollary}\label{lem:LZ_optimal}
    Let $G$ be a substring surrounded by unique $\#$ symbols in some string $S$, i.e., $S= X \# G \# Y$.
    Let $z$ be the number of $\LZss$-phrases in $\LZss(S)$ that cover $G\#$, and let $z'$ be the number of factors covering $G\#$ in an $\LZss$-like factorization of $S$.
    Then it must hold that $z \le z'$.
\end{corollary}

Let $\ourvec{u}\in B$.
For every $i \in [0..n]$ we define $\Mat_i(A)=\bigodot_{k=1}^i\left(\AVect(k)\#\right)$ is the prefix of $\Mat(A)$ until the end of the representation of the $i$th vector (including the $\#$).
Moreover, let $S_i(\ourvec{u})=\Dict\cdot\BVect(\ourvec{u})\cdot\Mat_i(A)$ be the prefix of $S(\ourvec{u})$ until the end of the representation of the $i$th vector (including the $\#$).

We define $\Phrases(S(\ourvec{u}),i)$ to be the number of phrases covering $\AVect(i)\#$  in $\LZss(S(\ourvec{u}))$.
Formally, $\Phrases(S(\ourvec{u}),i)=|\LZss(S_i(\ourvec{u}))|-|\LZss(S_{i-1}(\ourvec{u}))|$.
Notice that due to the usage of $\#$ separators, the phrases covering $\AVect(i)\#$ are disjoint from the phrases covering $\AVect(k)\#$  for every $k\ne i$.

\begin{lemma}\label{lem:phrases}
    For every $i\in [n]$ it holds that $\Phrases(S(\ourvec{u}),i)= d+1 + \ourvec{u}\cdot \ourvec{v}_i$.
\end{lemma}
\begin{proof}
    We start by showing that $\Phrases(S(\ourvec{u}),i)\le d+2$.
    \begin{claim}\label{clm:atmostd2}
        $\Phrases(S(\ourvec{u}),i)\le d+2$.
    \end{claim}
    \begin{claimproof}
        We introduce an $\LZss$-like factorization of $\AVect(i)\#$ into $d+2$ phrases, thus by \cref{lem:LZ_optimal} the claim follows.
        Let $P_1=(\ourvec{v}_i[1])_1{}^{2\sqrt n}$ (that is, $P_1= \begin{cases}
        0_1{}^{2\sqrt n} & \ourvec{v}_i[j]=0\\
        1_1{}^{2\sqrt n} & \ourvec{v}_i[j]=1
\end{cases}$).
        For every $j\in [2..d]$ let $P_j=(\ourvec{v}_i[j-1])_{j-1}{}^{\sqrt n+i_1}2_{j-1}{}^{i_2}(\ourvec{v}_i[j])_{j}{}^{2\sqrt n}$, and let $P_{d+1}=(\ourvec{v}_i[d])_{d}{}^{\sqrt n+i_1}2_d{}^{i_2}$ and $P_{d+2}=\#$.
        Notice that $P_1$ is a substring of $\skiphalves(1,i,\ourvec{v}_i[1],0)$ and for every $j\in [2..d+1]$ we have $P_j$ is a substring of $\skiphalves(j-1,i,\ourvec{v}_i[j-1],\ourvec{v}_i[j])$ and $\#$ is a single (fresh) character.
        Moreover, $\bigodot_{j=1}^{d+2}P_j=\AVect(i)\#$.
        Since all gadgets $\skiphalves$ appear in $S(\ourvec{u})$ before $\AVect(i)$, these are valid $\LZss$-like phrases.
        Therefore, $P_1,P_2,\dots,P_{d+2}$ is indeed a valid $\LZss$-like factorization of $\AVect(i)\#$ of length $d+2$.
    \end{claimproof}

    The following shows that every $\ACoord(i,j)$ gadget induces at least one unique phrase.
    \begin{claim}\label{clm:no-superstring}
        For any $j\in [1..d]$ there is no $\LZss$-like phrase of $S(\ourvec{u})$ covering part of $\AVect(i)\#$ that is a proper superstring of $\ACoord(i,j)$.
    \end{claim}
    \begin{claimproof}
        Notice that any proper superstring of $\ACoord(i,j)$ is a superstring of either $x_{j-1}\ACoord(i,j)$ or $\ACoord(i,j)x_{j+1}$ for $x\in\{0,1,2\}$ (the case where the following or preceding character is $\#$ is clearly impossible).
        We will prove the claim for superstrings of $\ACoord(i,j)x_{j+1}$, the proof for $x_{j-1}\ACoord(i,j)$ is symmetric.
        First, note that for every $x,y\in\{0,1,2\}$, $\ACoord(i,j)x_{j+1}$ is not a substring of $\BVect(\ourvec{u})$ or $\Dict$.
        This is because the only gadgets (in $\BVect(\ourvec{u})$ or $\Dict$) that contain both $y_j$ and $x_{j+1}$ characters are $\skiphalves$ gadgets and none of them contain more than $2\sqrt n$ occurrences of $(\ourvec{v}_i[j])_j$.

        Finally, for every $k<i$ the gadget $\ACoord(i,j)$ is not a substring of $\AVect(k)$.
        This is because the only substring of $\AVect(k)$ that contains characters $y_j$ for $y\in\{0,1,2\}$ is $\ACoord(k,j)$.
        Since $k<i$ it must be that either $k_1<i_1$ or $k_2<i_2$.
        Therefore, $\ACoord(k,j)$ either does not contain enough $(\ourvec{v}_i[j])_j$ characters, or it does not contain enough $2_j$ characters to be a superstring of $\ACoord(i,j)$.
    \end{claimproof}

    The following claim immediately follows from \cref{clm:no-superstring} since at least one phrase starts in every $\ACoord$ gadget and an additional phrase for $\#$.

    \begin{claim}\label{clm:atleastd1}
        $\Phrases(S(\ourvec{u}),i)\ge d+1$
    \end{claim}

The following two claims complete the proof of the lemma.
\begin{claim}
        If $\ourvec{u}\cdot \ourvec{v}_i=0$ then $\Phrases(S(\ourvec{u}),i)=d+1$.
    \end{claim}
    \begin{claimproof}
        Let $P_j=\ACoord(i,j)$ and $P_{d+1}=\#$.
        Notice that $\AVect(i,j)\#=\bigodot_{j=1}^{d+1}P_j$.
        We will show that for every $j\in[1..d]$ the string $P_j$ occurs in $\Dict\cdot\BVect(\ourvec{u})$ (which is a prefix of $S(\ourvec{u})$).
        Let $j\in [1..d]$.
        If $\ourvec{v}_i[j]=0$, then $P_j=\ACoord(i,j)=0_j{}^{3\sqrt n+i_1}2_j{}^{i_2}$ is a substring of $\skipo(j)=0_j{}^{4\sqrt n}2_j{}^{\sqrt n}$.
        Otherwise, if $\ourvec{v}_i[j]=1$ then it must be that $\ourvec{u}[j]=0$ (since $\ourvec{u}\cdot \ourvec{v}_i=0$) and therefore $P_j=\ACoord(i,j)=1_j{}^{3\sqrt n+i_1}2_j{}^{i_2}$ is a substring of $\BCoord(\ourvec{u},j)=1_j{}^{4\sqrt n+1}2_j{}^{\sqrt n}$.
        We have shown that $P_1,P_2,\dots,P_{d+1}$ is an $\LZss$-like factorization of $\AVect(i)$, thus by \cref{lem:LZ_optimal} we have $\Phrases(S(\ourvec{u}),i)\le d+1$.
        Combining with \cref{clm:atleastd1} we have $\Phrases(S(\ourvec{u}),i)= d+1$, as required.
    \end{claimproof}
    \begin{claim}
        If $\ourvec{u}\cdot \ourvec{v}_i=1$ then $\Phrases(S(\ourvec{u}),i)=d+2$.
    \end{claim}
    \begin{claimproof}
        Assume to the contrary that $\Phrases(S(\ourvec{u}),i)\ne d+2$.
        By \cref{clm:atmostd2,clm:atleastd1} it must be that $\Phrases(S(\ourvec{u}),i)=d+1$.
        Clearly, the last phrase is exactly $\#$.
        By \cref{clm:no-superstring} it must be that the remaining $d$ phrases are exactly $\ACoord(i,j)$ for $j\in[1..d]$.
        Let $j\in [1..d]$ be a coordinate such that $\ourvec{u}[j]=\ourvec{v}_i[j]=1$ (such a $j$ must exist due to $\ourvec{u}\cdot \ourvec{v}_i=1$).
        We will show that $\ACoord(i,j)=1_j{}^{3\sqrt n+i_1}2_j{}^{i_2}$ does not occur in $S(\ourvec{u})$ before $\AVect(i)\#$.
        Consider all occurrences of $X=1_j{}^{3\sqrt n}$ is $S(\ourvec{u})$ before $\AVect(i)$.
        Clearly, $X$ does not occur in $\Dict$.
        The only occurrence of $1_j$ in $\BVect(\ourvec{u})$ is in $\BCoord(\ourvec{u},j)$, which in this case is $\BCoord(\ourvec{u},j)=1_j{}^{2\sqrt n}{\#}1_j{}^{2\sqrt n}2_j{}^{\sqrt n}$ which does not contain an occurrence of $X$.

        Finally, for every $k<i$ the gadget $\ACoord(i,j)$ is not a substring of $\AVect(k)$.
        This is because the only substring of $\AVect(k)$ that may contain  $1_j$ is $\ACoord(k,j)$.
        Since $k<i$ it must be that either $k_1<i_1$ or $k_2<i_2$.
        Therefore, $\ACoord(k,j)$ either does not contain enough $1_j$ characters, or it does not contain enough $2_j$ characters to be a superstring of $\ACoord(i,j)$.

        To conclude $\ACoord(i,j)$ cannot be an $\LZss$ phrase, a contradiction.
    \end{claimproof}
    Thus, the lemma is proven.
\end{proof}

For a vector $\ourvec{u}$, we denote $S'(\ourvec{u})=\Dict\cdot\BVect(\ourvec{u})$.
That is, $S'(\ourvec{u})$ is the prefix of $S(\ourvec{u})$ obtained by removing the suffix $\Mat(A)$.
\begin{lemma}\label{lem:allPhrases}
    $|\LZss(S(\ourvec{u}))|-|\LZss(S'(\ourvec{u}))|=(d+2)n$ if and only if $u$ is not orthogonal to any $v\in A$.
\end{lemma}
\begin{proof}
    By definition of $\Phrases$ and by the $\#$ separators we have \[|\LZss(S(\ourvec{u}))|-|\LZss(\Dict\cdot\BVect(\ourvec{u}))|=\sum_{i=1}^n\Phrases(S(\ourvec{u}),i).\]
    The lemma follows immediately from \cref{lem:phrases}.
\end{proof}

We are ready to prove \cref{thm:lb}.

\begin{proof}[Proof of \cref{thm:lb}]
Assume by contradiction the there is an data structure $D$ that solves \cref{prb:lb} with $O(N^c)$ preprocessing time and $O(N^{2/3-\eps})$ update time.
We will show how to exploit $D$ to preprocess a set $A$ of $n$ binary $d$ dimensional vectors in $O(n^{1.5c})$ such that given a set $B$ of $n$ vectors, we can solve $\OV$ on $A,B$ in $O(\mathsf{poly}(d)\cdot n^{2-1.5\eps})$.
Thus, by \cref{lem:OVA} the theorem follows.

Given the set $A$ we construct two strings: $S=S(\ourvec{1})=\Dict\cdot \BVect(\ourvec{1})\cdot \Mat(A)$ and $S'=\Dict\cdot \BVect(\ourvec{1})$, and initialize two instances of $D$, denoted by $D_S$ and $D_{S'}$.
When receiving a set $B$ of vectors we iterate over all $\ourvec{u}\in B$.
For every $\ourvec{u}\in B$ we apply at most $d$ updates to transform $\BVect(\ourvec{1})$ to $\BVect(\ourvec{u})$ both in $D_S$ and $D_{S'}$ (notice that to transform $\BCoord(\ourvec 1,j)$ to $\BCoord(\ourvec{u},j)$ a single substitution is sufficient).
After all updates are applied we have $S=S(\ourvec{u})$ and $S'=S'(\ourvec{u})$.
At this point $D_S$ and $D_{S'}$ report $|\LZss(S(\ourvec{u}))|$ and $|\LZss(S'(\ourvec{u}))|$, respectively.
Thus, the algorithm checks if $|\LZss(S(\ourvec{u}))|-|\LZss(S'(\ourvec{u}))|<(d+2)\cdot n$, and if so reports that $A$ and $B$ contain a pair of orthogonal vectors.
Otherwise, the algorithm updates $S$ and $S'$ to be $S(\ourvec{1})$ and $S'(\ourvec{1})$, undoing all the substitutions.
If when finishing the scan of $B$ it was not reported that $A$ and $B$ contain a pair of orthogonal vectors, the algorithm reports that no such pair exists.
The correctness of the algorithm follows immediately from \cref{lem:allPhrases}.

\para{Complexity.}
Notice that the lengths of the strings are
$|\Dict|=O(\sqrt n\cdot d\cdot \sqrt n)=O(dn)$, $|\BVect(\ourvec{1})|=O(d\sqrt n)$ and
$|\Mat(A)|=O(n\cdot d \cdot \sqrt n)=O(d\cdot n^{1.5})$.
Therefore $|S|,|S'|\in O(d\cdot n^{1.5})$.
It immediately follows that the preprocess of $A$ takes $O(|S|^c+|S'|^c)=O(d^cn^{1.5c})$.
Moreover, the algorithm processes every vector $\ourvec{u}\in B$ with $O(d)$ updates, summing up to $O(d\cdot n)$ updates which take in total $O(d\cdot n\cdot (d\cdot n^{1.5})^{2/3-\eps})=O(d^{5/3-\eps}n^{2-1.5\eps})$.
By taking $d\in \Theta(2^{\sqrt{\log n}})$ and assuming SETH, this contradicts \cref{lem:OVA}.
\end{proof}

The above construction also establish the lower bound of \cref{thm:weak}.
Clearly, the following lemma implies the lower bound part of \cref{thm:weak}.
\begin{lemma}\label{thm:lbweak}
For every $c,\eps>0$, there is no data structure that solves \cref{prb:lb} with $O(N^c)$ preprocessing time and $O(|\LZss(S)|^{1-\eps})$ update time, unless SETH is false.
\end{lemma}
\begin{proof}
Assume by contradiction the there is a data structure $D$ that solves \cref{prb:lb} with $O(N^c)$ preprocessing time and $O(|\LZss(S)|^{1-\eps})$ update time.
The majority of the proof follows the same structure as the proof of \cref{thm:lb}.
Here, we rewrite only the complexity part of the proof of \cref{thm:lb}.

\para{Complexity.}
Since $|\Dict|,|\BVect(\ourvec{1})|\in O(dn)$, it holds that $|\LZss(\Dict\cdot\BVect(\ourvec{1}))|\in O(dn)$.
By \cref{lem:phrases}, there are $O(dn)$ phrases in the $\Mat$ part of $S(\ourvec{u})$ for every vector $\ourvec{u}$.
Therefore $|\LZss(S)|,|\LZss(S')|\in O(d\cdot n)$.
Moreover, the algorithm processes every vector $\ourvec{u}\in B$ with $O(d)$ updates, summing up to $O(d\cdot n)$ updates which take in total $O(dn\cdot (dn)^{1-\eps})=O(d^{2-\eps}n^{2-\eps})$.
By taking $d\in \Theta(2^{\sqrt{\log n}})$ and assuming SETH, this contradicts \cref{lem:OVA}.
\end{proof}

\appendix

\section{Implementation of a Dynamic $\LPF$-tree using Top Trees}\label{sec:top_trees}

In this section, we describe how to implement the data structure presented in \cref{sec:LPFandLPFTree} with $\Otild(1)$ update and query times.
To do that, we reproduce the proof from \cite{BCR24} with some minor changes.
Our minor changes are required to obtain a worst-case running time (compared to the amortized running time obtained by \cite{BCR24}), and to support $\findIAncestor$ queries.

We replace link-cut trees with top trees \cite{AHLT05}, as they support the same operations as link-cut trees, but top trees guarantee $\Otild(1)$ worst-case time complexity.
The operations supported by top trees are as follows.
\begin{enumerate}
    \item insert a single node to the forest.
    \item remove a leaf from a tree.
    \item attach a root to another node as a child.
    \item detach a node from its parent.
    \item add a weight $x$ to a node and all of its descendants.
    \item get the weight of a given node.
    \item get the $i$th ancestor of a node.
\end{enumerate}

As in \cite{BCR24}, we use Red-Black (RB) trees, that are balanced search trees supporting the following operations.
\begin{enumerate}
    \item Insert an element to an RB tree.
    \item Delete an element from an RB tree.
    \item Join two RB trees $T_1,T_2$ into a single RB tree containing the values of both, under the promise that the maximal value in $T_1$ is smaller than the minimal value in $T_2$.
    \item Given a value $c$ and an RB tree $T$, split $T$ into two RB trees $T_1$ and $T_2$ such that $T_1$ contains all values smaller than $c$ and the rest of the values are in $T_2$.
\end{enumerate}
All these operations cost $\Otild(1)$ time.

The following observation is useful.

\begin{observation}[{\cite[{Observation 23}]{BCR24}}]\label{obs:23}
    A collection of RB trees on $n$ nodes can be simulated using top trees.
The cost of every operation on an RB tree is then $O(\log^2 n)$.
\end{observation}

The indices of the string $S$ are stored as nodes of one top tree $U$ as follows.
For every node $i$ of the $\LPF$-tree, the children of $i$ are stored in an RB tree $T_i$, where the root of the RB tree is linked to the node $i$ (i.e., $i$ is the parent of the root of $T_i$ in $U$).
The edges of $U$ between nodes $i$ and the root of $T_i$ are called $\LPF$-edges, and all other edges of $U$ (i.e., edges inside $T_i$ for some $i$) are called internal edges.
We consider internal edges of $U$ as edges of weight 0, and $\LPF$-edges of $U$ as edges of weight 1.
These weights preserve the property that the distance between $i$ and every child of $i$ in the $\LPF$-tree is $1$ in $U$ (as it is in the $\LPF$-tree).
We preserve the invariant that the weight of every node $v$ in $U$ is exactly the sum of the (imaginary) edge weights from the root to $v$, and hence it is also the depth of $v$ in the $\LPF$ tree.

$\Insert$, $\Link$, and $\Delete$ are simple top tree operations, and cost $\Otild(1)$ time by \cref{obs:23}.
In the $\MoveInterval$ operation, we need to move consecutive nodes (representing consecutive indices of $S$) from some RB tree $T_i$ tree to another RB tree $T_j$.
This is done in $\Otild(1)$ time by \cref{obs:23} (using split and join operations on both $T_i$ and $T_j$).
In addition, the weight of these nodes should be updated.
This is done in $O(1)$ calls to get-weight and add-weight.
Therefore, $\MoveInterval$ costs $\Otild(1)$ time.
Moreover, $\GetDepth$ is implemented in $\Otild(1)$ by returning the weight of the node in $U$.
Finally, $\findIAncestor$ is implemented with a simple binary-search on the ancestors of $v$ in $U$, looking for the lowest ancestor with weight $\GetDepth(v)-i$.

\section{Proof of \texorpdfstring{\cref{lem:subPMDS}}{Lemma 9}}\label{sec:DynSPM}

In this section, we provide the following data structure.
\subPMDS*

For a string $S[1..n]$, the suffix array of $S$ is an array $SA_S$ of length $n$ that contains all the suffixes of $S$ in their lexicographic order.
The suffix $S[i..n]$ is represented in $SA_S$ by its starting index $i$.

We make use of the following data structure.
\begin{lemma}[{cf. \cite[{Theorem 10.16}]{kempa2022dynamic}}]\label{lem:DymSA}
There is a data structure for maintaining a dynamic string $S$ that undergoes edit operations and supports the following query in $\Otild(1)$ time:
Given an index $i$, return $SA_S[i]$.
The update time of the data structure is $\Otild(1)$ and the construction time for an initial string of length $n$ is $\Otild(n)$.
\end{lemma}

It is well-known that the suffix array can be used to search a pattern in a text.
We bring the following observation to describe the exact framework in which we use suffix arrays for text searching.
\begin{observation}\label{obs:SAandLCPisPM}
There is an algorithm that decides if a pattern $P$ occurs in a text $T$ of length $n$ using $O(\log n)$ queries of the form $\LCP(P, T[i..n])$ for some $i\in [|T|]$, $O(\log n)$ queries to $SA_T$, and $O(\log n)$ additional time.
\end{observation}
\begin{proof}
    One can binary search for an occurrence of $P$ in $T$.
    At every point, the algorithm checks if there is a starting index of an occurrence of $P$ in $SA_T[i..j]$ for some $i$ and $j$ initially set as $i=1$ and $j=|T|$.
    If $i = j$, the algorithm simply checks if $x=SA_T[i]$ is a starting index of an occurrence of $P$ by querying $\LCP(P,T[x..n])$.
    Otherwise, the algorithm sets $c = \left \lceil{\frac{j-i}{2}}\right \rceil $, obtains $x =SA_T[c]$ and queries for $\ell = \LCP(P,T[x..n])$.
    If $\ell = |P|$, we have found an occurrence of $P$ in $T$.
    Otherwise, the lexicographical order between $P$ and $T[x..n]$ is decided by the order between $P[\ell+1]$ and $T[x+\ell+1]$.
    We proceed our search in the half of $SA_T$ that may contain an occurrence of $P$.
    Clearly, the binary search terminates in $O(\log|T|)$ steps.
    In each step, the algorithm executes a single $\LCP$ query and a single query to $SA_T$.
\end{proof}

We are ready to describe our data structure for \cref{lem:subPMDS}.
We describe a data structure with amortized update time.
The running time can be de-amortized using standard techniques.
We maintain the dynamic data structure of \cref{lem:dympillar} on $S$.
Additionally, for every $i\in [\ceil*{\log |S|}+1]$, we are interested in maintaining a partition $\mathcal{P}^i$ of the indices of $[|S|]$ into intervals $I^i_1,I^i_2, \ldots ,I^i_{|\mathcal{P}^i|}$ with the following properties.
\begin{enumerate}
    \item For every $j\in [|\mathcal{P}^i|]$, $|I^i_j| < 2^{i+1}$
    \item For every $j\in [|\mathcal{P}^i|-1]$, $|I^i_{j}| + |I^i_{j+1}| \ge 2^i$
\end{enumerate}
Every interval $I^i_j = [a^i_j..b^i_j]$ corresponds to a substring $S^i_j= S[a^i_j..b^i_j]$.
For every $S^i_j$ the algorithm stores the structure of \cref{lem:DymSA} allowing queries for $SA_{S^i_j}$.

We maintain the intervals as follows.
Initially, $\mathcal{P}^i$ is simply a partition of $[1..|S|]$ into disjoint intervals of length exactly $2^i$ (excluding the last interval that may be shorter).
When an insertion (resp.\ deletion) is applied at index $x$ in $S$, the algorithm increases (resp.\ reduces) the ending index of $I^i_j$ that contains $x$ by $1$ (and implicitly shift all indices following $I^i_j$).
The algorithm also applies the appropriate insertion (resp.\ deletion) to the dynamic data structure of $S^i_j$.

If the update results in the length of $I^i_j$ reaching $2^{i+1}$, the algorithm splits $I^i_j$ into two intervals $I^i_j$ and $I^i_{j+1}$ of length $2^i$ each (thus shifting all intervals following $I^i_j$ prior to the update).
Then, the algorithm constructs the data structure of \cref{lem:DymSA} for $S^i_{j}$ and for $S^{i}_{j+1}$ from scratch.
Similarly, if the update results in two adjacent intervals $I^i_j$ and $I^i_{j+1}$ with combined length that exactly $2^{i}-1$, the algorithm merges $I^i_j$ and $I^i_{j+1}$ into a single interval $I^i_j$ of length $2^i-1$, and computes $S^i_j$ for the new $I^i_j$ from scratch.

It should be clear that all invariants are maintained and that for every interval $I^i_j$ the data structure of \cref{lem:DymSA} is maintained for the corresponding $S^i_j$.

Finding the interval $I^i_j$ can be implemented in $O(\log n)$ time by storing the dynamic endpoints of the intervals in a balanced search tree.
Then, applying a constant number of updates for the data structure of \cref{lem:DymSA} takes $\Otild(1)$ time.
The only case when we apply more then a single operation to the data structure of \cref{lem:DymSA} is when we merge or split intervals.

The newly created intervals, either merged or splitted, are always of size $O(2^i)$, so initializing the data structure of \cref{lem:DymSA} can be implemented by initializing an empty data structure and applying $O(2^i)$ insertions, which requires $\Otild(2^i)$ time.
We make a standard charging argument to account for this running time.

When an interval $I^i_j$ is split, its length is $2^{i+1}$, meaning that at least $2^i-1$ insertion operations were applied to the interval since its creation.
We can charge the $\Otild(2^i)$ cost of computing $S^i_j$ and $S^i_{j+1}$ from scratch on these $2^i$ edit operations for an amortized cost of $\Otild(1)$ per update.
Similarly, when two intervals $I^i_j$,$I^i_{j+1}$ are merged, one of $I^i_j$,$I^i_{j+1}$ must be of length at most $2^{i-1} - 1$.
It follows that at least $2^{i-1} + 1$ deletion operations were applied to the shorter of the two merged intervals since its creation.
Again, we can charge the $\Otild(2^i)$ cost of constructing the data structure of \cref{lem:DymSA} for the merged interval on these deletions, resulting in $\Otild(1)$ amortized running time per update.

The only part of the algorithm that is not worst-case but amortized is the construction of the suffix array data structure for new intervals.
We can de amortize this part of the algorithm by employing the standard technique of
gradually constructing the data structure for future intervals in advance.

We show how to implement this approach to de-amortize the running time for splitting an interval.
A similar de-amortization can be applied to de-amortize the running time for merging intervals.

Every interval $I^i_j$ stores, in addition to the data structure of $S^i_j=S[a^i_j..b^i_j]$, additional two suffix array data structures $SA^i_j(1)$ and $SA^i_j(2)$, for the strings $S^i_j(1) = S[a^i_j..c]$ and $S^i_j(2) = S[c+1..b^i_j]$ for $c = \floor*{\frac{a+b}{2}}$, respectively.
To be more precise, $SA^i_j(1)$ and $SA^i_j(2)$ are \textit{under construction}, and we wish to have the property that if $I^i_j$ is split, $SA^i_j(1)$ and $SA^i_j(2)$ are the proper suffix array data structures of the newly created intervals.
We also want $I^i_j$ to store a copy of the string $S^i_j$ in its state when the interval $I^i_j$ was created.

When $I^i_j$ is created, both $SA^i_j(1)$ and $SA^i_j(2)$ data structures are initialized as suffix array data structures for the empty string in $\Otild(1)$ time.
The interval $I^i_j$ also initializes an empty list $U^i_j$ meant for storing all future insertions, substitutions, and deletions applied to $I^i_j$.

Denote $S^i_j(1,init)$ (resp.\ $S^i_j(2,init)$) as the content of $S^i_j(1)$ (resp $S^i_j(2)$) at the time of the creation of the interval $I^i_j$.
Let us consider the following implementation of splitting an interval in the amortized data structure:
When the interval $I^i_j$ is split, the algorithm initializes the suffix array data structure for $SA^i_j(1)$ as a data structure for the empty string, inserts all symbols of $S^i_j(1,init)$ one by one and then applies all updates that $S^i_j(1)$ received during the lifetime of $I^i_j$ (a similar procedure is applied to obtain $SA^i_j(2)$).
Clearly, this correctly constructs $SA^i_j(1)$ and $SA^i_j(2)$.
This implementation requires $1+ 2^{i+1}+|U^i_j|$ operations of the dynamic suffix array data structure of \cref{lem:DymSA}.
We call this sequence of operations the \textit{construction sequence} of $I^i_j$.

Now, every time that an operation is applied to $I^i_j$, if $SA^i_j(1)$ and $SA^i_j(2)$ are already fully constructed for the current left and right halves of $S^i_j$, the algorithm applies the appropriate update to the one containing the updated index.

Otherwise, $SA^i_j(1)$ and $SA^i_j(2)$ are still under construction.
The algorithm applies the next $8$ operations from the construction sequence of $I^i_j$ to $SA^i_j(1)$ and to $SA^i_j(2)$.
Note that these operations are known, as we have access to $S^i_j(1,init)$ and to $S^i_j(2,init)$ (which is sufficient for applying the first $2^{i+1}$ operations) and to $U^i_j$.
Recall that at least $2^{i-1}$ updates are applied to $I^i_j$ before it is split.
It is therefore guaranteed that by the time $I^i_j$ splits, the algorithm has already applied $8 |U^i_j| = 4|U^i_j| + 4|U^i_j| \ge 2^{i+2} + |U^i_j| + 1$ operations from the construction sequence.
Therefore, the construction sequence has already been executed, and $SA^i_j(1)$ and $SA^i_j(2)$ are already constructed as required.
We also need to have an already prepared, copy of the current $S^i_j(1)$ and $S^i_j(2)$, which could take $O(2^{i+1})$ time to write and store if done at the time of the split.
This can be de-amortized in the same manner.

Recall that the data structure maintains dynamic partition of intervals for every $i\in [\ceil*{\log n}+1]$.
Each partition is stored as a balanced search tree with delta representation \cite{PlanarBook}, storing the starting indices of the intervals.
We make the following observation.
\begin{lemma}\label{clm:intervalpartition}
    For every $a\le b \in [|S|]$, there is a set $\mathcal{I}_{a,b}$ of $O(\log n)$ intervals in $\cup_{i=1}^{\ceil*{\log n} + 1}\mathcal{P}^i$ such that $\cup _{I \in \mathcal{I}_{a,b}} I = [a..b]$.
    The set $\mathcal{I}_{a,b}$ can be found in $\Otild(1)$ time given $a$ and $b$.
\end{lemma}
\begin{proof}
    We describe a greedy algorithm for finding $\mathcal{I}_{a,b}$.
    The algorithm initializes an empty list $L$ and an index $k=a$ and runs the following procedure until a termination condition is met.
    \begin{enumerate}
        \item \label{intpartition:1} If $k>b$, return $L$ as $\mathcal{I}_{a,b}$.
        \item \label{intpartition:2} Find the maximal $i$ such that the interval $I^i_j$ of $\mathcal{P}^i$ containing $k$ has $I^i_j \subseteq [a..b]$.
        \item \label{intpartition:3} Append $I^i_j$ to $L$ and set $k \leftarrow b^i_j+1$.
    \end{enumerate}

    First, we claim that $I^i_{j}$ is well defined for every $k\in [a..b]$.
    This follows from the fact that all intervals in $\mathcal{P}^0$ are of length exactly 1, therefore the interval $I^0_k=[k..k]$ is a feasible candidate for $I^i_j$.
    It follows that the algorithm terminates.
    It should be clear that at the beginning of every iteration of the algorithm, the list $L$ already contains a set of intervals covering exactly $[a..k-1]$, so when the algorithm terminates the list $L$ indeed satisfies $\cup_{I\in L}I=[a..b]$.

    Every iteration of the algorithm is implemented in $\Otild(1)$ time by querying each of the search trees of $\mathcal{P}^i$ for the interval containing $k$.
    It is left to prove that $|L|\in O(\log n)$ - notice that this also leads to the running time of the algorithm being $\Otild(1)$.

    Let $k_t$ be the value of $k$ at the beginning of the $t$-th iteration.
    Similarly, let $i_t$ be the maximal $i$ value found in Step \ref{intpartition:2} of the algorithm in the $t$-th iteration, and $I_t=[a_t..b_t]$ be the interval added to $L$ in the $t$-th iteration.
    We claim that for every $i\in [\ceil*{\log n}+1]$, $i$ can appear in the sequence $i_1,i_2 ,\ldots ,i_{|L|}$ at most 16 times - the bound on $|L|$ directly follows from this claim.

    Let us fix some $i\in [\ceil*{\log n}+1]$.
    Assume to the contrary that there exist 17 iterations, $t_1<t_2<\ldots < t_{17}$, of the algorithm in which $i_t = i$.
    Notice that for every $x< y\in [17]$ we have $I_{t_x}\neq I_{t_y}$.
    It follows from $k_{t_y} \ge k_{t_{x+1}}> b_{t_x}$, i.e., $I_{t_x}$ does not contain $k_{t_y}$.
    It immediately follows that the claim is correct for $i=\ceil*{\log n} + 1$ as $|\mathcal{P}^{\ceil*{\log n} + 1}|\le 2$.

    Let us assume that $i\in [\ceil*{\log n}]$.
    Recall that the length of any two consecutive intervals in $\mathcal{P}^i$ is least $2^i$.
    Therefore, $k_{t_9} > k_{t_1} + 2^{i+2}$ and $k_{t_{17}} > k_{t_9} + 2^{i+2}$.
    Implying $k_{t_9}>a+2^{i+1}$ and $k_{t_9}<b-2^{i+1}$.
    Let $I^{i+1}_j =[a'..b'] \in \mathcal{P}^{i+1}$  be the interval in $\mathcal{P}^{i+1}$  that contains $k_{t_9}$.
    Recall that the length of $I^{i+1}_j$ is less than $2^{i+2}$.
    It follows from $k_{t_9}>a+2^{i+1}$ that $a'>a$ and it follows from $k_{t_9}<b-2^{i+1}$ that $b'<b$.
    Hence $I^{i+1}_j\subseteq[a..b]$, a contradiction to $i_{t_9}=i$.
    \end{proof}

We now show how to answer a query given the partitions, and the data structure \cref{lem:dympillar}.
First, we apply \cref{clm:intervalpartition} to obtain a set $\mathcal{I}_{i_T,j_T}$ of intervals from the partitions such that the union of the intervals is exactly $S[i_T..j_T]$.
For every $I^i_j \in \mathcal{I}_{i_T,j_T}$, we can access entries of the suffix array of $S^i_j = S[a^i_j..b^i_j]$ in $\Otild(1)$ time, and we can query the $\LCP$ between indices in $S^i_j$ and $i_P$ via the data structure of \cref{lem:dympillar} in $\Otild(1)$ time.
It follows from \cref{obs:SAandLCPisPM} that we can report if there is an occurrence of $P= S[i_P..j_P]$ in $S^i_j$ in $\Otild(1)$ time.
We do this for each of the $\Otild(1)$ intervals of $\mathcal{I}_{i_T,j_T}$ for a total running time of $\Otild(1)$ and report that there is an occurrence of $P$ in $T = S[i_T..j_T]$ if one of the applications of \cref{obs:SAandLCPisPM} finds an occurrence.

Notice that this is insufficient: there may be an occurrence of $P$ in $T$ that is not contained in any of the intervals of $\mathcal{I}_{i_T,j_T}$.
We observe that every such 'missed' occurrence must contain an endpoint of an interval in $\mathcal{I}_{i_T,j_T}$.

We complement our algorithm with the following procedure.
For every $x\in \cup_{[a..b]\in \mathcal{I}_{i_T,j_T}}\{a,b\}$ (i.e., each endpoint of an interval of $\mathcal{I}_{i_T,j_T}$) we find a representation of all occurrences of $P$ in $S[\max\{x-|P|,i_T\}..\min\{x+|P|,j_T\}]$ using an internal pattern matching query to the data structure of \cref{lem:dympillar}.
This adds another factor of $\Otild(1)$ to the time complexity.
It is easy to see that each occurrence of $P$ that contains an endpoint of an interval is found in this way, thus concluding the algorithm and proving \cref{lem:subPMDS}.

\section{Proofs of Technical Lemmas}\label{sec:missing_proofs}
For two integers $p,q$, we denote as $\mathrm{gcd}(p,q)$ the greatest common divisor of $p$ and $q$.
In the proofs presented in this section, we use the periodicity lemma of Fine and Wilf \cite{FiWi65}.
\begin{fact}[Periodicity Lemma \cite{FiWi65}]\label{lem:per}
    If a string $S$ has a period $p$ and a period $q$ such that $|S| \ge p+q - \mathrm{gcd}(p,q)$, it holds that $\mathrm{gcd}(p,q)$ is a period of $S$.
\end{fact}

\begin{proof}[Proof of \cref{fact:break_period}]
    We prove the first statement, the second statement can be proven in a symmetrical manner.
    Recall that $S$ is periodic, so $p \le |S|/2$.
    Assume to the contrary  that $Sx$ is periodic with some period $q\le (|S| +1 )/2$.
    Since $x\neq S[n-p+1]$, and therefore $x\ne S[n-p+1] = S[n- i\cdot p+1]$ for every positive integer $i$ such that $n  \ge i\cdot p$, it must hold that $p$ is not a divisor of $q$.
    This implies that $\mathrm{gcd}(p,q) < p$.
    Since $S$ is a substring of $Sx$, it holds $q$ is a period of $S$ as well.
    Since both $p$ and $q$ are periods of $S$, and $p+q- \mathrm{gcd}(p,q) \le |S|$ it follows from \cref{lem:per} that $\mathrm{gcd}(p,q) < p$ is also a period of $S$.
    This is a contradiction to $p$ being the minimal period of $S$.
\end{proof}

\begin{proof}[Proof of \cref{lem:ipmalllengths}]
    We simply maintain the data structure of \cref{lem:dympillar}.
    Upon a query for $T$ and $P$, for every integer $i$ we define $T_i = T[i_T + i\cdot |P|.. \min(i_T+i|P|+2|P|-1 ,j_T)]$.
    It can be easily verified that every occurrence of $P$ in $T$ is contained in $T_i$ for some $i\in [0..|T|/|P|)$.
    We query for $\IPM_S(P,T_i)$ for every $i\in [0..{|T|}/{|P|})$, this yields a set of $O(|T|/|P|)$ arithmetic progressions that together represent all occurrences of $P$ in $T$.
    To obtain $\Clusters_T(P)$, attempt to merge clusters of occurrences from adjacent $T_i$'s by checking the fist occurrence of the cluster in $T_i$ is at distance exactly $p$ from the last occurrence of $T_{i-1}$.
    This can be easily done in $O(|T|/|P|)$ time.
\end{proof}

\begin{proof}[Proof of \cref{lem:first_occ}]
    We describe an algorithm for finding the first occurrence of $P=[i_P..j_P]$ in $T=[i_T..j_T]$, an algorithm for finding the last occurrence can be constructed in a similar way.
    We use the data structure of \cref{lem:subPMDS}.
    First we check whether $P$ occurs in $T$.
    If not, return null.
    Otherwise, binary search for the minimal index $k\le j_T$ such that $T[i_T..k]$ contains an occurrence of $P$ and return $k-|P|+1$.
\end{proof}

\bibliography{ref.bib}

\end{document}